\documentclass[]{llncs}

\usepackage{graphicx}
\usepackage{tcolorbox}
\newtcolorbox{mybox}[2][]{title=#2,#1,colframe  = black!100,}

\newcommand*\circled[1]{\tikz[baseline=(char.base)]{
            \node[shape=circle,draw,inner sep=1pt,color=black!100,fill=black!100] (char) {\tiny\color{white} #1};}}

\newcommand{\bbox}[3][black]% [options] title, content
{   \par
    \begin{tikzpicture}
        \node[draw,minimum width=\textwidth,inner sep=2.5mm,align=center,#1] (tempnode) {#3};
        \node[right,#1,fill=white] at ($(tempnode.north)!0.9!(tempnode.north west)$) {#2};
    \end{tikzpicture}
    \par
}
\usepackage{xcolor}
\usepackage{graphics}
\usepackage[font=small,labelfont=bf]{caption}
\usepackage{subfigure,wrapfig}
\usepackage{soul}
\usepackage{fixme}
\sethlcolor{black}
\makeatletter
\newif\if@blind
\@blindtrue %use \@blindfalse on final version
\if@blind \sethlcolor{black}\else
   
\fi
\usepackage[pagewise]{lineno}

\usepackage{color, colortbl}

\usepackage{thmtools,thm-restate}
\usepackage{amsmath}
\usepackage{wrapfig}
\usepackage{amssymb}
\usepackage[draft]{hyperref}
\usepackage{enumitem}
\usepackage{tikz,pgfplots}
\usepackage{url}
\usepackage{graphicx}
\usepackage{graphics}
\usepackage{multicol}
\usepackage{amsfonts}
\usepackage{xspace}
\usepackage{booktabs}
\usepackage{pifont}
\usepackage{multirow}
\newcommand*\rot[1]{\rotatebox{#1}}
\usetikzlibrary{automata,petri,positioning,shapes,snakes,arrows,backgrounds}
\usetikzlibrary{shapes,snakes,arrows, automata,shadows,calc}

\newcounter{bulletscount}
\newcommand{\rcounter}[1]{\refstepcounter{bulletscount}\label{#1}}
\newcommand*\samethanks[1][\value{footnote}]{\footnotemark[#1]}

\definecolor{mygreen}{rgb}{0,0.5,0}

\usepackage[ruled,noend,linesnumbered]{algorithm2e}

\SetAlFnt{\small}
\SetAlCapFnt{\small}
\SetAlCapNameFnt{\small}
\SetAlCapHSkip{0pt}
\IncMargin{-\parindent}

\newcommand\bndefine{\mathrel{::=}}
\newcommand\bnexpr{\langle \mathit{exp} \rangle}
\newcommand\bnpred{\langle \mathit{pred} \rangle}
\newcommand\bnrel{\langle \mathit{rel} \rangle}
\newcommand\bnass{\langle \mathit{assert} \rangle}
\newcommand\bnev{\langle \mathit{set} \rangle}
\newcommand\bntermbase{\langle \mathit{b} \rangle}
\newcommand\bnterm{\langle \mathit{r} \rangle}
\newcommand\bnmcm{\langle \mathit{MCM} \rangle}

\newcommand\tsone{\mathit{TS}_1}
\newcommand\tstwo{\mathit{TS}_2}

\newcommand\msource{\mm_S}
\newcommand\mtarget{\mm_T}

\newcommand\redcross{\color{red}{\ding{55}}}
\newcommand\tick{\color{mygreen}{\ding{52}}}

\newcommand\tot[2]{\mathit{total}(#1,#2)}
\newcommand\equi[2]{\mathit{equiv}(#1,#2)}
\newcommand\acy[1]{\mathit{acyclic}(#1)}
\newcommand\cyc[1]{\mathit{cyclic}(#1)}
\newcommand\irref[1]{\mathit{irreflexive}(#1)}

\newcommand\dfloc{\phi_{\mathit{DF_{thrd}}}}

\newcommand\define{\mathrel{:=}}
\newcommand\imp{\Rightarrow}

\newcommand\threads{\ensuremath{{\cal T}}}

\newcommand\events{\mathbb{E}}
\newcommand\executed{\mathbb{X}}
\newcommand\executedi{\mathbb{I}}
\newcommand\memory{\mathbb{M}}

\newcommand\stores{\mathbb{W}}
\newcommand\loads{\mathbb{R}}

\newcommand\mm{\mathcal{M}}
\newcommand\state{\sigma}

\newcommand{\reach}[1]{\mathit{state}(#1)}
\newcommand{\consistent}[2]{\mathit{cons}_{#1}(#2)}
\newcommand{\execproj}[1]{\mathit{execproj}(#1)}

\newcommand\bench[2][]{%
\if\relax\detokenize{#1}\relax%
\textsc{#2}\else \textsc{#2}\hspace{0.3pt}{\footnotesize(}#1{\footnotesize)}\fi\@\xspace}

\newcommand\syncrel{\mathit{sync}}
\newcommand\lwsyncrel{\mathit{lwsync}}
\newcommand\isyncrel{\mathit{isync}}
\newcommand\mfencerel{\mathit{mfence}}
\newcommand\rfrel{\mathit{rf}}
\newcommand\rferel{\mathit{rfe}}
\newcommand\frrel{\mathit{fr}}
\newcommand\porel{\mathit{po}}
\newcommand\pporel{\mathit{ppo}}
\newcommand\corel{\mathit{co}}
\newcommand\adrel{\mathit{ad}}
\newcommand\ddrel{\mathit{dd}}
\newcommand\cdrel{\mathit{cd}}
\newcommand\sthdrel{\mathit{sthd}}
\newcommand\slocrel{\mathit{sloc}}

\newcommand\iirel{\mathit{ii}}
\newcommand\icrel{\mathit{ic}}
\newcommand\cirel{\mathit{ci}}
\newcommand\ccrel{\mathit{cc}}

\newcommand\var[3]{\texttt{#1}(#2, #3)}
\newcommand\rrel{\stackrel{r}{\rightarrow}}

\newcommand\newrow{\\[-1.5pt]}

\newcommand\herd  {\textsc{Herd7}\@\xspace}

\newcommand\offence  {\textsc{offence}\@\xspace}
\newcommand\dfence  {\textsc{DFence}\@\xspace}
\newcommand\fender  {\textsc{Fender}\@\xspace}
\newcommand\musketeer  {\textsc{musketeer}\@\xspace}

\newcommand\porthos  {\textsc{porthos}\@\xspace}

\newcommand\cedge[3]{\texttt{C}_#1(#2, #3)}
\newcommand\cvar[1]{\texttt{C}(#1)}

\newcommand\exec{X}

\newcommand\rf{\stackrel{\rfrel}{\rightarrow }}
\newcommand\fr{\stackrel{\frrel}{\rightarrow}}
\newcommand\uniproc  {uniproc}

\newcommand\poloc{\stackrel{{\porel} \cap {\slocrel}}{\longrightarrow}}
\newcommand\co{\stackrel{\corel}{\rightarrow}}
\newcommand\mem{$\sigma$}

\newcommand\A{E_1}
\newcommand\B{E_2}
\newcommand\C{E_3}

\newcommand\stdef[2]{\state [#1;#2]}
\newcommand\ass[2]{#1\leftarrow #2}

\newcommand{\Ppsi}{P_{\psi}}
\newcommand{\Pallpsi}{P_{\forall\psi}}
\newcommand{\Pnp}{P_{\mathit{np}}}

\newcommand{\prog}{\mathit{prog}}
\newcommand{\instr}{\mathit{inst}}
\newcommand{\tid}{\mathit{tid}}
\newcommand{\loc}{\mathit{loc}}
\newcommand{\reg}{\mathit{reg}}
\newcommand\rel{\stackrel{\mathit{rel}}{\rightarrow}}
\newcommand{\com}{\mathit{atom}}

\newcommand{\mfence}{\mathit{mfence}}
\newcommand{\hfence}{\mathit{sync}}
\newcommand{\lfence}{\mathit{lwsync}}
\newcommand{\cfence}{\mathit{isync}}
\newcommand{\thrd}{\mathit{thrd}}
\newcommand{\name}{\langle\mathit{name}\rangle}
\newcommand{\id}{\mathit{id}}
\newcommand{\iriw}{{\bf IRIW}}

\newcommand{\Set}[2]{\{#1\mid #2\}}

\newcommand{\execsof}[1]{\mathit{exec}(#1)}

\newcommand\hlcom{\mathit{hlinst}}

\usepackage{amsmath}

\begin{document}

%\title{Portability Analysis for Axiomatic Memory Models \\ \porthos: {\bf One} \emph{Tool} {\bf for all} \emph{Models}}
\title{Portability Analysis for Weak Memory Models \\ \mdseries\porthos: {\bf One} \emph{Tool} {\bf for all} \emph{Models}}

%\author{Hern\'an Ponce-de-Le\'on\inst{1} \and Florian Furbach\inst{2} \and Keijo Heljanko\inst{3} \and Roland Meyer\inst{2}}
\author{{\small Hern\'an Ponce-de-Le\'on\inst{1}\thanks{This work was carried out when the author was at Aalto University.} \and Florian Furbach\inst{2} \and Keijo Heljanko\inst{3} \and Roland Meyer\inst{4}\samethanks}}

\institute{ 
$^1$fortiss GmbH, Germany \
$^2$TU Kaiserslautern, Germany \
$^3$Aalto University and HIIT, Finland \
$^4$TU Braunschweig, Germany \\
\email{ponce@fortiss.org},
\email{furbach@cs.uni-kl.de},
\email{keijo.heljanko@aalto.fi}, 
\email{roland.meyer@tu-braunschweig.de}
}

\maketitle
%!TEX root = sas2017.tex
% !TEX spellcheck = en-EN

\vspace{-5mm}
\begin{abstract}
We present \porthos, the first tool that discovers porting bugs in performance-critical code. 
\porthos\ takes as input a program and the memory models of the source architecture for which the program has been developed and the target model to which it is ported.
If the code is not portable, \porthos\ finds a bug in the form of an unexpected execution --- an execution that is consistent with the target but inconsistent with the source memory model.
Technically, \porthos\ implements a bounded model checking method that 
reduces the portability analysis problem to \mbox{satisfiability}
modulo theories (SMT).
There are two main problems in the reduction that we present novel and efficient solutions for. 
First, the formulation of the portability problem contains a quantifier alternation (consistent + inconsistent). 
We introduce a formula that encodes both in a single existential query.
Second, the supported memory models (e.g., Power) contain recursive definitions.
We compute the required least fixed point semantics for recursion (a problem that was left open in \cite{memalloy}) efficiently in SMT.
Finally we present the first experimental analysis of portability from TSO to Power.
\end{abstract}

%\category{CR-number}{subcategory}{third-level}

%\terms
%term1, term2
%!TEX root = sas2017.tex
% !TeX spellcheck = en_GB

\section{Introduction}
Porting code from one architecture to another is a routine task in system development. 
Given that no functionality has to be added, porting is rarely considered interesting from a programming point of view. 
At the same time, porting is non-trivial as the hardware influences both the semantics and the compilation of the code in subtle ways.
The unfortunate combination of being routine and yet subtle makes porting prone to mistakes. 
This is particularly true for performance-critical code that interacts closely with the execution environment. 
Such code often has data races and thus exposes the programmer to the details of the underlying hardware. 
When the architecture is changed, the code may have to be adapted to the primitives of the target hardware.

We tackle the problem of porting performance-critical code among hardware architectures.
Our contribution is the new (and to the best of our knowledge first) tool \porthos\ to fight porting bugs.
It takes as input a piece of code, a model of the source architecture for which the code has been developed, and a model of the target architecture to which the code is to be ported.
\porthos\ automatically checks whether every behaviour of the code on the target architecture is also allowed on the source platform. 
This guarantees that correctness of the program in terms of safety properties (in particular properties like mutual exclusion) carry over to the targeted hardware, and the program remains correct after porting.

\begin{comment}
Semantics and verification under weak memory models have been the subject of study at least since 2007.
Initially, the behaviour of x86 and TSO has been clarified~\cite{Checkfence2007,SarkarSNORBMA09}, then the Power architecture has been addressed~\cite{SarkarSAMW11,Mador-HaimMSMAOAMSW12}, now ARM is being tackled~\cite{FlurGPSSMDS16}, and the study also looks beyond hardware, in particular C++11 received considerable attention~\cite{BattyOSSW11,BattyDW16}. 
Research in semantics goes hand in hand with the development of verification methods. 
They come in two flavours: program logics~\cite{VafeiadisN13,TuronVD14} and algorithmic approaches~\cite{Burckhardt2008,AtigBBM10,DanMVY13,AlglaveKT13,AlglaveMT14,BouajjaniDM13,DanMVY15,AbdullaAAJLS15,AbdullaAJL16}. 
Notably, each of these methods and tools is designed for a specific memory model and hence is not directly able to handle porting tasks.
\end{comment}

Portability requires an analysis method that is hardware-architecture-aware in the sense that a description of the memory models of source and target platforms has to be part of the input. 
A language for memory models, called CAT~\cite{cat}, has been developed only recently. 
In CAT, memory models are defined in terms of relations between memory operations of a program. 
There are some relations (program order, reads from, coherence) that are common to all memory models.
A memory model may define further so-called derived relations by restricting and composing base relations.
The memory model specifies axioms in the form of acyclicity and irreflexivity constraints over relations.
An execution is consistent if it satisfies all axioms. 
Our work builds on the CAT language.

There are three problems that make portability different from most common verification tasks.
\begin{enumerate}
\item[\emph{(i)}] We have to deal with user-defined memory models.
These models may define derived relations as least fixed points.
\item[\emph{(ii)}] The formulation of portability involves an alternation (consistent + inconsistent) of quantifiers.
\item[\emph{(iii)}] High-level code may be compiled into different low-level code depending on the architecture.
\end{enumerate}

Concerning the first problem, we implement in SMT the operations that CAT defines on relations. 
Notably, we propose an encoding for derived relations that are defined as least fixed points. 
Such least fixed points are prominently used in the Power memory model~\cite{AlglaveMT14} and their computation was identified as a key problem in \cite{memalloy}. To quote the authors~\emph{[...] the proper fixpoint construction [...]
is much more expensive than a fixed unrolling}. We show that, with our encoding, this is not the case. %To see the difficulty, 
A naive approach would implement the Kleene iteration in SAT by introducing copies of the variables for each iteration step, resulting in a very large encoding. 
We show how to employ SAT + integer difference logic~\cite{DBLP:conf/formats/CottonAMN04} to compactly encode the Kleene iteration process. Notably, every bounded model checking technique reasoning about complex memory models defined in CAT will face the problem of dealing with recursive definitions and can make use of our technique to solve it efficiently.

The second problem is to encode the quantifier alternation underlying the definition of portability. 
A porting bug is an execution that is consistent with the target but inconsistent with the source memory model.
We capture this alternation with a single existential query. 
Consistency is specified in terms of acyclicity (and irreflexivity) of relations.
Hence, an execution is inconsistent if a derived relation of the (source) memory model contains a cycle (or is not irreflexive). 
The naive idea would be to model cyclicity by unsatisfiability. 
Instead, we reduce cyclicity to satisfiability by introducing auxiliary variables that guess the cycle.

The reader may criticise our definition of portability: one could claim that all that matters is whether safety is preserved, even if the executions differ. 
To be precise, a state-based notion of portability requires that every state computable under the target architecture is already computable on the source platform.
We study state portability and come up with two results.

\begin{enumerate}
\item[\emph{(a)}] Algorithmically, state portability is beyond SAT.
\item[\emph{(b)}] Empirically, there is little difference between state portability and our notion.
\end{enumerate}

The third problem is that the same high-level program is compiled to different assembly programs depending on the source and the target architectures. %~\cite{MorissetN17}. 
Even the number of registers and the semantics of the synchronisation primitives provided by those architectures usually differ. 
%For example, the program from~\autoref{fig:iriw} may be compiled to different low-level code depending on the architecture~\cite{MorissetN17}.
Consider the program from~\autoref{fig:iriw}, written in C++11 and compiled to x86 and Power.
The observation is this. 
Even if the assembly programs differ, 
one can map every assembly memory access to the corresponding read or write operation in the high-level code. 
In the example, clearly ``$\texttt{MOV [y],\$1}$'' and ``$\texttt{stw r1,y}$'' correspond to ``$\texttt{y.store(memory\_order\_relaxed, 1)}$''. 
This allows us to relate low-level and high-level executions and to compare executions of both assembly programs by checking if they map to the same high-level execution.
With this observation, our analysis %is performed on a single low-level program, 
%however, it 
can be extended by translating an input program into two corresponding assembly programs and making explicit the relation among the low-level and high-level executions.  
While this relation among executions is not studied in the present paper, details of how to construct it and how to incorporate it into our approach are given in Appendix~\ref{sec:common}.

In summary, we make the following contributions.
\begin{enumerate}
\item We present the first SMT-based implementation of a core subset of CAT which can handle recursive definitions efficiently. %and can be integrated into bounded model checking techniques dealing with weak memory models (ours and others).
\item We formulate the portability problem based on the CAT language.
% for axiomatic memory models.
\item We develop a bounded analysis for portability. Despite the apparent alternation of quantifiers, our SMT encoding is a satisfiability query of polynomial size and optimal in the complexity sense.
\item We compare our notion of portability to a state-based notion and show that the latter does not afford a polynomial SAT encoding.
\item We present experiments showing that \emph{(i)} in a large majority of cases both notions of portability coincide, and
\emph{(ii)} mutual exclusion algorithms are often non portable, particularly we perform the first analysis from TSO to Power.
\end{enumerate}
% !TeX spellcheck = en_GB
%!TEX root = sas2017.tex
\section{Portability Analysis on an Example}
\label{sec:example}

Consider program \iriw\ in~\autoref{fig:iriw}, written in C++11 and using the atomic operator \texttt{memory\_order\_relaxed} which provides no guarantees on how memory accesses in different threads are ordered. 
When porting, the program is compiled to two different architectures. 
%Note that, while the assembly codes differ, they contain the same read and write accesses to the memory. \hp{The last sentence is repeated in the end of last section.}
The corresponding low-level programs behave differently on x86 and on IBM's Power. 
On TSO, the memory model implemented by x86, each thread has a store buffer of pending writes. A thread can see its own writes before they become visible to other threads (by reading them from its buffer), but once a write hits the memory it becomes visible to all other threads simultaneously: TSO is a multi-copy-atomic model~\cite{opacb1105256}. 
Power on the other hand does not guarantee that writes become visible to all threads at the same point in time. 
Think of each thread as having its own copy of the memory. 
With these two architectures in mind, consider the execution in~\autoref{fig:iriw}. 
Thread $t_2$ reads $x=1, y=0$ and thread $t_3$ reads $x=0, y=1$, indicated by the solid edges $\rferel$ and $\rfrel$. 
Since under TSO every execution has a unique global view of all operations, no interleaving allows both threads to read the above values of the variables. 
Under Power, this is possible. 
Our goal is to automatically detect such differences when porting a program from one architecture to another, here from TSO to Power.
\begin{figure}[t]
\centering
\[\arraycolsep=5pt
\begin{array}{lllllll}
\texttt{thread } t_0 &\texttt{thread } t_1 \\
\texttt{y.store(memory\_order\_relaxed, 1)} & \texttt{x.store(memory\_order\_relaxed, 1)} \\
\vspace{-3mm}
\\
\texttt{thread } t_2 & \texttt{thread } t_3\\
r_1 = \texttt{x.load(memory\_order\_relaxed)}; & r_1 = \texttt{y.load(memory\_order\_relaxed)}; \\
r_2 = \texttt{y.load(memory\_order\_relaxed)} & r_2 = \texttt{x.load(memory\_order\_relaxed)} & \\
\end{array}
\]
\[\arraycolsep=5pt
\begin{array}{|l|l|l|l|}
\multicolumn{4}{c}{\textsc{x86 Assembly}} \\
\hline
\texttt{  thread } t_0 &\texttt{  thread } t_1 & \texttt{  thread } t_2 &\texttt{  thread } t_3 \\
\hline
\texttt{MOV [y],\$1} & \texttt{MOV [x],\$1} & \texttt{MOV EAX,[x]} & \texttt{MOV EAX,[y]} \\
& & \texttt{MOV EAX,[y]} & \texttt{MOV EAX,[x]} \\
\hline
\end{array}
\]
\[\arraycolsep=5pt
\begin{array}{|l|l|l|l|}
\multicolumn{4}{c}{\textsc{Power Assembly}} \\
\hline
\texttt{ thread } t_0 &\texttt{ thread } t_1 & \texttt{ thread } t_2 &\texttt{ thread } t_3 \\
\hline
\texttt{li  r1,1} & \texttt{li  r1,1} & \texttt{lwz  r1,x} & \texttt{lwz  r1,y} \\
\texttt{stw r1,y}  & \texttt{stw r1,x} & \texttt{lwz r3,y} & \texttt{lwz r3,x} \\
\hline
\end{array}
\]
\begin{tikzpicture}

\node[] (a) at (0, 0) {$Rx1$};
\node[] (b) at (0, -1) {$Ry0$};
\node[] (c) at (1.5,0) {$Ry1$};
\node[] (d) at (1.5, -1) {$Rx0$};
\node[] (e) at (3, 0) {$Wx1$};
\node[] (ix) at (4.5, 0) {$Ix0$};
\node[] (f) at (-1.5,0) {$Wy1$};
\node[] (iy) at (-3, 0) {$Iy0$};

\draw[->, dashed, red, thick] (a) edge[right] node {\footnotesize $\porel$} (b);
\draw[->, bend left=10, red, thick] (b) edge[below] node {\footnotesize $\frrel$} (f);
\draw[->] (iy) edge[above] node {\footnotesize $\corel$} (f);
\draw[->, dashed, red, thick] (c) edge[left] node {\footnotesize $\porel$} (d);
\draw[->, bend right=10, red, thick] (d) edge[below] node {\footnotesize $\frrel$} (e);
\draw[->] (ix) edge[above] node {\footnotesize $\corel$} (e);

%\draw[->, dashed, bend left=10] (f) edge[above] node {\footnotesize $\rferel$} (b);
%\draw[->, dashed, bend right=10] (e) edge[above] node {\footnotesize $\rferel$} (d);

\draw[->, bend right=25] (iy) edge[below] node {\footnotesize $\rfrel$} (b);
\draw[->, bend left=25] (ix) edge[below] node {\footnotesize $\rfrel$} (d);

\draw[->, bend right=15, red, thick] (e) edge node[above right] {\footnotesize $\rferel$} (a);
\draw[->, bend left=15, red, thick] (f) edge node[above left] {\footnotesize $\rferel$} (c);
%\draw[->, dashed, bend right=30] (ix) edge node[above right] {\footnotesize $\rfrel$} (a);
%\draw[->, dashed, bend left=30] (iy) edge node[above left] {\footnotesize $\rfrel$} (c);

%\node[] (null) at (0, -1.5) {};
\end{tikzpicture}
\caption{Portability of program \iriw\ from TSO to Power.}
\label{fig:iriw}
\end{figure}
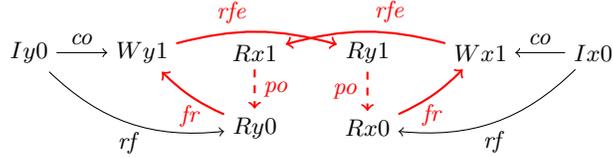

Our tool \porthos\ applies to various architectures, and 
we not only have a language for programs but also a \emph{language for memory models}. 
The semantics of a program on a memory model is defined axiomatically, following two steps~\cite{AlglaveMT14,memalloy}. 
We first associate with the program (and independent of the memory model) a set of executions which are candidates for the semantics. 
An execution is a graph (\autoref{fig:iriw}) whose nodes (events) are program instructions and whose edges are basic dependencies: the program order $\porel$, the reads-from relation $\rfrel$ (giving the write that a load reads from), and the coherence order $\corel$ (stating the order in which writes take effect). 
The memory model then defines which executions are consistent and thus form the semantics of the program on that model.

\begin{figure}[t]
\begin{minipage}{.68\linewidth}
\begin{mybox}{\emph{Consistent$_{TSO}$}}
\footnotesize\centering
\rcounter{b:uniproc}
$\circled{\thebulletscount}\ \acy {{({\porel} \cap {\slocrel})} \cup {\rfrel} \cup {\frrel} \cup {\corel}}$
\rcounter{b:tso}
$\circled{\thebulletscount}\ \acy {{\rferel} \cup {\corel} \cup {\frrel} \cup {({\porel} \setminus (\stores \times \loads))} \cup {\mathit{mfence}}}$
\end{mybox}
\end{minipage}
\hspace{.2cm}
\begin{minipage}{.28\linewidth }
\centering
\footnotesize\centering
${\frrel} \define {\rfrel^{-1};\corel}$ \hspace{3mm}
${\rferel} \define {\rfrel \setminus \sthdrel}$ \\
\end{minipage} 
\caption{TSO.}
\label{fig:sc}
\end{figure}

We describe memory models in the recently proposed language CAT~\cite{cat}. 
Besides the base relations, a model may define so-called derived relations. 
The consistency requirements are stated in terms of acyclicity and irreflexivity axioms over these (base and derived) relations. 
The CAT formalisation of TSO is given in~\autoref{fig:sc}. 
It forbids executions forming a cycle over ${\rferel} \cup {\frrel} \cup {({\porel} \setminus {(\stores \times \loads)})}$. 
The red edges in~\autoref{fig:iriw} yield such a cycle; the execution is not consistent with TSO. 
Power further relaxes the program order (\autoref{fig:power}), the dotted lines are no longer considered for cycles and %; the formalisation is given in~\autoref{fig:power}) 
thus the execution is consistent. 
Hence, {\bf IRIW} has executions consistent with Power but not with TSO and is hence not portable.

Our contribution is a bounded analysis for portability implemented in the \porthos\ tool 
(\url{http://github.com/hernanponcedeleon/PORTHOS}).
First, the program is unrolled up to a user-specified bound.
Within this bound, \porthos\ is guaranteed to find all portability bugs. 
It will neither see bugs beyond the bound nor will it be able to prove a cyclic program portable. 
The unrolled program, together with the CAT models, is transformed into an SMT formula where satisfying assignments correspond to bugs.

%Concerning \emph{(i)}, a 
A bug is an execution consistent with the target memory model $\mtarget$ but inconsistent with the source $\msource$. 
We express this combination of consistency and inconsistency with only one existential quantification. 
The key observation is that the derived relations, which may differ in $\mtarget$ and $\msource$, are fully defined by the execution. 
Hence, by guessing an execution we also obtain the derived relations (there is nothing more to guess). 
Checking consistency for $\mtarget$ is then an acyclicity (or irreflexivity) constraint on the derived relations that immediately yields an SMT query. Inconsistency for $\msource$ requires cyclicity. 
The trick is to explicitly guess the cycle. 
We introduce Boolean variables for every event and every edge that could be part of the cycle. 
In~\autoref{fig:iriw}, if $Rx1$ is on the cycle, indicated by the variable $\cvar{Rx1}$ being set, then there should be one incoming and one outgoing edge also in the cycle.
Besides the incoming edge shown in the graph, $Rx1$ could read from the initial value $Ix0$.
Since there are two possible incoming edges but only one outgoing edge, we obtain ${\cvar{Rx1}} \imp (({{\cedge{\rferel}{Wx1}{Rx1}} \lor {\cedge{\rfrel}{Ix0}{Rx1}}})  {} \land {{\cedge{\porel}{Rx1}{Ry0}}})$.
If a relation is on the cycle, then also both end-points should be part of the cycle and the relation should belong to the execution: ${\cedge{{\porel}}{Rx1}{Ry0}} \imp {({\cvar{Rx1}} \land {\cvar{Ry0}} \land {\var{po}{Rx1}{Ry0}})}.$
Finally, at least one event has to be part of the cycle: ${\cvar{Ix0}} \lor {\cvar{Wx1}} \lor {\cvar{Rx1}} \lor {\cvar{Rx0}} \lor {\cvar{Iy0}} \lor {\cvar{Wy1}} \lor {\cvar{Ry1}} \lor {\cvar{Ry0}}.$
The execution in~\autoref{fig:iriw} contains the relations marked in red and forms a cycle which violates Axiom~\circled{\ref{b:tso}} in TSO.
The assignment respects the axioms of Power (\autoref{fig:power}), showing the existence of a portability bug in \iriw\ from TSO to Power.

%Concerning \emph{(ii)}, the 
The other challenge is to capture relations that are defined recursively. %as the least solution to a system of equations. 
The Kleene iteration process \cite{stoltenberg-hansen_lindström_griffor_1994} starts with the empty relation and repeatedly adds pairs of events according to the recursive definitions.
We encode this into (quantifier-free) integer difference logic~\cite{DBLP:conf/formats/CottonAMN04}. 
For every recursive relation $\texttt{r}$ and every pair of events $(e_1, e_2)$, we introduce an integer variable $\Phi^\texttt{r}_{e_1,e_2}$ representing the iteration step in which the pair entered the value of $\texttt{r}$.
A Kleene iteration then corresponds to a total ordering on these integer variables. 
Crucially, we only have one Boolean variable $\texttt{r}(e_1, e_2)$ per pair rather than one per iteration step. We illustrate the encoding on a simplified version of the preserved program order for Power defined as $\pporel \define \iirel \cup \icrel$ (cf. ~\autoref{fig:power} for the full definition). 
The relation is derived from the mutually recursive relations $\iirel \define \ddrel \cup \icrel$ and $\icrel \define \cdrel \cup \iirel$, where $\ddrel$ and $\cdrel$ represent data and control dependencies. 
Call $Rx1$ and $Ry0$ respectively $e_1$ and $e_2$. 
The encoding is
\[\arraycolsep=2.5pt
\begin{array}{rcl}
{\var{ii}{e_1}{e_2}} & \Leftrightarrow & {(\var{dd}{e_1}{e_2} \land ({\Phi^\texttt{ii}_{e_1,e_2} > \Phi^\texttt{dd}_{e_1,e_2}}))} \lor {} 
 (\var{ic}{e_1}{e_2} \land ({\Phi^\texttt{ii}_{e_1,e_2} > \Phi^\texttt{ic}_{e_1,e_2}}))\\
 \\
{\var{ic}{e_1}{e_2}} & \Leftrightarrow & {(\var{cd}{e_1}{e_2} \land ({\Phi^\texttt{ic}_{e_1,e_2} > \Phi^\texttt{cd}_{e_1,e_2}}))} \lor {} 
 (\var{ii}{e_1}{e_2} \land ({\Phi^\texttt{ic}_{e_1,e_2} > \Phi^\texttt{ii}_{e_1,e_2}})). \\
\end{array}
\]
%\flo{remove second equation?}
%We explain the first equivalence, the second is similar. 
The pair $(e_1, e_2)$ that belongs to relation $\ddrel$ in step 
$\Phi^\texttt{dd}_{e_1,e_2}$ of the Kleene iteration can be added to relation 
$\iirel$ at a later step $\Phi^\texttt{ii}_{e_1,e_2}>\Phi^\texttt{dd}_{e_1,e_2}$.
As $\iirel \define \ddrel \cup \icrel$, the disjunction allows us to also add the elements of $\icrel$ to $\iirel$.
Since $\ddrel$ and $\cdrel$ are empty for \iriw, the relations $\iirel$ and $\icrel$ have to be identical. 
Identical non-empty relations will not yield a solution: the integer variables cannot satisfy $(\Phi^\texttt{ii}_{e_1,e_2} > \Phi^\texttt{ic}_{e_1,e_2})$ and $(\Phi^\texttt{ic}_{e_1,e_2} > \Phi^\texttt{ii}_{e_1,e_2})$ at the same time. 
Hence, the only satisfying assignment is the one where both $\iirel$ and $\icrel$ are the empty relation, which implies that $\pporel$ is empty.
This is consistent with the preserved program order of Power for \iriw.
%!TEX root = sas2017.tex
% !TeX spellcheck = en_GB

\section{Programs and Memory Models}
\label{sec:back}

We introduce our language for programs and the core of the language CAT. 
%To define the semantics of a program relative to a memory model, we proceed in two steps. 
%We first associate with the program (and independent of the memory model) a set of  executions that are candidates for the semantics. 
%The memory model is formalized in terms of axioms that rule out some of those executions. 
%The executions that are consistent with the axioms form the semantics of the program on the memory model.
%This axiomatic style of giving semantics is standard. 
The presentation follows~\cite{cat,memalloy} and we refer the reader to those works for details.
% -----------------------------------------------------------------------
% -----------------------------------------------------------------------
% -----------------------------------------------------------------------
% -----------------------------------------------------------------------

\paragraph{\bf Programs.}
%\subsection{Programs}
%\label{sec:prog}

Our language for shared memory concurrent programs is given in~\autoref{fig:grammar}. 
Programs consist of a finite number of threads from a while-language.  
The threads operate on assembly level, which means they explicitly read from the shared memory into registers, write from registers into memory, and support local computations on the registers. 
The language has various fence instructions (sync, lwsync, and isync on Power and mfence on x86) that enforce ordering and visibility constraints among instructions.
We refrain from explicitly defining the expressions and predicates used in assignments and conditionals. 
They will depend on the data domain. 
For our analysis, we only require the domain to admit an SMT encoding in a logic which has its satisfiability problem in NP.
%Also note that we make the locations that reads and writes operate on explicit. 
%This eases the presentation. 
%Our technique can be adapted to support address arithmetic.
For the rest of the paper we will assume that programs are acyclic: any while
statement is removed by unrolling the program to a depth specified by the user. 
Since verification is generally undecidable for
while-programs~\cite{Rice}, this under-approximation is necessary for cyclic programs.

\begin{figure}[t]
	\begin{minipage}[b]{0.4\textwidth}
		\begin{align*}
		\langle \prog \rangle \bndefine&\ \texttt{program}\ \langle \thrd \rangle^*\\
		\langle \thrd \rangle \bndefine &\ \texttt{thread}\ \langle \tid \rangle\ \langle \instr \rangle \\
		\langle \instr \rangle  \bndefine &\ \langle \com \rangle\mid \langle \instr \rangle ; \langle \instr \rangle\\
		& \mid  \texttt{while}\ \bnpred\ \langle \instr \rangle\\
		& \mid  \texttt{if}\ \bnpred\ \texttt{then}\  \langle \instr \rangle\ \\
		&\qquad \qquad \hspace{4mm} \texttt{else}\ \langle \instr \rangle\\
		\langle \com \rangle  \bndefine &\  
		\langle \reg \rangle \leftarrow \bnexpr \mid \langle \reg \rangle \leftarrow \langle \loc \rangle \\ 
		& \mid \langle \loc \rangle \define \langle \reg\rangle \mid \langle \mfence \rangle \\
		& \mid \langle \hfence \rangle \mid \langle \lfence \rangle \mid \langle \cfence \rangle
%		\langle \fence \rangle  \bndefine &\  \langle \hfence \rangle \mid \langle \lfence \rangle \mid \langle \cfence \rangle \mid \langle \mfence \rangle
		\end{align*}
		\caption{Programming language.}
		\label{fig:grammar}
	\end{minipage}
	\begin{minipage}[b]{0.5\textwidth}
		\begin{align*}
		\bnmcm  \bndefine &\  \bnass\mid \bnrel \mid \bnmcm \wedge \bnmcm\\
		\bnass  \bndefine &\ \acy \bnterm  \mid \irref \bnterm  \\ %\mid \empt \bnterm
		\bnterm \bndefine&\  \bntermbase \mid \bnterm \cup \bnterm \mid \bnterm \cap \bnterm \mid \bnterm \setminus \bnterm \\
		& \mid \bnterm^{-1}  \mid \bnterm^+ \mid \bnterm^* \mid \bnterm;\bnterm  \\
		\bntermbase \bndefine &\ \porel \mid \rfrel \mid \corel \mid \adrel \mid \ddrel \mid \cdrel \mid \sthdrel \mid \slocrel\\
		& \mid \mfencerel \mid \syncrel \mid \lwsyncrel \mid \isyncrel \\
		& \mid \id(\bnev) \mid \bnev \times \bnev\mid \name  \\
%		\bnev  \bndefine &\ \events \mid \memory \mid \barrier \mid \stores \mid \loads \mid \hbarrier \mid \lbarrier \mid \cbarrier \\
		\bnev  \bndefine &\ \events \mid \stores \mid \loads\\
		\bnrel  \bndefine &\ \name \define \bnterm
		%& & \mid \dom \bnrel \mid \ran \bnrel \\
		\end{align*}
		\caption{Core of CAT~\cite{cat}.}
		\label{fig:model}
	\end{minipage}
\end{figure}
%\bnexpr & \bndefine & {\mathbb{Z}} \mathrel{\rvert} {\langle reg \rangle} \mathrel{\rvert}  \bnexpr + \bnexpr \\ 
%& &\mathrel{\rvert} \bnexpr - \bnexpr \mathrel{\rvert}  \bnexpr \times \bnexpr \\
%\bnpred & \bndefine & {\mathbb{B}} \mathrel{\rvert} \bnexpr = \bnexpr \mathrel{\rvert} \bnpred \land \bnpred \\
%& & \mathrel{\rvert} \bnpred \lor \bnpred \mathrel{\rvert}  \neg \bnpred \\
% ---------------------------------------------------------------------------
% ---------------------------------------------------------------------------
% ---------------------------------------------------------------------------
% ---------------------------------------------------------------------------

\paragraph{\bf Executions.}
%\subsection{Executions}\label{sec:exec}

The semantics of a program is given in terms of \emph{executions}, partial orders where the events represent occurrences of the instructions and the ordering edges represent dependencies. 
The definition is given in~\autoref{fig:candidate}. 
An execution consists of a set $\executed$ of executed events and so-called \emph{base} and \emph{induced relations} satisfying the Axioms~\circled{\ref{b:sthd}}-\circled{\ref{b:co2}}.
%The set of all executions is denoted by $\execs$.
Base relations $\rfrel$ and $\corel$ and the set $\executed$ define an execution (they are the ones to be guessed). 
Induced relations can be extracted directly from the source code of the program. 
The axioms in~\autoref{fig:candidate} are common to all memory models and are natively implemented by our tool. 
To state them, let $\events$ represents memory events coming from program instructions accessing the memory.  
Memory accesses are either read or writes ${\events} := {{\loads} \cup {\stores}}$. 
By $\loads_l$ and $\stores_l$ we refer respectively to the reads and writes that access location $l$. 
%The equivalence relation $\slocrel$ expresses that events act on the same location. 
The events of thread $t$ form the set $\events_t$. %Set $\barrier$ is composed of three types of fences: high, low and control, i.e. ${\barrier} := {{\hbarrier} \cup {\lbarrier} \cup {\cbarrier}}$.
%The equivalence relation $\sthdrel$ identifies events from the same thread. 
%Induced relations are divided in two categories. 
Relations $\sthdrel$ and $\slocrel$ are equivalences relating events belonging to the same thread \circled{\ref{b:sthd}} and  accessing the same location \circled{\ref{b:sloc}}. 
Relations $\porel, \adrel, \ddrel$ and $\cdrel$ represent program order and address/data/control dependencies.
Axiom~$\circled{\ref{b:po1}}$ states that the \emph{program order} $\porel$ is an intra-thread relation which~$\circled{\ref{b:po2}}$ forms a total order when projected to events in the same thread (predicate $\tot r A$ holds if $r$ is a total order on the set~$A$). 
\emph{Address dependencies} are either read-to-read or read-to-write~$\circled{\ref{b:ad}}$, \emph{data dependencies} are read-to-write~$\circled{\ref{b:dd}}$, and \emph{control dependencies} originate from reads~$\circled{\ref{b:cd}}$.
Fence relations are architecture specific and relate only events in program order \circled{\ref{b:sync}}-\circled{\ref{b:mfence}}. 
Axiom~\circled{\ref{b:path}}, which we do not make explicit, requires the executed events $\executed$ to form a path in the threads' control flow. 
By Axioms~$\circled{\ref{b:rf1}}$ and $\circled{\ref{b:rf2}}$, the \emph{reads-from relation} $\rfrel$ gives for each read a unique write to the same location from which the read obtains its value. 
Here, $r_1 ; r_2 := \Set{ (x,y)}{\exists z: (x,z) \in r_1\text{ and } (z,y) \in r_2}$ is the composition of the relations $r_1$ and $r_2$. 
We write $r^{-1}:=\Set{(y, x)}{(x, y)\in r}$ for the inverse of relation~$r$. 
Finally, $\id(A)$ is the identity relation on the set $A$. 
By Axioms~$\circled{\ref{b:co1}}$ and $\circled{\ref{b:co2}}$, the \emph{coherence relation} $\corel$ relates writes to the same location, and it forms a total order for each location. 
We will assume the existence of an initial write event for each location which assigns value $0$ to the location. 
This event is first in the coherence order.  

\begin{figure}[t]
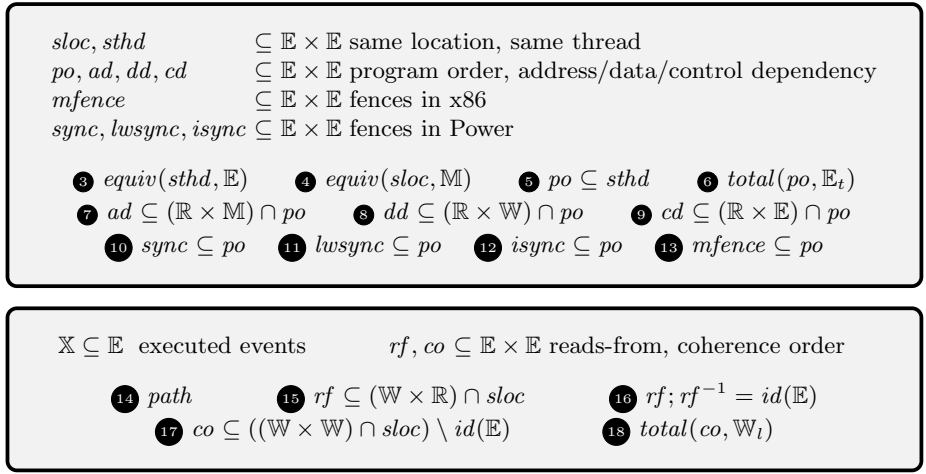

\footnotesize
\begin{mybox}{}
%Induced relations:\\[0.2cm]
$\begin{array}{lcll}
\slocrel, \sthdrel & \subseteq & \events \times \events & \text{same location, same thread}\\
\porel, \adrel, \ddrel, \cdrel & \subseteq & \events \times \events & \text{program order, address/data/control dependency}\\
\mfencerel & \subseteq & \events \times \events & \text{fences in x86}\\
\syncrel, \lwsyncrel, \isyncrel & \subseteq & \events \times \events & \text{fences in Power}\\
\end{array}$
\begin{center}
\rcounter{b:sthd}
$\circled{\thebulletscount}\ \equi{\sthdrel}{\events}$ \hspace{4mm}
\rcounter{b:sloc}
$\circled{\thebulletscount}\ \equi{\slocrel}{\memory}$ \hspace{4mm}
\rcounter{b:po1}
$\circled{\thebulletscount}\ \porel \subseteq \sthdrel$  \hspace{4mm}
\rcounter{b:po2}
$\circled{\thebulletscount}\ \tot{\porel}{\events_t}$ \\[0.05cm]
\rcounter{b:ad}
$\circled{\thebulletscount}\ \adrel \subseteq (\loads\times \memory)\cap \porel$ \hspace{4mm}
\rcounter{b:dd}
$\circled{\thebulletscount}\ \ddrel \subseteq (\loads\times \stores)\cap \porel$ \hspace{4mm}
\rcounter{b:cd}
$\circled{\thebulletscount}\ \cdrel \subseteq (\loads \times \events) \cap \porel$ \\[0.05cm]
\rcounter{b:sync}
$\circled{\thebulletscount}\ {\syncrel} \subseteq {\porel}$ \hspace{2mm}
\rcounter{b:lwsync}
$\circled{\thebulletscount}\ {\lwsyncrel} \subseteq {\porel}$ \hspace{2mm}
\rcounter{b:isync}
$\circled{\thebulletscount}\ {\isyncrel} \subseteq {\porel}$ \hspace{2mm}
\rcounter{b:mfence}
$\circled{\thebulletscount}\ {\mfencerel} \subseteq {\porel}$ \\
\end{center}
\end{mybox}
\begin{mybox}{}
$\begin{array}{lclllcll}
& \executed \subseteq \events & & \text{executed events} & \hspace{10mm}
\rfrel, \corel & \subseteq & \events \times \events & \text{reads-from, coherence order}
\end{array}$
\begin{center}
\rcounter{b:path}
$\circled{\thebulletscount}\ \mathit{path}$\hspace{10mm}
\rcounter{b:rf1}
$\circled{\thebulletscount}\ \rfrel \subseteq (\stores \times \loads) \cap \slocrel$\hspace{10mm}
\rcounter{b:rf2}
$\circled{\thebulletscount}\ \rfrel ;\rfrel^{-1} = \id(\events)$ \\[0.05cm]
\rcounter{b:co1}
$\circled{\thebulletscount}\ \corel \subseteq ((\stores \times \stores) \cap \slocrel) \setminus \id(\events)$ \hspace{10mm}
\rcounter{b:co2}
$\circled{\thebulletscount}\ \tot{\corel}{\stores_l}$\\
\end{center}
\end{mybox}
\caption{Executions; adapted from~\cite{memalloy}.}
\label{fig:candidate}
\end{figure}

\paragraph{\bf Memory Consistency Models.}
%\subsection{Memory Consistency Models}
%\label{sec:ConsExec}

We give in~\autoref{fig:model} a core subset of the CAT language for memory consistency models (MCMs). 
A \emph{memory model} is a constraint system over so-called \emph{derived relations}. 
Derived relations are built from the base and induced relations in an execution, hand-defined relations that refer to the different sets of events, and named relations that we will explain in a moment. 
%The operators are the standard set-theoretic ones, inverse, transitive closure (${}^+$), transitive and reflexive closure (${}^*$), and relational composition.
The assertions are acyclicity and irreflexivity constraints over derived relations.
CAT also supports recursive definitions of relations.  
We assume a set $\name$ of relation names (different from the predefined relations) and require that each name used in the memory model has associated a defining equation $\name:=\bnterm$. 
Notably, $\bnterm$ may again contain relation names, making the system of defining equations recursive.
The actual relations that are denoted by the names are defined to be
 the least solution to this system of equations.
%To be precise, the domain is the complete lattice of relations over the set of events in the execution of interest. 
%The operations allowed by CAT over relations are immediately checked to be monotone. 
%Hence, a unique least solution is guaranteed to exist by Knaster and Tarski's theorem~\cite{Tarski}. 
%As the domain is finite, monotonicity implies continuity and w
We can compute the least solution with a standard Kleene iteration~\cite{stoltenberg-hansen_lindström_griffor_1994} starting from the empty relations and iterating until the least fixed point is reached.

%{\it Architecture-specific Fence Semantics.} 
%Different architectures provide specific fence instructions with their corresponding semantics. In our setting, fence instructions generate relations between memory events. We split set $\barrier$ into three different sets, each of them representing a particular fence instruction; this allows us to model all fence instructions of the memory models used in this article, i.e. \emph{mfence}, \emph{sync}, \emph{lwsync} and \emph{isync}. Each instruction is associated with its corresponding relation: ${\mathit{mfence}} \define {\porel; \idd \hbarrier;\porel}$, ${\mathit{sync}} \define {\porel; \idd \hbarrier;\porel}$, ${\mathit{lwsync}} \define {\porel; \idd \lbarrier;\porel}$ and ${\mathit{cd\text{-}isync}} \define {\cdrel; \idd \cbarrier;\porel}$. Notice that the semantics of \emph{mfence} and \emph{sync} instructions coincide.

In~\autoref{sec:experiments} we study portability to Power;  we use its formalization~\cite{AlglaveMT14} in the core of CAT as given in~\autoref{fig:power}. 
Power is a highly relaxed memory model that supports program-order relaxations depending on address and data dependencies, that is non-multi-copy atomic, and that has a complex set of fence instructions.
The axioms defining Power are uniproc \circled{\ref{b:uniproc}} and the constraints \circled{\ref{b:power1}} to \circled{\ref{b:power3}}.
The model relies on the recursively defined relations $\iirel$, $\cirel$, $\icrel$, and $\ccrel$.
%Recall that these are the names that denote the relations obtained as the least solution to the four recursive equations.
%The recursive equations use several compositions and transitive closures which, as we will see, are the most expensive operators to encode. 

%Given a program $P$ and an MCM $\mm$, an execution $\exec$ of $P$ is \emph{consistent with $\mm$} if the derived relations of $\mm$ constructed from the base relations of $\exec$ satisfy the axioms of $\mm$. 
%We denote the set of consistent executions by $\consistent \mm P$. 
%For example, the execution of {\bf IRIW} shown in~\autoref{fig:iriw} satisfies Axioms \circled{\ref{b:uniproc}},\circled{\ref{b:power1}}-\circled{\ref{b:power3}}, but not \circled{\ref{b:tso}}. 
%Thus, the execution is consistent with Power but not with TSO.

\begin{figure*}[t]
\begin{mybox}{\emph{Consistent$_{Power}$}}
\footnotesize\centering
$\circled{\ref{b:uniproc}}\ \acy{(\porel \cap \slocrel) \cup \rfrel \cup \frrel \cup \corel}$ \hspace{6mm}
\rcounter{b:power1}
$\circled{\thebulletscount}\ \acy{\mathit{hb}}$ \hspace{6mm}
\rcounter{b:power2}
$\circled{\thebulletscount}\ \irref{\mathit{fre};\mathit{prop};\mathit{hb}^*}$ \hspace{6mm}
\rcounter{b:power3}
$\circled{\thebulletscount}\ \acy{\corel \cup \mathit{prop}}$
\end{mybox}
\centering
\footnotesize
\bbox[densely dashed,text=black]{\emph{Preserved Program Order}}{
${\mathit{dp}} \define {{\adrel} \cup {\ddrel}}$ \hspace{2mm}
${\mathit{rdw}} \define {{(\porel \cap \slocrel)} \cap {(\mathit{fre};\mathit{rfe})}}$ \hspace{2mm}
${\mathit{detour}} \define {{(\porel \cap \slocrel)} \cap {(\mathit{coe};\mathit{rfe})}}$ \\
${\mathit{ii_0}} \define {{\mathit{dp}} \cup {\mathit{rdw}} \cup {\mathit{rfi}}}$ \hspace{5mm}
${\mathit{ci_0}} \define {{\mathit{cd\text{-}isync}} \cup {\mathit{detour}}}$ \\
${\mathit{ic_0}} \define {\emptyset}$ \hspace{5mm}
${\mathit{cc_0}} \define {{\mathit{dp}} \cup {(\porel \cap \slocrel)} \cup {\cdrel} \cup {(ad;po)}}$ \\
${\iirel} \define {{\mathit{ii_0}} \cup {\cirel} \cup {(\icrel;\cirel)} \cup {(\iirel;\iirel)}}$ \hspace{1mm}
${\cirel} \define {{\mathit{ci_0}} \cup {(\cirel;\iirel)} \cup {(\ccrel;\cirel)}}$ \\
${\icrel} \define {{\mathit{ic_0}} \cup {\iirel} \cup {\ccrel} \cup {(\icrel;\ccrel)} \cup {(\iirel;\icrel)}}$ \hspace{1mm}
${\ccrel} \define {{\mathit{cc_0}} \cup {\cirel} \cup {(\cirel;\icrel)} \cup {(\ccrel;\ccrel)}}$ \\
${\mathit{ppo}} \define {{((\loads \times \loads) \cap \iirel)} \cup {((\loads \times \stores) \cap \icrel)}}$
}

\begin{minipage}{.58\linewidth}
\centering
\bbox[densely dashed,text=black]{\emph{Fences}}
{
${\mathit{fence}} \define {\mathit{sync} \cup {(\mathit{lwsync} \setminus (\stores \times \loads))}}$
}
\end{minipage}
\hspace{.2cm}
\begin{minipage}{.38\linewidth }
\centering
\bbox[densely dashed,text=black]{\emph{Thin Air}}
{ 
${\mathit{hb}} \define {{\mathit{ppo}} \cup {\mathit{fence}} \cup {\mathit{rfe}}}$
}
\end{minipage} 

\bbox[densely dashed,text=black]{\emph{Propagation}}{
${\mathit{prop\text{-}base}} \define {({\mathit{fence}} \cup {(\mathit{rfe};\mathit{fence})}); \mathit{hb}^*}$ \\
${\mathit{prop}} \define {({(\stores \times \stores)} \cap {\mathit{prop\text{-}base}}) \cup {(\mathit{com}^*;\mathit{prop\text{-}base}^*;\mathit{sync};\mathit{hb}^*)}}$
}
\caption{Power~\cite{AlglaveMT14}.}
\label{fig:power}
\end{figure*}
%!TEX root = sas2017.tex
% !TEX spellcheck = en-EN

\section{Portability Analysis}
\label{sec:encoding}

Let $\consistent{{\mm}} P$ be the set of executions of program $P$ consistent with $\mm$. Given a program $P$ and two MCMs $\msource$ and $\mtarget$, our goal is to find an execution $\exec$ which is consistent with the target ($\exec \in \consistent{{\mtarget}} P$) but not with the source ($\exec \not\in \consistent{{\msource}} P$). In such a case $P$ is not portable from $\msource$ to $\mtarget$.

\begin{definition}[Portability]
\label{def:traceportability}
Let $\msource$, $\mtarget$ be two MCMs. A program $P$ is \emph{portable from $\msource$ to $\mtarget$} if $\consistent{{\mtarget}} P \subseteq \consistent{{\msource}} P.$
\end{definition}

%\noindent Note than if $\msource$ is weaker than $\mtarget$ in the sense that for any program, every execution allowed by $\msource$ is also allowed by $\mtarget$, portability boils down to checking that the consistent executions of the program for both MCMs coincide. This is the case for example between SC and TSO, but not between the incomparable RMO and Alpha.

Our method finds non-portable executions as satisfying assignments to an SMT formula.
Recall that an execution is uniquely represented by the set $\executed$ and the relations $\rfrel$ and $\corel$, which need to be guessed by the solver. 
All other relations are derived from these guesses, the source code of the program, and the MCMs in question.  
Thus. we also have to encode the derived relations of the two MCMs defined in the language of~\autoref{fig:model}. 
As the last part, we encode the assertions expressed in the language of~\autoref{fig:model} on these relations in such a way that the guessed execution is allowed by $\mtarget$ (all the assertions stated for $\mtarget$ hold) while the same execution is not allowed by $\msource$ (at least one of the axioms of $\msource$ is violated).
%The sublogic of SMT we use is integer difference logic, which allows our formulas to be more compact than a pure Boolean encoding.
The full SMT formula is of the form ${\phi_{\mathit{CF}} \land \phi_{DF} \land \phi_{\mtarget} \land \phi_{\neg\msource}}.$
Here, $\phi_{\mathit{CF}}$ and $\phi_{DF}$ encode the control flow and data flow of the executions, $\phi_{\mtarget}$ encodes the derived relations and all assertions of $\mtarget$, and $\phi_{\neg\msource}$ encodes the derived relations of $\msource$ and a violation of at least one of the assertions of the source memory model. 
%For the generation of the encoding, we assume the program under test has already been unrolled.
The control-flow and data-flow encodings are standard for bounded model checking~\cite{CollavizzaR06}. 
%Loops of the program are first unrolled to a user specified depth to remove any cycles in the program. 
%represents the unrolling of the program as a direct acyclic graph, 
%The control-flow formula captures the branching and merging in the acyclic program. %branchings and how those control flows merge again. 
%The data-flow encoding relates the values of the variables according to program assignments. 
%For this, we first transform the program into static single assignment form.
%These parts of the encoding are standard approaches to encoding program analysis tasks into SMT formulas.
The rest of the section focuses on the parts that are new in this work: how to encode the derived relations needed for representing both the MCMs, % in \autoref{sec:encomm}, 
how to encode assertions for the target memory model %in \autoref{sec:encoass}, 
and how to encode an assertion violation in the source memory model. %in \autoref{sec:cycles}.
% The control-flow and data-flow encodings can be found in the appendix.

\paragraph{\bf Encoding Derived Relations.}
%\subsection{Encoding Derived Relations}
%\label{sec:encomm}

%\autoref{fig:model} shows the language for defining MCMs. We first show how to encode the derived relations using SMT.
For any pair of events $e_1, e_2 \in \events$ and relation $r \subseteq \events \times \events$ we use a Boolean variable $\var{r}{e_1}{e_2}$ representing the fact that $e_1 \stackrel{r}{\rightarrow} e_2$ holds. 
We similarly use fresh Boolean variables to represent the derived relations, using the encoding to force its value as follows. 
For the union (resp. intersection) of two relations, at least one of them (resp. both of them) should hold; set difference requires that the first relation holds and the second one does not; %reflexive closure checks if the events are the same or are already related by $r$; 
for the composition of relations we iterate over a third event and check if it belongs to the range of the first relation and the domain of the second.
Computing a reverse relation requires reversing the events.
We define the transitive closure of $\mathit{r}$ recursively where the base case $tc_0$ holds if events are related according to $\mathit{r}$ and the recursive case uses a relation composition. These are computed with the iterative squaring technique using the relation composition. Finally reflexive and transitive closure checks if the events are the same or are related by $\mathit{r}^+$.
The encodings are summarized below.

%\begin{minipage}{.48\linewidth}
%$$\begin{array}{rcl}
%\var{{r$_1$}$\cup${r$_2$}}{e_1}{e_2} & \Leftrightarrow & {\var{r$_1$}{e_1}{e_2}} \lor {\var{r$_2$}{e_1}{e_2}} \\
%\var{{r$_1$}$\cap${r$_2$}}{e_1}{e_2} & \Leftrightarrow & {\var{r$_1$}{e_1}{e_2}} \land {\var{r$_2$}{e_1}{e_2}} \\
%\var{{r$_1$}$\setminus${r$_2$}}{e_1}{e_2} & \Leftrightarrow & {\var{r$_1$}{e_1}{e_2}} \land \neg {\var{r$_2$}{e_1}{e_2}} \\
%\end{array}$$
%\end{minipage}
%\hspace{.2cm}
%\begin{minipage}{.48\linewidth }
%$$\begin{array}{rcl}
%\var{r$_1$;r$_2$}{e_1}{e_2} & \Leftrightarrow & {\bigvee\limits_{e_3 \in \events}{{\var{r$_1$}{e_1}{e_3}} \land {\var{r$_2$}{e_3}{e_2}}}}\\
%\var{r$^{-1}$}{e_1}{e_2} & \Leftrightarrow & \var{r}{e_2}{e_1} \\
%\var{r$^*$}{e_1}{e_2} & \Leftrightarrow & {\var{r$^+$}{e_1}{e_2}} \lor ({e_1} = {e_2}) \\
%\end{array}$$
%\end{minipage} 
%$$
%\var{r$^+$}{e_1}{e_2}  \Leftrightarrow  {\var{tc$_{\lceil \log |\events| \rceil}$}{e_1}{e_2}}\mathit{, where }$$
%$$\var{tc$_0$}{e_1}{e_2}  \Leftrightarrow  {\var{r}{e_1}{e_2}}\mathit{ ,and } \var{tc$_{i+1}$}{e_1}{e_2}  \Leftrightarrow  {\var{r}{e_1}{e_2}} \lor ({\var{tc$_{i}$}{e_1}{e_2}}; {\var{tc$_{i}$}{e_1}{e_2}}).
%$$
\vspace{-3 mm}
$$\begin{array}{rclrcl}
\var{{r$_1$}$\cup${r$_2$}}{e_1}{e_2}  &\Leftrightarrow  &{\var{r$_1$}{e_1}{e_2}} \lor {\var{r$_2$}{e_1}{e_2}} &% \;\;&
\var{{r$_1$}$\cap${r$_2$}}{e_1}{e_2}  &\Leftrightarrow & {\var{r$_1$}{e_1}{e_2}} \land {\var{r$_2$}{e_1}{e_2}} \\
\var{{r$_1$}$\setminus${r$_2$}}{e_1}{e_2} & \Leftrightarrow & {\var{r$_1$}{e_1}{e_2}} \land \neg {\var{r$_2$}{e_1}{e_2}} &% &
\var{r$^{-1}$}{e_1}{e_2}  & \Leftrightarrow & \var{r}{e_2}{e_1} \\
\var{r$_1$;r$_2$}{e_1}{e_2} & \Leftrightarrow & {\bigvee\limits_{e_3 \in \events}{{\var{r$_1$}{e_1}{e_3}} \land {\var{r$_2$}{e_3}{e_2}}}} &% &
\var{r$^*$}{e_1}{e_2}  &\Leftrightarrow & {\var{r$^+$}{e_1}{e_2}} \lor ({e_1} = {e_2}) \\
%\var{r$^+$}{e_1}{e_2}  \Leftrightarrow  {\var{tc$_{\lceil \log |\events| \rceil}$}{e_1}{e_2}}\mathit{, where} & &
%\var{tc$_0$}{e_1}{e_2}  \Leftrightarrow  {\var{r}{e_1}{e_2}}\mathit{, and}
\end{array}$$
\vspace{-2 mm}
%$\var{tc$_{i+1}$}{e_1}{e_2}  \Leftrightarrow  {\var{r}{e_1}{e_2}} \lor ({\var{tc$_{i}$}{e_1}{e_2}}; {\var{tc$_{i}$}{e_1}{e_2}}).$
$$\begin{array}{rcl}
\var{r$^+$}{e_1}{e_2} & \Leftrightarrow & {\var{tc$_{\lceil \log |\events| \rceil}$}{e_1}{e_2}}\mathit{, where}\\
\var{tc$_0$}{e_1}{e_2} & \Leftrightarrow & {\var{r}{e_1}{e_2}}\mathit{, and}\\
\var{tc$_{i+1}$}{e_1}{e_2} & \Leftrightarrow & {\var{r}{e_1}{e_2}} \lor ({\var{tc$_{i}$}{e_1}{e_2}}; {\var{tc$_{i}$}{e_1}{e_2}}).\\
\end{array}$$

%$$\begin{array}{rcl}
%\var{{r$_1$}$\cup${r$_2$}}{e_1}{e_2} & \Leftrightarrow & {\var{r$_1$}{e_1}{e_2}} \lor {\var{r$_2$}{e_1}{e_2}} \\
%\var{{r$_1$}$\cap${r$_2$}}{e_1}{e_2} & \Leftrightarrow & {\var{r$_1$}{e_1}{e_2}} \land {\var{r$_2$}{e_1}{e_2}} \\
%\var{{r$_1$}$\setminus${r$_2$}}{e_1}{e_2} & \Leftrightarrow & {\var{r$_1$}{e_1}{e_2}} \land \neg {\var{r$_2$}{e_1}{e_2}} \\
%\var{r$_1$;r$_2$}{e_1}{e_2} & \Leftrightarrow & {\bigvee\limits_{e_3 \in \events}{{\var{r$_1$}{e_1}{e_3}} \land {\var{r$_2$}{e_3}{e_2}}}}\\
%\var{r$^{-1}$}{e_1}{e_2} & \Leftrightarrow & \var{r}{e_2}{e_1} \\
%\var{r$^*$}{e_1}{e_2} & \Leftrightarrow & {\var{r$^+$}{e_1}{e_2}} \lor ({e_1} = {e_2}) \\
%\var{r$^+$}{e_1}{e_2} & \Leftrightarrow & {\var{tc$_{\lceil \log |\events| \rceil}$}{e_1}{e_2}}\mathit{, where}\\
%\var{tc$_0$}{e_1}{e_2} & \Leftrightarrow & {\var{r}{e_1}{e_2}}\mathit{, and}\\
%\var{tc$_{i+1}$}{e_1}{e_2} & \Leftrightarrow & {\var{r}{e_1}{e_2}} \lor ({\var{tc$_{i}$}{e_1}{e_2}}; {\var{tc$_{i}$}{e_1}{e_2}}).\\
%\end{array}$$

%\paragraph{Encoding Recursively-Defined Relations:} 
Recall that some of the relations (e.g. $\iirel$ and $\icrel$ of Power) can be defined mutually recursively, and that we are using the least fixed point (smallest solution) semantics for cyclic definitions. 
A classical algorithm for solving such equations
is the Kleene fixpoint iteration. %\flo{add second cite?}
The iteration starts from the empty relations as initial approximation and on each round computes a new approximation until the (least) fixed point is reached. 
Such an iterative algorithm can be easily encoded into SAT. 
The problem of such an encoding is the potentially large number of iterations needed, and thus the resulting formula size can grow to be large. 
A more clever way to encode this is an approach that has been already used in earlier work on encoding mutually recursive monotone equation systems with nested least and greatest fixpoints~\cite{DBLP:journals/jcss/HeljankoKLN12}. %, a framework more expressive than the one used in this work. 
The encoding of this paper uses an extension of SAT with integer difference logic (IDL), a logic that is still NP complete. %and supported by the SMT solver natively. 
A SAT encoding is also possible but incurs an overhead in the encoding size: if the SMT encoding is of size $O(n)$, the SAT encoding is of size $O(n \log n)$~\cite{DBLP:journals/jcss/HeljankoKLN12}. We chose IDL since our experiments showed the encoding to be the most time consuming of the tasks.
%The encoding is also closely related to~\cite{DBLP:journals/amai/Niemela08,DBLP:conf/birthday/JanhunenN11} which encodes logic programming under the stable model semantics. 
%They also use a least fixed point semantics for recursive definitions encoded into integer difference logic and SAT. 
%While the basic encoding idea is not new, the application of the approach to encoding axiomatic memory models is novel, and can be of interest to other researchers for efficiently encoding, e.g. Power memory model semantics using an SMT solver.
%Our MCMs formalization allows to compute relations mutually recursively.
%
% We use an encoding that computes a set of recursive relations using a least fix-point iteration. Suppose a relation is defined as $\mathit{f} \define {\mathit{br} \cup \mathit{rr}}$ where $br$ represent the base case relation and $rr$ is some derived relation including $\mathit{f}$ (the recursive case). Besides the Boolean variables mentioned above, for each relation we use a numerical variable $\Phi^\texttt{r}_{e_1,e_2} \in \mathbb{N}$ which represents the iteration at which the relation reaches a fix point; then $\mathit{f}(e_1,e_2)$ holds if either $\mathit{br}(e_1,e_2)$ or $\mathit{rr}(e_1,e_2)$ hold already at some previous iteration. We formalize this with the following encoding:

The basic idea of the encoding is to guess a certificate that contains the iteration number in which a tuple would be added to the relation in the Kleene iteration. 
For this we use additional integer variables and enforce that they actually locally follow the propagations made by the fixed point iteration algorithm. 
Thus, for any pair of events $e_1, e_2 \in \events$ and relation $r \subseteq \events \times \events$ we introduce an integer variable $\Phi^\texttt{r}_{e_1,e_2}$ representing the round in which $\var{r$$}{e_1}{e_2}$ would be set by the Kleene iteration algorithm. 
Using these new variables we guess the execution of the Kleene fixed point iteration algorithm, and then locally check that every guess that was made is also a valid propagation of the fixed point iteration algorithm. 
For a simple example on how the encoding for the union of relations needs to be modified to also handle recursive definitions, consider a definition where ${r_1} \define {{r_2} \cup {r_3}}$ and ${r_2} \define {{r_1} \cup {r_4}}$.
%The integer variable for $\var{r$_1$}{e_1}{e_2}$ is  $\Phi^\texttt{r$_{1}$}_{e_1,e_2}$. % , and such integer variables need to be introduced for all member of relations. 
The encoding is as follows
\vspace{-.25mm}
\[\arraycolsep=2.5pt\def\arraystretch{.2}
\begin{array}{rcl}
%{\bigwedge\limits_{e_1,e_2 \in \events} 
\var{r$_1$}{e_1}{e_2} & \Leftrightarrow & {(\var{r$_2$}{e_1}{e_2} \land ({\Phi^\texttt{r$_{1}$}_{e_1,e_2} > \Phi^\texttt{r$_{2}$}_{e_1,e_2}}))} \lor {} 
 (\var{r$_3$}{e_1}{e_2} \land ({\Phi^\texttt{r$_{1}$}_{e_1,e_2} > \Phi^\texttt{r$_{3}$}_{e_1,e_2}}))\\
 \\
%{\bigwedge\limits_{e_1,e_2 \in \events} 
\var{r$_2$}{e_1}{e_2} & \Leftrightarrow & {(\var{r$_1$}{e_1}{e_2} \land ({\Phi^\texttt{r$_{2}$}_{e_1,e_2} > \Phi^\texttt{r$_{1}$}_{e_1,e_2}}))} \lor {}  (\var{r$_4$}{e_1}{e_2} \land ({\Phi^\texttt{r$_{2}$}_{e_1,e_2} > \Phi^\texttt{r$_{4}$}_{e_1,e_2}})). \\
\end{array}
\]
%If both $r_3$ and $r_4$ happen to be empty relations, the $r_1$ and $r_2$ relations need to become identical. Also, any candidate solution where both relations are non-empty is not a solution, as in particular the integer variables will not have any solution where both $(\Phi^\texttt{r$_{1}$}_{e_1,e_2} > \Phi^\texttt{r$_{2}$}_{e_1,e_2})$ and $(\Phi^\texttt{r$_{2}$}_{e_1,e_2} > \Phi^\texttt{r$_{1}$}_{e_1,e_2})$ hold at the same time. 
%Thus, the only satisfying assignment for the encoding is the one where both $r_1$ and $r_2$ are empty, which is the correct least fixed point semantics.
%A similar  encoding is needed for all derived relations that can induce recursive dependencies.
A pair $(e_1,e_2)$ is added to $r_1$ by the Kleene iteration in step $\Phi^\texttt{r$_{1}$}_{e_1,e_2}$. It comes from either $r_2$ or $r_3$. If it came from $r_2$ then it is of course also in $r_2$ and it was added to $r_2$ in an earlier iteration $\Phi^\texttt{r$_{2}$}_{e_1,e_2}$ and thus $({\Phi^\texttt{r$_{1}$}_{e_1,e_2} > \Phi^\texttt{r$_{2}$}_{e_1,e_2}})$. It is similar if it came from $r_3$. 
The only satisfying assignment for the encoding is one where both $r_1$ and $r_2$ are the union of $r_3$ and $r_4$.

\paragraph{\bf Encoding Target Memory Model Assertions.}
%\subsection{Encoding Target Memory Model Assertions}
%\label{sec:encoass}

For the target architecture we need to encode all acyclicity and irreflexivity assertions of the memory model. %This can be done with a conjunction. % of all the acyclicity and irreflexivity assertions.
For handling acyclicity we again use non-Boolean variables in our SMT encoding for compactness reasons. One can encode that a relation is acyclic by adding a numerical variable 
$\Psi_e \in \mathbb{N}$ for each event $e$ in the relation we want to be acyclic. Then acyclicity of relation $r$ is encoded as 
${\acy r} \Leftrightarrow {\hspace{-1mm}\bigwedge\limits_{{e_1,e_2} \in {\events}} \hspace{-1mm} (\var{r}{e_1}{e_2} \imp (\Psi_{e_1} < \Psi_{e_2}})).$
Notice that we can impose a total order with $ \Psi_{e_1} < \Psi_{e_2}$ only if there is no cycle.
%, this encodes acyclicity.
%Note that the constraints on the numerical variables 
%become unsatisfiable if there is any cycle in the relation $r$. However, for any acyclic relations the solver is able to find values for all the integer variables. 
Our encoding is the same as the SAT + IDL encoding in~\cite{DBLP:conf/jelia/GebserJR14} where more discussion of SAT modulo acyclicity can be found.
%Acyclicity assertions are enforced using numerical variables; each event $e$ is associated with a timestamp $\texttt{r}_e \in \mathbb{N}$ representing the moment in which the event can be executed according to $r$: $${\acy r} \Leftrightarrow {\hspace{-1mm}\bigwedge\limits_{{e_1,e_2} \in {\events}} \hspace{-1mm} \var{r}{e_1}{e_2} \imp \texttt{r}_{e_1} < \texttt{r}_{e_2}}$$It is worth noticing that the use of numerical variables allows to encode acyclicity in $\mathcal{O}(n)$ size (where $n$ is the size of $r$ and thus quadratic on the size of $\events$) instead of $\mathcal{O}(n \log n)$ as it is the best case when using SAT. However this reduction does not improve the overall complexity of our method which is cubic.
%de that a relation is a total order (over a certain domain $A$), we add an additional constraint to check specifying that every pair of events should be related:$${\mathit{total(r,A)}} \Leftrightarrow {\bigwedge\limits_{e_1, e_2 \in A} ({\var{r}{e_1}{e_2}} \lor {\var{r}{e_2}{e_1}}}) \land\ {\mathit{acyclic(r)}}$$
%The irreflexive constraint assures that not event is related to itself, and is simply encoded as: $${\irref r} \Leftrightarrow {\bigwedge\limits_{e \in \events} \neg \var{r}{e}{e}}$$
The irreflexive constraint is simply encoded as: ${\irref r} \Leftrightarrow {\bigwedge\limits_{e \in \events} \neg \var{r}{e}{e}}.$
%and the empty assertion assures that no pair of events are related:$${\empt r} \Leftrightarrow {\bigwedge\limits_{e_1,e_2 \in \events} \neg \var{r}{e_1}{e_2}}$$

%Notice that $\mathit{equiv(r,A)}$ is used to define $\sthdrel$ and $\slocrel$; this can be encodes as:
%$$\begin{array}{lll}
%{\hspace{-2mm} \equi{\sthdrel}{\events}} & \hspace{-2.5mm} \Leftrightarrow & {\hspace{-6mm}\bigwedge\limits_{{e_1,e_2} \in {\events}} \hspace{-2.5mm} \var{sthd}{e_1}{e_2} \Leftrightarrow \mathit{thrd}(e_1) = \mathit{thrd}(e_2)} \\
%{\hspace{-2mm} \equi{\slocrel}{\memory}} & \hspace{-2.5mm} \Leftrightarrow & {\hspace{-6mm}\bigwedge\limits_{{e_1,e_2} \in {\events}} \hspace{-2.5mm} \var{sloc}{e_1}{e_2} \Leftrightarrow \mathit{loc}(e_1) = \mathit{loc}(e_2)}
%\end{array}$$

\paragraph{\bf Encoding Source Memory Model Assertions.}
%\subsection{Encoding Source Memory Model Assertions}
%\label{sec:cycles}
For the source architecture we have to encode that one of the derived relations does not fulfill its assertions. On the top level this can be encoded as a simple disjunction over all the assertions of the source memory model, forcing at least one of the irreflexivity or acyclicity constraints to be violated.

For the irreflexivity violation, we can reuse the same encoding as for the target memory model simply as $\neg \irref{r}$. %Now assume we have one acyclicity constraint on relation $r$ and one irreflexivity constraint on relation $r'$. Then the violation of at least of them is encoded as $\cyc r \lor \neg \irref{r'}$, where we can reuse the same encoding for irreflexivity as for the target memory model.\flo{That top level is trivial, it should be enough to say one of the asserions doesn't hold so we need to encode cyclicity (reflexiveness is trivial).}
What remains to be encoded is $\cyc r$, which requires the relation $r$ to be cyclic. 
Here, we give an encoding that uses only Boolean variables. We add Boolean variables $\cvar e$ and $\cedge{\texttt{r}}{e_1}{e_2}$, which guess the edges and nodes constituting the cycle. %\footnote{To improve performance, we actually encode that events and edges \emph{may} belong to the cycle, i.e. for every event (resp. edge) in the cycle its corresponding variable is set to true, but a variable can be set to true even if the event or the edge are not pat of the cycle.}
We ensure that for every event in the cycle, there should be at least one incoming edge and at least one outgoing edge that are also in the cycle: 
$${c_{\mathit{n}}} = {\bigwedge_{e_1 \in \events} ({\cvar {e_1}} \imp ({\bigvee_{{e_2} \rrel {e_1}} \cedge{\text{r}}{e_2}{e_1}} \wedge {\bigvee_{{e_1} \rrel {e_2})} \cedge{\text{r}}{e_1}{e_2}}))}.$$
If an edge is guessed to be in a cycle, the edge must belong to relation $r$, and both events must also be guessed to be on the cycle: $$c_{e} = {\bigwedge_{e_1, e_2 \in \events} ({\cedge{\texttt{r}}{e_1}{e_2}} \imp ({\var{r}{e_1}{e_2}} \wedge {\cvar{e_1}} \wedge {\cvar{e_2}})}).$$
%Finally we need to ensure that if we guess this relation to be cyclic, the formulas above hold, and at least one of the events belong to the cycle:
%$$\cyc r \imp ({c_{e}} \wedge {c_{n}} \wedge {\bigvee_{e \in \events} \cvar e})$$
A cycle exists, if these formulas hold and there is an event in the cycle: $$\cyc r \Leftrightarrow ({c_{e}} \wedge {c_{n}} \wedge {\bigvee_{e \in \events} \cvar e}).$$
%!TEX root = sas2017.tex
% !TEX spellcheck = en-EN

\section{State Portability}
\label{sec:complexity}
%\flo{bounded programs or bounded portability?}
%Say the outcome of our study: Complexity-wise, SAT does not work. Experiments show that it is not worth it.} 
Portability from $\msource$ to $\mtarget$ requires that there are no new executions in $\mtarget$ that did not occur in $\msource$.  
One motivation to check portability is to make sure that safety properties of $\msource$ carry over to $\mtarget$. 
Safety properties only depend on the values that can be computed, not on the actual executions. 
Therefore, we now study a more liberal notion of so-called \emph{state portability}: $\mtarget$ may admit new executions as long as they do not compute new states.
%\roland{This link to fence insertion needs more text. The reader does not get the jump.}
Admitting more executions means we require less synchronization (fences) to consider a ported program correct, and thus state portability promises more efficient code. 
The notion has been used in~\cite{KupersteinVY12}.
%State portability is attractive as it admits more executions, and hence needs le
%A common task is enforcing portability of a program by inserting fences. Each fence forbids certain executions in the target architecture that would prevent portability.
%Since we want to insert as few fences as possible (they slow down computations), we want to allow more executions. 

The main finding in this section is negative: a polynomial encoding of state portability to SAT does not exist (unless the polynomial hierarchy collapses). 
Phrased differently, state portability does not admit an efficient bounded analysis (like our method for portability).
We remind the reader that we restrict our input to acyclic programs (that can be obtained from while-programs with bounded unrolling). 
For while-programs, verification tasks are generally undecidable \cite{Rice}.

Fortunately, our experiments indicate that new executions often compute new states.
This means portability is not only a sufficient condition for state portability but, in practice, the two are equivalent.
Combined with the better algorithmics of portability, we do not see a good motivation to move to state portability.
%We turn to the development. 

A state is a function that assigns a value to each location and register. 
An execution $X$ computes the state $\reach X$ defined as follows: a location receives the value of the last write event (according to $\corel$) accessing it; 
for a register, its value depends on the last event in $\porel$ that writes to it. 
The relationship between the notions is as in Lemma~\ref{lem:statevstrace}.
%A program is \emph{state portable from $\msource$ to $\mtarget$} if every state that can be computed with an $\mtarget$-consistent execution can also be computed with some $\msource$-consistent execution.
% ------------------------------------------------------------------------
% ------------------------------------------------------------------------
% ------------------------------------------------------------------------
% ------------------------------------------------------------------------

\begin{definition}[State Portability]\label{def:StatePort}
Let $\msource$, $\mtarget$ be MCMs. Program $P$ is \emph{state portable from $\msource$ to $\mtarget$} if 
%Let $P$ be a program and $\msource$, $\mtarget$ two MCMs. $P$ is state-portable from $\msource$ to $\mtarget$ if: 
$\reach{\consistent{\mtarget}{P}}\subseteq\reach{\consistent{\msource}{P}}.$
\end{definition}

%\begin{align*}
%\reach{\consistent{\mtarget}{P}}\subseteq\reach{\consistent{\msource}{P}}.
%\end{align*}
% ----------------------------------------------------------------------------
% ----------------------------------------------------------------------------
% ----------------------------------------------------------------------------
% ----------------------------------------------------------------------------
%\begin{definition}[State Portability]\label{def:StatePort}
%Program $P$ is \emph{state portable from $\msource$ to $\mtarget$} if 
%\end{definition}
% ----------------------------------------------------------------------------
% ----------------------------------------------------------------------------
% ----------------------------------------------------------------------------
% ----------------------------------------------------------------------------
%Portability is a sufficient but not necessary for state portability.

% ----------------------------------------------------------------------------
% ----------------------------------------------------------------------------
% ----------------------------------------------------------------------------
% ----------------------------------------------------------------------------
\begin{lemma}\label{lem:statevstrace}
(1) Portability implies state portability. 
(2) State portability does not imply portability. 
\end{lemma}
% ----------------------------------------------------------------------------
% ----------------------------------------------------------------------------
% ----------------------------------------------------------------------------
% ----------------------------------------------------------------------------
\noindent For \autoref{lem:statevstrace}.(2), consider a variant of {\bf IRIW} (\autoref{fig:iriw}) where all written values are $0$. 
The program is trivially state portable from Power to TSO, but like {\bf IRIW}, not portable.

We turn to the hardness argumentation.
To check state portability, every $\mtarget$-computable state seems to need a formula  checking whether some $\msource$-consistent execution computes it.
The result would be  an exponential blow-up or a quantified Boolean formula, which is not practical.  
But can this exponential blow-up or quantification be avoided by some clever encoding trick?
The answer is no! 
Theorem~\ref{thm:TraceComplete0} shows that state portability is in a higher class of the polynomial hierarchy than portability.
So state portability is indeed harder to check than portability.

The polynomial hierarchy~\cite{STOCKMEYER19761} contains complexity classes between NP and PSPACE. 
Each class is represented by the problem of checking validity of a Boolean formula with a fixed number of quantifier alternations.
We need here the classes $\text{co-NP}=\Pi^P_1 \subseteq \Pi^P_2$.
The \emph{tautology problem} (validity of a closed Boolean formula with a universal quantifier $\forall x_1\dots x_n:\psi$ ) is a $\Pi^P_1$-complete problem. 
The higher class $\Pi^P_2$ allows for a second quantifier:
validity of a formula ($\forall x_1 \dots x_n \exists y_1 \dots y_n:\psi$) is a $\Pi^P_2$-complete problem. 
Theorem~\ref{thm:TraceComplete0} refers to a class of common memory models that we define in a moment. 
Moreover, we assume that the given pair of memory models $\msource$ and $\mtarget$ is \emph{non-trivial} in the sense that $\consistent{{\mtarget}} P \subseteq  \consistent{{\msource}} P$ fails for some program, and similar for state portability.
%It is generally assumed, but not proven, that $\Pi^P_1 \neq \Pi^P_{2}$holds. 
%\roland{Assumed but not proven sounds like people did not bother. What you want to say: PH has been developed as a way to showing P vs. NP vs. PSPACE (stratifies the problem). }
%TODO: add ref.
% -------------------------------------------------------------------------
% -------------------------------------------------------------------------
% -------------------------------------------------------------------------
% -------------------------------------------------------------------------

\begin{restatable}{theorem}{tracecomplete}
%\begin{theorem}
\label{thm:TraceComplete0}
Let $\msource, \mtarget$ be a non-trivial pair of common MCMs.
(1) Portability from $\msource$ to $\mtarget$ is $\Pi^P_1$-complete. 
(2) State portability is $\Pi^P_2$-complete. 
\end{restatable}
% -------------------------------------------------------------------------
% -------------------------------------------------------------------------
% -------------------------------------------------------------------------
% -------------------------------------------------------------------------

By Theorem~\ref{thm:TraceComplete0}.(2), state portability cannot be solved efficiently. 
The first part says that our portability analysis is optimal. 
We focus on this lower bound to give a taste of the argumentation: given a non-trivial pair of memory models, we know there is a program that is not portable.
Crucially, we do not know the program but give a construction that works for any program. 
The proof of Theorem~\ref{thm:TraceComplete0}.(2) is along similar lines but more involved. 
% ----------------------------------------------------------------------
% ----------------------------------------------------------------------
% ----------------------------------------------------------------------
% ----------------------------------------------------------------------

\begin{definition}\label{Definition:Common}
We call an MCM \emph{common}\footnote{Notice that all memory models considered in~\cite{AlglaveMT14} and in this paper are common ones.} if
\begin{itemize}
	\item[(i)] the inverse operator is only used in the definition of $\frrel$, 
	\item[(ii)] the constructs $\sthdrel$, $\slocrel$, and ${\bnev} \times {\bnev} $ are only used to restrict (in a conjunction) other relations,
	\item[(iii)] it satisfies \uniproc\ (Axiom~\circled{\ref{b:uniproc}}) , and
	\item[(iv)] every program is portable from this MCM to SC.
\end{itemize}
%\flo{more formal?}\roland{What is common in these assumptions? Again, the name begs for bad comments.}
%We say the model is inverse-free.
%\flo{is ;(i) correct? looks odd.}
\end{definition}
% ----------------------------------------------------------------------
% ----------------------------------------------------------------------
% ----------------------------------------------------------------------
% ----------------------------------------------------------------------

We explain the definition.  
When formulating a MCM, one typically forbids well-chosen cycles of base relations (and~$\frrel$).
To this end, derived relations are introduced that capture the paths of interest, and acyclicity constraints are imposed on the derived relations. 
%This is done by restricting the paths between events with some derived relation. \hp{I would say ``this is done by impose acyclicity constraints on derived relation"}
The operators inverse and $\bnev \times \bnev$ may do the opposite, they add   relations that do not correspond to paths of base relations (and~$\frrel$). 
Besides stating what is common in MCMs,   
Properties~\emph{(i)} and~\emph{(ii)} help us compose programs (cf. next paragraph). 
Uniproc is a fundamental property without which an MCM is hard to program.
Since the purpose of an MCM is to capture SC relaxations, we can assume MCMs to be weaker than SC. 
Properties~\emph{(iii)} and~\emph{(iv)} guarantee that the program $\Ppsi$ given below is portable between any common MCMs.

The crucial property of common MCMs is the following.
For every pair of events $e_1, e_2$ in a derived relation, (1) there are (potentially several) sequences of base relations (and $\frrel$) that connect $e_1$ and $e_2$, and 
(2) the derived relation only depends on these sequences.
The property ensures that if we append a program $P'$ to a location-disjoint program $P$,  
consistency of composed executions is preserved.

\begin{comment}
For programs with arithmetic over integers of fixed length (say 32 bits), we can adapt our method to construct a polynomial-sized Boolean formula with a universal quantifier that is valid if the program is portable. It follows that portability is polynomial-time reducible to the tautology problem and thus it is in $\Pi^P_1$.
%The encoding presented in \autoref{sec:encoding} constructs an SMT-formula $\psi$ of polynomial size in the program. 
%We can easily convert this into an equivalent Boolean formula by representing numbers with a sequence of Boolean variables (using 32 bits to reflect the usual integer encoding).
%We can built a Boolean formula $\psi$ of polynomial size, that is satisfied if there is an execution that is not portable.
%So the program is portable if $\forall x_1, \dots, x_n : \neg \psi$ is valid.
\end{comment}

It remains to prove $\Pi^P_1$-hardness of portability. 
We first introduce the program $P_\psi$ that generates some assignment  and checks if it satisfies the Boolean formula $\psi (x_1 \ldots x_m)$ (over the variables $x_1\ldots x_m$). 
%Let $\psi (x_1 \dots x_m)$ be a Boolean formula over the variables $x_1\dots x_m$.
The program $P_\psi \define  t_1\parallel t_2$ consists of the two threads $t_1$ and $t_2$ defined below. 
Note that we cannot directly write a constant $i$ to a location, so we first assign $i$ to register $r_{c,i}$.
%This program generates some assignment and checks if it satisfies $\psi$. If it does, then it writes $y=2$ otherwise it writes $y=1$. 

\[\arraycolsep=2.5pt
{\small
\begin{array}{l|l}
\multicolumn{1}{c|}{thread\ t_1} & \multicolumn{1}{c}{thread\ t_2} \\
\hline
r_{c,0} \leftarrow 0;\ r_{c,1}\leftarrow 1;\ r_{c,2} \leftarrow 2\ & \ r_{c,1} \leftarrow 1; \\
x_1 \define r_{c,0} \dots x_m\define r_{c,0}; & \ x_1 \define r_{c,1} \dots x_m\define r_{c,1}; \\
r_1 \leftarrow x_1 \dots r_m \leftarrow x_m; & \\
\textbf{if } \psi(r_1 \dots r_m) \textbf{ then} & \\
\ \ \ \ y \define r_{c,2}; & \\
\textbf{else}\ \  y \define r_{c,1}; & \\
\end{array}}
\]

We reduce checking whether $\forall x_1\ldots x_m:\psi(x_1\ldots x_m)$ holds to portability of a program $\Pallpsi$. 
The idea for $\Pallpsi$ is this. 
First $\Ppsi$ is run, it guesses and evaluates an assignment for $\psi$.
If $\psi$ is not satisfied ($y=1$), then some non-portable program $\Pnp$ is executed.
The program $\Pallpsi$ is portable iff the non-portable part is never executed. This is the case iff $\psi$ is satisfied by all assignments.

Let $\msource$, $\mtarget$ be common and non-trivial. 
%The porting problem from $\msource$ to $\mtarget$ is trivial if for all programs we have $ \consistent{{\mtarget}} P \subseteq  \consistent{{\msource}} P$.
By non-triviality, there is a program $\Pnp= t'_1\parallel \dots \parallel t'_k$ that is {\bf n}ot {\bf p}ortable from $\msource$ to $\mtarget$. 
We can assume $\Pnp$ has no registers or locations in common with $\Ppsi$. 
Program $\Pallpsi$ prepends $\Ppsi$ to the first two threads of $\Pnp$.
Once $y=1$, $\Pnp$ starts running.
Formally, let $t_1$ and $t_2$ be the threads in $\Ppsi$ and let $t_i:=skip$ for $3\leq i\leq k$. 
%\flo{skip is not in program syntax anymore. leave it like that or say epsilon instead or use separate definition $t''_i$ for $i>2$?}
We define $\Pallpsi\ :=\ t''_1 \parallel \dots \parallel t''_k$ with
%\begin{align*}
$t''_i\ :=\ t_i;\ r \leftarrow y;\ \textbf{if} (r=1)\ \textbf{then } t'_i$.
%\end{align*}

%From the definition of common MCMs, it follows that if there is a pair ${(e_1,e_2)} \in \mathit{r}$ with $r$ a derived relation, then there is a path from $e_1$ to $e_2$ that consists of the base relations and $\frrel$ (see Lemma~\ref{lem:common} in the appendix). 
% -------------------------------------------------------------
% -------------------------------------------------------------
% -------------------------------------------------------------
% -------------------------------------------------------------

% -------------------------------------------------------------
% -------------------------------------------------------------
% -------------------------------------------------------------
% -------------------------------------------------------------  
We show that $\Pallpsi$ is portable iff $\psi$ is satisfied for every assignment by proving the following: 
if $\Pallpsi$ is not portable then $\psi$ has an unsatisfying assignment and vice versa. 
%\input{complexity_tomove}
% !TEX spellcheck = en-EN
%!TEX root = sas2017.tex
\section{Experiments}
\label{sec:experiments}

\begin{table}[t]
\begin{minipage}[b]{0.415\textwidth}
%!TEX root = sas2017.tex
% !TEX spellcheck = en-EN

%\begin{table}[t]
\setlength{\tabcolsep}{-1.5pt}
\def\sep{\hspace{4pt}}
\footnotesize
\tt
\scalebox{.75}{
\begin{tabular}{lc@{\sep} cccccccccccccccccccc}
 \rm\small Benchmark
& $\rot{45} {SC\text{-}TSO}$
& $\rot{45} {SC\text{-}Power}$
& $\rot{45} {TSO\text{-}Power}$
\\
\midrule

\rm\bench{Bakery} &
\redcross & \redcross & \redcross \newrow
\rm\bench{Bakery x86} &
\tick & \redcross & \redcross \newrow
\rm\bench{Bakery Power} &
\tick & \tick & \tick \newrow
\rm\bench{Burns} &
\redcross & \redcross & \redcross \newrow
\rm\bench{Burns x86} &
\tick & \redcross & \redcross \newrow
\rm\bench{Burns Power} &
\tick & \tick & \tick \newrow
\rm\bench{Dekker} &
\redcross & \redcross & \redcross \newrow
\rm\bench{Dekker x86} &
\tick & \redcross & \redcross \newrow
\rm\bench{Dekker Power} &
\tick & \tick & \tick \newrow
\rm\bench{Lamport} &
\redcross & \redcross & \redcross \newrow
\rm\bench{Lamport x86} &
\tick & \redcross & \redcross \newrow
\rm\bench{Lamport Power} &
\tick & \tick & \tick \newrow
\rm\bench{Peterson} &
\redcross & \redcross & \redcross \newrow
\rm\bench{Peterson x86} &
\tick & \redcross & \redcross \newrow
\rm\bench{Peterson Power} &
\tick & \tick & \tick \newrow
\rm\bench{Szymanski} &
\redcross & \redcross & \redcross \newrow
\rm\bench{Szymanski x86} &
\tick & \redcross & \redcross \newrow
\rm\bench{Szymanski Power \ } &
\tick & \tick & \tick \newrow
\midrule
\\
\end{tabular}
}
\rm
%\caption{Bounded portability analysis of mutual exclusion algorithms: portable ({\tick}), non-portable ({\redcross})}
%\label{tab:mutual}
%\end{table}
\end{minipage}
\begin{minipage}[b]{0.4\textwidth}
%!TEX root = sas2017.tex
% !TEX spellcheck = en-EN

%\begin{table}[t]
\setlength{\tabcolsep}{2.5pt}
\small
\tt
\scalebox{.75}{
\begin{tabular}{l|c|cc|ccc}
\multicolumn{3}{c}{} & &  \multicolumn{2}{c}{Deadness} \\
\hline
& \color{red}{\ding{55}}\color{red}{\ding{55}}
& \color{mygreen}{\ding{52}}\color{mygreen}{\ding{52}}
& \color{red}{\ding{55}}\color{mygreen}{\ding{52}}
& \color{mygreen}{\ding{52}}\color{mygreen}{\ding{52}}
& \color{red}{\ding{55}}\color{mygreen}{\ding{52}}
\\
\hline
SC\text{-}TSO & 27 & 898 & 75 & 933 & 40 \newrow
SC\text{-}PSO & 27 & 777 & 196 & 836 & 137 \newrow
SC\text{-}RMO & 27 & 737 & 236 & 780 & 193 \newrow
SC\text{-}Alpha & 27 & 846 & 127 & 887 & 86 \newrow
TSO\text{-}PSO & 0 & 833 & 67 & 883 & 27 \newrow
TSO\text{-}RMO & 0 & 760 & 240 & 798 & 202 \newrow
TSO\text{-}Alpha & 0 & 877 & 133 & 912 & 88 \newrow
PSO\text{-}RMO & 0 & 831 & 169 & 844 & 156 \newrow
PSO\text{-}Alpha & 0 & 968 & 32 & 973 & 27 \newrow
RMO\text{-}Alpha & 0 & 999 & 1 & 999 & 1 \newrow
Alpha\text{-}RMO & 0 & 856 & 144 & 864 & 136 \newrow
\hline
\rowcolor{black!20}
&  0.98\% & 85.29 \% & 13.73 \% & 88.26 \% & 10.73 \% \newrow
\hline
SC\text{-}Power & 1477 & 898 & 52 & 936 & 14 \newrow
TSO\text{-}Power & 917 & 1132 & 378 & 1166 & 344 \newrow
PSO\text{-}Power & 502 & 1880 & 45 & 1892 & 33 \newrow
RMO\text{-}Power & 40 & 2227 & 160 & 2239 & 148 \newrow
Alpha\text{-}Power & 0 & 2427 & 0 & 2427 & 0 \newrow
\hline
\rowcolor{black!20}
&  24.20\% & 70.57\% & 5.23\%  & 71.35\% & 4.45\% \newrow 
\hline
\end{tabular}
}
\rm
\vspace{-26.2mm}
%\caption{Portability vs. State Portability on litmus tests.}
%\label{tab:st_vs_tr_sparc}
%\end{table}
\end{minipage}
\caption{{\bf (Left)} Bounded portability analysis of mutual exclusion algorithms: portable ({\tick}), non-portable ({\redcross}). {\bf (Right)} Portability vs. State Portability on litmus tests.}
\label{tab:portability}
\end{table}

The encoding from~\autoref{sec:encoding} has been implemented in a tool called \porthos. % which uses \zthree~\cite{DeMoura} as the back-end SMT solver. 
We evaluate \porthos on benchmark programs
% (litmus tests and mutual exclusion algorithms)
on a wide range of well-known MCMs. For SC, TSO, PSO, RMO and Alpha (henceforth called traditional architectures) we use the formalizations from~\cite{jade}; for Power the one in~\autoref{fig:power}.
We divide our results in three categories: portability of mutual exclusion algorithms, portability of litmus tests, and performance of the tool.

\paragraph{\bf Portability of Mutual Exclusion Algorithms.}
%\subsection{Portability of Mutual Exclusion Algorithms}
%\label{sec:mutex}

Most of the tools that are MCM-aware~\cite{AlglaveMT14,Mador-HaimAM10,TorlakVD10,memalloy,YangGLS04} accept only litmus tests as inputs. \porthos, however, can analyze cyclic programs with control flow branching and merging by unrolling them into acyclic form.
In order to show the broad applicability of our method, we tested portability of several mutual exclusion algorithms: Lamport's bakery~\cite{bakery}, Burns' protocol~\cite{BURNS1993171}, Dekker's~\cite{dekker}, Lamport's fast mutex~\cite{Lamport87}, Peterson's~\cite{Peterson81} and Szymanski's~\cite{Szymanski88}.
The benchmarks also include previously known fenced versions for TSO (marked as \rm\bench{x86}) and new versions we introduced using Power fences (marked as \rm\bench{Power}). The mutual exclusion program loops were unrolled once in all the experiments to obtain an acyclic program, and the discussion in what follows is for the portability analysis of this acyclic program.

While these algorithms have been proven correct for SC, it is well known that they do not guarantee mutual exclusion when ported to weaker architectures. The effects of relaxing the program order have been widely studied; there are techniques that even place fences automatically to guarantee portability, but they assume SC as the source architecture~\cite{AlglaveKNP14,BouajjaniDM13}. In~\autoref{tab:portability} (left) we do not only confirm that fenceless versions of the benchmarks are not portable from SC to TSO and fenced versions of them are, we also show that those fences are not enough to guarantee mutual exclusion when ported from TSO to Power.
%\autoref{tab:portability} shows that n
%New behaviors are introduced in several of the fenceless benchmarks for most of the architectures combinations (for lack of space we just present results on SC, TSO and Power in~\autoref{tab:portability}; complete results can be found in~\ref{sec:Aexperiments}). The few cases where no behaviors are added occur when porting from TSO to PSO or between RMO and Alpha (in both directions). Even though such models are different, such differences are subtle enough not to be exhibited for those particular programs at the analyzed bounded depth.
%Even though previously known fenced versions of the benchmarks guarantee mutual exclusion on TSO, that guarantee is lost when porting to Power.
We have used \porthos to find portability bugs when porting from TSO to Power and manually added fences to forbid such executions (see benchmarks marked as \rm\bench{Power}). To the best of our knowledge these are the first results about portability of mutual exclusion algorithms from memory models weaker than SC to the Power architecture. %Additionally we provide versions of those programs that guarantee mutual exclusion when run on a Power architecture (when ported from any other architecture). Since fences have been manually placed, we do not guarantee the minimality of the number of fences.

\paragraph{\bf Checking Portability on Litmus Tests.}
%\subsection{Checking Portability on Litmus Tests}
%\label{sec:litmus}

We compare the results of \porthos (which implements portability) against \herd (\url{http://diy.inria.fr/herd}) which reasons about state reachability and can be used to test state portability. 
\herd systematically constructs all consistent executions of the program and exhaustively enumerate all possible computable states.
Such enumeration can be very expensive for programs with lots of computable states, e.g. for many programs with a very large level of concurrency.
Since \herd only allows to reason about one memory model at a time, for each test we run the tool twice (one for each MCM) and compare the set of computable states. The program is not state portable if the target MCM generates computable states that are not computable states of the source MCM.

\begin{figure}[t]
\begin{minipage}{.729 \linewidth}
\begin{mybox}{}
%\begin{mybox}{\emph{Deadness}}
\footnotesize\centering
\rcounter{b:dead1}
$\circled{\thebulletscount}\ domain(\cdrel) \subseteq range(\rfrel)$
\rcounter{b:dead2}
$\circled{\thebulletscount}\ imm(\corel);imm(\corel);imm(\corel^{-1}) \subseteq \rfrel^?;(\porel;(\rfrel^{-1})^?)^?$
\end{mybox}
\end{minipage}
\centering
\footnotesize\centering
\begin{minipage}{.25 \linewidth}
\hspace{1mm}
$\mathit{imm}(r) \define r \setminus (r; r^{+})$
%$\mathit{r}^?:$ reflexive closure \\
\end{minipage}
\caption{Syntactic Deadness~\cite{memalloy}.}
\label{fig:dead}
\end{figure}

Our experiments contain two test suites: $\tsone$ contains 1000 randomly generated litmus tests in x86 assembly (to test traditional architectures) and $\tstwo$ contains 2427 litmus tests in Power assembly taken from~\cite{Mador-HaimMSMAOAMSW12}. Each test contains between 2 and 4 threads and between 4 and 20 instructions.
%For very small programs, such as the litmus tests, sometimes the new executions allowed by the target architecture do not create new computable states of the program. 
\autoref{tab:portability} (right) reports the number of non-portable (w.r.t.{\ }both definitions) litmus tests ({\color{red}\ding{55}\ding{55}}), the number of portable and state-portable litmus tests ({\color{mygreen}\ding{52}\ding{52}}) and the number of litmus tests that are not portable but are still state portable ({\color{red}\ding{55}}{\color{mygreen}\ding{52}}). In the last case the new executions allowed by the target memory model do not result in new computable states of the program. We show that in many cases both notions of portability coincide. 
For $\tsone$ on traditional architectures, the amount of non state-portable tests is very low (0.98\%), while the non portability of the program does not generate a new computable state in 13.73\% of the cases. For $\tstwo$ from traditional architectures to Power, the number of non state-portable litmus tests rises to 24.20\%, while only in 5.24\% of cases the two notions of portability do not match because the new executions do not result in a new computable state for the program.

%\paragraph{Portability vs State Portability.}

In order to remove some executions that do not lead to new computable states, \porthos\ optionally supports the use of syntactic deadness which has been recently proposed in~\cite{memalloy}. 
Dead executions are either consistent or lead to not computable states. 
Formally an execution $\exec$ is dead if $X \not \in \consistent \mm P$ implies that ${\reach X \not = \reach Y}$ for all ${Y \in \consistent \mm P}$.
Instead of looking for any execution which is not consistent for the source architecture, 
we want to restrict the search to non-consistent and dead executions of $\msource$. 
This is equivalent to checking state portability. 
As shown by Wickerson et al. \cite{memalloy}, dead executions can be approximated with constraints 
\circled{\ref{b:dead1}} and \circled{\ref{b:dead2}} given in~\autoref{fig:dead} where $\mathit{r}^?$ is the reflexive closure of $\mathit{r}$. 
These constraints can be easily encoded into SAT. Our tool has an implementation which rules out quite a few executions not computing new states. The last two columns of the table show that by restricting the search to (syntactic) dead executions, the ratio of  litmus tests the tool reports as non portable, but are actually state portable is reduced to 10.73\% for traditional architectures and to 4.44\% for Power.

The experiments above show that in most of the cases both notions of portability coincide, specially when using dead executions. The cases where such differences are manifested is very low, specially when porting to Power. To test state portability, our method can be complemented with an extra query to check if the final state of the counter-example execution is also reachable in the source model by another execution. However, as shown in~\autoref{sec:complexity}, the price to obtain such a result is to go one level higher in the polynomial hierarchy which affects the performance of the analyses.

\paragraph{\bf Performance.}
%\subsection{Performance}
%\label{sec:performance}
%We evaluate the performance of our tool and the factors that affect it. For the mutual exclusion algorithms we show the time it takes the solver to find an answer. Since a MCM can be defined in many equivalent ways, we also study the impact on the choicing what formalization to use. Finally we show how much overhead is added by encoding the deadness constraints to remove some false positives.
%
%\paragraph{Mutual Exclusion Algorithms.}

\begin{figure}[t]
\centering
\scalebox{.8}{%!TEX root = ../sas2017.tex

\begin{tikzpicture}
\begin{axis}[
	symbolic x coords={
	Bakery, Bakery x86, Bakery Power, 
	Burns, Burns x86, Burns Power,
	Dekker, Dekker x86, Dekker Power,
	Lamport, Lamport x86,Lamport Power,
	Peterson, Peterson x86, Peterson Power,
	Szymanski, Szymanski x86, Szymanski Power,
	a
	},
	xtick=data,
	xtick style={draw=none},
	ytick style={draw=none},
	ytick={0,5,10,...,35,40},
	x tick label style={rotate=45, anchor=east, font=\small},
	enlargelimits=0.025,
	legend style={at={(0.15,.9)},
		anchor=north,legend columns=1,
	legend entries={SC-TSO, SC-Power, TSO-Power}, font=\footnotesize},
	ybar interval=0.7,
	width=\textwidth,
	height=6cm
]
\addplot [fill=blue]
	coordinates {
(Bakery, 0.21112990379333496)
(Bakery x86, 0.2102651596069336)
(Bakery Power, 0.21029901504516602)
(Burns, 0.05124306678771973)
(Burns x86, 0.04942202568054199)
(Burns Power, 0.050054073333740234)
(Dekker, 0.16898703575134277)
(Dekker x86, 0.16709613800048828)
(Dekker Power, 0.1702289581298828)
(Lamport, 0.4890270233154297)
(Lamport x86, 0.759335994720459)
(Lamport Power, 0.7899990081787109)
%(Parker, 0.01889801025390625)
%(Parker Power, 0.019166946411132812)
%(Parker x86, 0.018632173538208008)
(Peterson, 0.10684895515441895)
(Peterson x86, 0.10451698303222656)
(Peterson Power, 0.10692906379699707)
(Szymanski, 0.5494050979614258)
(Szymanski x86, 0.5911979675292969)
(Szymanski Power, 0.6410701274871826)
(a, 0)
};
\addplot [fill=white]
	coordinates {
(Bakery, 1.718883991241455)
(Bakery x86, 2.0373098850250244)
(Bakery Power, 2.0865390300750732)
(Burns, 0.2777390480041504)
(Burns x86, 0.3118760585784912)
(Burns Power, 0.3193979263305664)
(Dekker, 1.4639639854431152)
(Dekker x86, 1.7945969104766846)
(Dekker Power, 1.7987449169158936)
(Lamport, 4.230733871459961)
(Lamport x86, 5.859885931015015)
(Lamport Power, 22.798483848571777)
%(Parker, 0.07511496543884277)
%(Parker Power, 0.08404111862182617)
%(Parker x86, 0.08121991157531738)
(Peterson, 0.5806450843811035)
(Peterson x86, 0.74149489402771)
(Peterson Power, 0.7837228775024414)
(Szymanski, 11.447784185409546)
(Szymanski x86, 10.17952585220337)
(Szymanski Power, 14.220797061920166)
(a, 0)
};
\addplot [fill=red]
	coordinates {
(Bakery, 1.669646978378296)
(Bakery x86, 2.095571994781494)
(Bakery Power, 2.1380279064178467)
(Burns, 0.2794959545135498)
(Burns x86, 0.3151819705963135)
(Burns Power, 0.328110933303833)
(Dekker, 1.4670288562774658)
(Dekker x86, 1.8313229084014893)
(Dekker Power, 1.8584771156311035)
(Lamport, 4.643279790878296)
(Lamport x86, 6.377609968185425)
(Lamport Power, 33.291738986968994)
%(Parker, 0.076995849609375)
%(Parker Power, 0.08423495292663574)
%(Parker x86, 0.0838158130645752)
(Peterson, 0.6031379699707031)
(Peterson x86, 0.7748010158538818)
(Peterson Power, 0.8080220222473145)
(Szymanski, 11.43646502494812)
(Szymanski x86, 9.99736499786377)
(Szymanski Power, 16.323559999465942)
(a, 0)
};

\end{axis}
\end{tikzpicture}}
  \caption{Solving times (in secs.) for portability of mutual exclusion algorithms.}
  \label{fig:solving}
\end{figure}
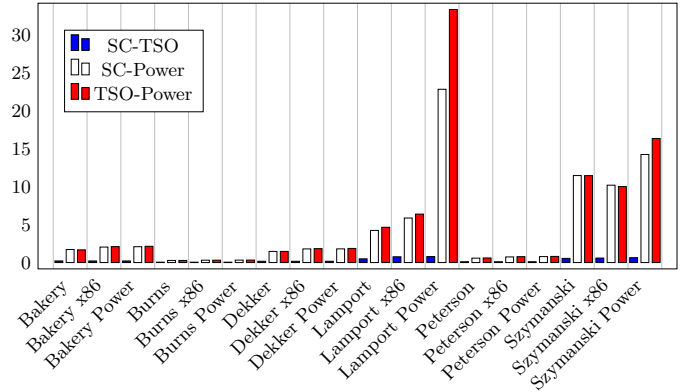

For small litmus test, the running times of \herd outperform \porthos. 
However, as soon as the programs become bigger, \herd does not perform as well as \porthos. We believe this is due to the use of efficient search techniques in the SMT solver.
The impact on efficiency is manifested as the number of executions \herd has to explicitly simulate by enumeration grows.
We evaluate the solving times of our tool on the mutual exclusion benchmarks.
Our prototype encoding implementation is done in Python; the encoding generations times have a minimum of 13~secs and a maximum of 303~secs. The encodings involving Power are usually more time consuming than traditional models since Power has both transitive closures and least fixed points in its encoding. We expect that the encoding times could be vastly improved by a careful C/C++ implementation of the encoding.
\autoref{fig:solving} presents the solving times of \porthos for the mutual exclusion algorithms, which are actually much lower than the encoding times for our prototype implementation.
%For most of the cases it takes less than 10 secs; exceptions are \rm\bench{Lamport Power} and all \rm\bench{Szymanski} versions when testing portability to Power. For most of such cases the programs are portable which correspond to the encoding being UNSAT, which typically is harder for SAT/SMT solvers than finding a satisfying truth assignment. See for example \rm\bench{Szymanski x86} for RMO-Power and Alpha-Power which takes more time than the other models combinations. Solving times for traditional models are faster than those of Power by one order of magnitude.

%In terms of encoding times, it can be observed that for any benchmark, the most time consuming cases are those related to the Alpha model. This is due to the the fact that Alpha requires an extract acyclicity constraint to avoid out-of-the-air values. It can be also observed that the times on the left of \autoref{fig:encoding} (corresponding to traditional architectures) are almost doubled when reasoning about Power (right). This is due to the fact that the Power model requires several transitive closures and compositions which have a cubic complexity and thus they increase the computing times.
%!TEX root = sas2017.tex
% !TEX spellcheck = en-EN

\section{Related Work}

Semantics and verification under weak memory models have been the subject of study at least since 2007.
Initially, the behavior of x86 and TSO has been clarified~\cite{Checkfence2007,SarkarSNORBMA09}, then the Power architecture has been addressed~\cite{Mador-HaimMSMAOAMSW12,SarkarSAMW11}, now ARM is being tackled~\cite{FlurGPSSMDS16}. The study also looks beyond hardware, in particular C++11 received considerable attention~\cite{BattyDW16,BattyOSSW11}. 
Research in semantics goes hand in hand with the development of verification methods. 
They come in two flavors: program logics~\cite{TuronVD14,VafeiadisN13} and algorithmic approaches~\cite{AbdullaAAJLS15,AbdullaAJL16,AlglaveKT13,AlglaveMT14,AtigBBM10,BouajjaniDM13,Burckhardt2008,DanMVY13,DanMVY15}. 
Notably, each of these methods and tools is designed for a specific memory model and hence is not directly able to handle porting tasks.

The problem of verifying consistency under weak memory models has been extensively studied. 
Multiple methods and the complexity of the corresponding problems have been analyzed~\cite{CantinLS05,EneaF16,FarzanM08}. 
A prominent approach is testing where an execution is (partially) given and consistency is tested for a specified model \cite{furbachmeyerschneidersenftleben2015,Gibbons97testingshared}. 
In this line we showed that state portability (formulated as a bounded analysis for cyclic programs) is $\Pi_2^p$-complete.
This means there is no hope for a polynomial encoding into SAT (unless the polynomial hierarchy collapses). 
In contrast, our execution-based notion of portability is co-NP-complete (we look for a violation to portability), which in particular means that our portability analysis is optimal in the complexity sense.
%Concerning \emph{(b)}, we evaluate our tool on a set of benchmarks and compare the results against state portability.
%We run \herd twice (once for each model) and compare the number of computable states to test state portability. This however will be infeasible for programs with a very large number of computable states.
Our experiments show that in most of the cases both notions of portability coincide.

A problem less general than portability is solved in~\cite{BouajjaniDM13} where non-portable traces from SC to TSO are characterized. 
The problem is reduced to state reachability under the SC semantics in an instrumented program and a minimal number of fences is synthesized to enforce portability.
One step further, one can enforce portability not only to TSO, but also to weaker memory models. 
The \offence tool~\cite{AlglaveM11} does this, but can only analyze litmus test and is limited to restoring SC. 
Checking the existence of critical cycles (i.e. portability bugs) on complex programs has been tacked in~\cite{AlglaveKNP14}, where such cycles are broken by automatically introducing fences. 
The cost of different types of fences is considered and the task is encoded as an optimization problem. 
The \musketeer tool analyzes C programs and has shown to scale up to programs with thousands lines of code, but the implementation is also restricted to the case were the source model is SC. 
Fence insertion can also be used to guarantee safety properties (rather than restoring SC behaviors). The \fender and \dfence tools~\cite{KupersteinVY12,LiuNPVY12} can verify real-world C code, but they are restricted to TSO, PSO, and RMO. 

%More recently, a method that synthesizes programs differentiating two memory models has been proposed~\cite{memalloy}. 
%The method reasons about \emph{any} program while the portability analysis in this paper focuses on a \emph{particular} program: the differences between memory models are not important as far as they do not impact the program of interest. 
%%Their implementation is the only hardware-architecture-aware tool we are aware of besides \herd and \porthos. 
% !TEX spellcheck = en-EN
%!TEX root = sas2017.tex

\section{Conclusion and Outlook}
\label{sec:conclusion}

We introduce the first method that tests portability between any two axiomatic memory models defined in the CAT language.
The method reduces portability analysis to satisfiability of an SMT formula in SAT + integer difference logic.
We propose efficient solutions for two crucial tasks: reasoning about two user-defined MCMs at the same time and encoding recursively defined relations (needed for Power) into SMT. The latter can be re-used by any bounded model checking technique reasoning about complex memory models.

 %This seems to limit the applications of our approach since the compilation of a high level program already depends on the memory model. Different architectures may use different atomic memory operations.\flo{todo: ref zappa} A C11 program for example may be mapped to two different LLVM programs for $\mtarget$ and $\msource$. However, one can still easily identify the corresponding reads and writes in the different low-level programs and thus an execution is valid for both programs. Our method can still be applied here and \porthos could be extended by an additional layer that translates an input program into two corresponding lower level input programs. 
%Anecdotal evidence suggests that common mappings exist for many architectures.\flo{Todo: improve, maybe move}

Our complexity analysis and experimental results both suggest that our definition of portability is preferable over the state-based notion of portability. If state-based portability is required, the complexity results show that it cannot be done with a single SMT solver query, unlike the approach to portability analysis suggested in this paper.
We also show that our method is not restricted just to litmus tests, % that have no branching structure 
but actually for the first time report on automated tool-based portability analysis of mutual exclusions algorithms from several axiomatic memory models to Power.

\section*{Acknowledgements}

We thank John Wickerson for his explanations about dead executions, Luc Maranget for several discussions about CAT models and Egor Derevenetc for providing help with the mutual exclusion benchmarks. This work was partially supported by the Academy of Finland project 277522. Florian Furbach was supported by the DFG project R2M2: Robustness against Relaxed Memory Models.

%\clearpage

\bibliographystyle{plain}
\bibliography{bib}

\clearpage
\appendix
% !TEX spellcheck = en-GB
% !TEX root = sas2017.tex

\section{Rest of the encoding}
\label{sec:app-encod}

This section details the remaining two sub formulas for the portability encoding, i.e. the control-flow and the data-flow.

\subsection{Control-flow}
\label{sec:cf}

Instead of representing the branching of the program with a tree~\cite{CruzFP12}, we use a direct acyclic graph (DAG) capturing the branches of the program and how those merge again. This allows to keep the size of the control-flow formula linear w.r.t the (unfolded) program. The tree representation can be exponential if the program has several if statements. %However, with this representation, it is necessary to encode how the branches merge again, i.e. what values are assigned to variables after an if statement. 
We encode this DAG in the formula $\phi_{\mathit{CF}}$.

For each instruction\footnote{Notice that instructions are defined recursively and thus the term ``instruction'' may represent a sequence of events accessing the memory.} $i$ we use a Boolean variable $\mathit{cf_i}$ representing the fact that the instruction is actually executed by the execution. For a sequence $i_1 \define i_2 ; i_3$, instruction $i_1$ belongs to the execution iff both $i_2$ and $i_3$ belong too~(\ref{eq:cf1}). Assignments (local computations, loads and stores) and fences do not impose any restriction in the control-flow encoding~(\ref{eq:cf2})-(\ref{eq:cf5}); belonging or not to the execution depends on them being part of the body of some if statement at a higher level of the recursive definition. Given an instruction $i_1 \define \texttt{if}\ b\ \texttt{then}\ i_2\ \texttt{else}\ i_3$, we use three control-flow variables $\mathit{cf_{i_1}, cf_{i_2}, cf_{i_3}}$; then $i_1$ is executed iff one of $i_2, i_3$ is performed~(\ref{eq:cf6}), which one actually depends on the value of $b$ and this is encoded in the data-flow formula $\phi_{DF}$. These restrictions are encoded recursively by the following constraints:
\arraycolsep=1.5pt
\small
\begin{eqnarray}
\phi_{\mathit{CF}}(i_2 ; i_3) & = & \mathit{cf_{i_1}} \Leftrightarrow (\mathit{cf_{i_2}} \land \mathit{cf_{i_3}}) \label{eq:cf1} \land\ \phi_{\mathit{CF}}(i_2) \land \phi_{\mathit{CF}}(i_3) \\
\phi_{\mathit{CF}}(r \leftarrow e) & = & \mathit{cf_{r \leftarrow e}} \label{eq:cf2} \\
\phi_{\mathit{CF}}(r \leftarrow l) & = & \mathit{cf_{r \leftarrow l}} \label{eq:cf3} \\
\phi_{\mathit{CF}}(l \define r) & = & \mathit{cf_{l \define r}} \label{eq:cf4} \\
\phi_{\mathit{CF}}(\texttt{fence}) & = & \mathit{cf_{\texttt{fence}}} \label{eq:cf5} \\
\phi_{\mathit{CF}}(\texttt{if}\ b\ \texttt{then}\ i_2\ \texttt{else}\ i_3) & = & \mathit{cf_{i_1}} \Leftrightarrow {(\mathit{cf_{i_2}} \lor \mathit{cf_{i_3}})} \label{eq:cf6} \land\ \phi_{\mathit{CF}}(i_2) \land \phi_{\mathit{CF}}(i_3)
\end{eqnarray}

\subsection{Data-flow}
\label{sec:dataflow}

We encode the data flow with single static assignments using the method of~\cite{CollavizzaR06}. Formula $\mathit{\phi_{DF}}$ represents how the data flows between locations and registers; we first focus on how the data-flow of the local thread behavior is encoded (sub-formula $\dfloc$). For each location of the program (resp. register) we use several integer variables (one for each variable in the SSA form of the program) representing the value carried by that location (resp. register) in the execution. For loads, stores and local computations, if the instruction is part of the execution (i.e. its control-flow variable is True) then both sides of the assignment should coincide~(\ref{eq:1})-(\ref{eq:3}). For a sequence, the formula is the conjunction of the encoding of the corresponding instructions~(\ref{eq:4}).

Suppose register $r$ and location $l$ have been already assigned $p$ and $q$ times respectively, then:
\small
\begin{eqnarray}
\dfloc(r \leftarrow e) & = & cf_{r \leftarrow e} \imp (r_{p+1} = e) \label{eq:1} \\
\dfloc(r \leftarrow l) & = & cf_{r \leftarrow l} \imp (r_{p+1} = l_{q+1}) \label{eq:2} \\
\dfloc(l \define r) & = & cf_{l \define r} \imp (l_{q+1} = r_p) \label{eq:3} \\
\dfloc(i_1 ; i_2) & = & \dfloc(i_1) \land \dfloc(i_2) \label{eq:4}
\end{eqnarray}

Following the SSA form, the left hand side of each assignment introduces new variables; for registers in the right hand side, indexes are not updated so they match with the last value which can only be modified by the same thread~(\ref{eq:3}). However for locations in the right hand side, the index is also updated~(\ref{eq:2}) to allow variables to match not only with the last assignment done by that thread, but also from other threads (see the sub-formula $\phi_{\mathit{DF_{mem}}}$ below).

If statements may have a different number of assignments in their branches for certain variables. The idea here is to insert dummy assignments to ensure that both branches have the same number of assignments. We show the encoding for the simple case where each branch consists only of local computations to a register $r$. The same process is applied individually for each register and location assigned in a branch. If the branches contain if statements, the procedure must be applied recursively to each of them. Consider the if statement $\texttt{if}\ b\ \texttt{then}\ i_1\ \texttt{else}\ i_2$ where the first branch has less assignments to $r$ than the second one, i.e. $i_1 \define r \leftarrow e_{1,1}; \dots; r \leftarrow e_{1,p}$ and $i_2 \define r \leftarrow e_{2,1}; \dots; r \leftarrow e_{2,q}$ with $p<q$ (the encoding is symmetric for $q<p$). Assume $r$ has been already assigned $x$ times, then the encoding of the instruction contains the following constraint:
%\[
\arraycolsep=1pt
\small
\begin{eqnarray}%{rcl}
\hspace{-6mm}
\dfloc(\texttt{if}\ b\ \texttt{then}\ i_1\ \texttt{else}\ i_2) & = & (b \imp cf_{i_1}) \land (\neg b \imp cf_{i_2}) \label{eq:5} \\
& & \land\ cf_{i_1} \imp (r_{x+1} = e_{1,1}) \label{eq:6}\\
& & \hspace{17mm} \vdots \nonumber \\
& & \land\ cf_{i_1} \imp (r_{x+p} = e_{1,p}) \label{eq:7}\\
& & \hspace{17mm} \vdots \nonumber \\
& & \land\ cf_{i_1} \imp (r_{x+q} = r_{x+p}) \label{eq:8}\\
& & \land\ cf_{i_2} \imp (r_{x+1} = e_{2,1}) \label{eq:9}\\
& & \hspace{17mm} \vdots \nonumber \\
& & \land\ cf_{i_2} \imp (r_{x+q} = e_{2,q}) \label{eq:10}
%& & \land\ \dfloc(i_1) \land \dfloc(i_2) \label{eq:11}
\end{eqnarray}
%\]
Constraint~(\ref{eq:5}) imposes that which branch is followed depends on the value of the predicate. Next, we specify how the value of $r$ is updated depending on the branch: if the first branch is taken, then the value of $r$ is updated according to the expressions the first $p$ times~(\ref{eq:6})-(\ref{eq:7}) and it remains unchanged for the remaining $q-p$ assignments~(\ref{eq:8}); if the second branch is taken, the value of $r$ is updated according to the corresponding expressions in that branch~(\ref{eq:9})-(\ref{eq:10}). 
By adding constraints which keep assignments unchanged, we can easily model how branches merge again since any variable assigned after the if statement would be matched with the last value assigned to that variable.

Since fresh variables are added for locations in both sides of the assignments~(\ref{eq:2}), their values are not yet constrained. 
%Finally the formula is recursively applied to the instructions in each branch~(\ref{eq:11}). 
We now specify how the data flows between instructions that access locations in the shared memory (possibly in different threads). This depends on where the values are read-from (i.e. the $\rfrel$ relation) and is encoded by constraints $\mathit{DF_{mem}}(i_1, i_2)$. 
A write instruction $l \define r$ generates data-flow constraint $cf_{l \define r} \imp (l_i = r)$ and a read $r \leftarrow l$ is encoded by $cf_{r \leftarrow l} \imp (r = l_j)$. 
The variables $l_i, l_j$ remain unconstrained. 
If both instructions (call them $i_1$ and $i_2$) are related over the $\rfrel$ relation, then their values need to match: $${\phi_{\mathit{DF_{mem}}}(i_1, i_2)} = {\texttt{rf}(i_1,i_2) \imp l_i = l_j}$$

%Suppose two instructions $l \define r$ and  $r \leftarrow l$ generate the data-flow constraints $cf_{l \define e} \imp (l_i = e)$ and  $cf_{r \leftarrow l} \imp (r = l_j)$ where the variables $l_i, l_j$ remain unconstrained. If both instructions (call them $i_1$ and $i_2$) are related over the $\rfrel$ relation, then their values should match: $${\phi_{\mathit{DF_{mem}}}(i_1, i_2)} = {\texttt{rf}(i_1,i_2) \imp l_i = l_j}$$

Finally, the data-flow between register and location either within or between threads is encoded as: $${\mathit{\phi_{DF}}} = {{\hspace{-6mm} \bigwedge\limits_{(i_1, i_2) \in {(\stores \times \loads)} \cap sloc} \hspace{-6mm} \phi_{\mathit{DF_{mem}}}(i_1, i_2)} \wedge {\bigwedge\limits_{t \in \threads} \dfloc(i_t)}}$$ where $i_t$ represent the instruction at the highest recursive level of the thread.

\section{Complexity Proofs}
\label{sec:app-comp}

We recall the main theorem and the program $\Ppsi \define  t_1\parallel t_2$ from the paper:%\autoref{sec:complexity}:

\tracecomplete*
%\begin{theorem}\label{thm:TraceComplete}
%Let $\msource$ and $\mtarget$ be a non-trivial pair of common MCMs.
%(1) Portability from $\msource$ to $\mtarget$ is $\Pi^P_1$-complete. 
%(2) State portability is $\Pi^P_2$-complete. 
%\end{theorem}

		\begin{algorithm}\label{thread:t1}
			\NoCaptionOfAlgo
			 $r_{c,0} \leftarrow 0;\ r_{c,1}\leftarrow 1;\ r_{c,2} \leftarrow 2$\;  \label{line:rci}
			 $x_1 \define r_{c,0}; \dots; x_m\define r_{c,0}$ \tcp*{Writes $w_{1,0} \dots w_{m,0}$}
			 $r_1 \leftarrow x_1; \dots; r_m \leftarrow x_m$ \label{line:rx}\tcp*{Reads $r_{1} \dots r_{m}$}
			 \eIf(\tcp*[f]{If $A$ satisfies $\psi$,}) {$\psi(r_1, \dots, r_m)$}{\label{line:readpsi}
			 		$y\define r_{c,2}$ \tcp*{return 2.}
			 }{
			 		$y\define r_{c,1}$ \tcp*{If it doesn't, return 1.}
			 }
 
			\caption{$t_1$}  
		\end{algorithm}

			\NoCaptionOfAlgo

		\begin{algorithm}\label{thread:t2}
			$r_{c,1} \leftarrow 1$\;
			 $x_1 \define r_{c,1};\ \ldots x_m\define r_{c,1}$ \tcp*{Writes $	w_{1,1} \dots w_{m,1}$}
			\caption{$t_2$}  
			\RestoreCaptionOfAlgo
		\end{algorithm}

We give some technical results and show $\Pi^P_1$-Completeness of Portability for common MCMs.

\begin{lemma}\label{lem:psiportable}
$P_\psi$ is portable from every common MCM to another common MCM.
\end{lemma}
\begin{proof}
According to property (iv), any common MCM is portable to SC.
In SC, an execution corresponds to an interleaving of the two thread executions. 
First, the threads create some variable assignment $A$ which is read by $t_1$.
Then, $t_1$ checks whether the assignment satisfies $\psi$. 
If it does, $y$ is set to 2, otherwise $y$ is set to 1. 

We show that any consistent execution of some common MCM is SC-consistent by examining possible executions. 
\begin{itemize}
	\item If $w_{i,1} \co w_{i,0}$  and $w_{i,0} \rf r_i$, then this corresponds to an interleaving where $w_{i,1}$ occurs first, then $t_1$ writes $w_{i,0}$ and reads 0 ($r_i$).
	\item If $w_{i,0} \co w_{i,1}$  and $w_{i,0} \rf r_i$, then $w_{i,0}$ and $r_i$ in $t_1$ occur first and then $w_{i,1}$.
	\item If $w_{i,0} \co w_{i,1}$  and $w_{i,1} \rf r_i$, then $t_1$ writes $w_{i,0}$, $t_2$ overwrites this with $w_{i,1}$ and afterward $t_1$ reads 1 with $r_i$.	
	\item If $w_{i,1} \co w_{i,0}$  and $w_{i,1} \rf r_i$, then the derived relation ${\frrel} \define {\rfrel^{-1};\corel}$ satisfies $r_i \fr w_{i,0}$. 
	Since $w_{i,0}$ and $r_i$ are related by $\porel$ and access the same location there is a cycle $w_{i,0} \poloc r_i \fr w_{i,0}$. 
  This is a violation of uniproc, the situation can not occur in a common MCM.
\end{itemize}
So we can construct a corresponding interleaving for any execution of a common MCM.
It follows that every execution of a common MCM is SC-consistent and according to property (iv) consistent with any common MCM. 
\end{proof}

%We show with the following lemma, that the consistent executions of $P_\psi$ remain the same under any model with uniproc.
%An assignment $A$ of a set of Boolean variables $V$ is a set of tuples $A\subset V\times \{ 0,1 \}$. 
%So $A\subset \state $ holds iff \mem\ contains assignments for $V$  and they are the same as $A$.
%This means it is portable for all models $\msource$ and $\mtarget$ with \uniproc.
%To prove this, we will first analyze the complexity of state portability.

We use the following technical lemmas to show hardness. 
We call the relations $\porel, \rfrel, \corel, \adrel, \ddrel, \cdrel$ and $ \frrel$ basic.
Given a common MCM, we define the \emph{violating cycles} as follows:
For an assertion $\acy{r}$, any cycle of $r$ is violating.
For an assertion $\irref{r}$, any cycle of the form $e \stackrel{r}{\rightarrow} e$ is violating.

The following lemma shows that an execution is not consistent if it contains a violating cycle. 
\begin{lemma}\label{lem:cycles}
Let $\mm$ be common. An execution $\exec$ is consistent with $\mm$ iff $\exec$ contains no violating cycle of $\mm$.
\end{lemma}
This follows directly from the definiton of violating cycles.
%\begin{proof}
%An execution is not consistent iff it violates an acyclicity assertion or an irreflexivity assertion.
%An assertion $\acy{r}$ is violated iff there is a cycle of $r$. 
%Assertion $\irref{r}$ is violated iff there is a cycle $e \stackrel{r}{\rightarrow } e$.
%%So the violating cycles contain all cycles of $r$.
%% for some event $e$ was generated for some relation and added to $r$ in steps of the Kleene Iteration.
%%If a pair $(e,e)$ was generated by a transitive closure, then it is always generated and added to $r$. Then irreflexivity is trivially false. 
%%This is not the case for common MCMs.
%%If the pair was generated by operators $r_1^+$ or $r_1;r_2$ then the pair is generated by the presence of a path from $e$ to $e$.
%%This is a violating cycle. 
%\end{proof}

We say a relation $r$ satisfies the path condition if $e_1 \stackrel{r}{\rightarrow } e_2$ implies $e_1 {\stackrel{b}{\longrightarrow^* }} e_2$ with $b\define po\cup rf \cup co \cup ad \cup dd \cup cd \cup fr$.
\begin{lemma}\label{lem:common}
Any relation $r$ of a common MCM satisfies the path condition. 
\end{lemma}

\begin{proof}
Note that the recursively defined relations of a common MCM can be obtained with a Kleene iteration.
We use a structural induction over the Kleene iteration.

Induction Basis: Any named relation is initially the empty relation and trivially satisfies the path condition. 

Induction Step: Assume all named relations satisfy the path condition.
A named relation is updated according to its defining equation using the current assignments of the named relations and basic relations. 
To simplify this proof, we can assume that any equation only contains one operator or a relation in $base$. 
Any basic relation trivially satisfies the condition. 
Let $r_1$ and $r_2$ satisfy the path condition. We examine the operations that can be applied to them according to properties (i) and (ii) of common MCMs.
We see that $r_1 \cup r_2, r_1 \cap r_2, r_1 \cap sloc, r_1 \cap sthd, r_1 \cap {\bnev} \times {\bnev}, r_1 \setminus r_2$ all satisfy the path condition. 
The relation $r_ 1;r_2 (e_1,e_2)$ requires an event $e_3$ with $r_1(e_1,e_3)$ and $r_2(e_3,e_2)$ and since $r_1$ and $r_2$ satisfy the path condition there is a path from $e_1$ to $e_3$ and from $e_3$ to $e_2$ and thus $r_ 1;r_2 $ also satisfies the path condition. 
Similarly, $r_1^+$ adds a relation only where there already is a path and satisfies the path condition.
Relation $r_1^*$ consists of $r^+$ and the identity relation, which also satisfies the path condition (there is always a path of length 0 from some $e_1$ to $e_1$).
A named relation still satisfies the path condition after a Kleene iteration step.
\end{proof}

Note that according to Lemma~\ref{lem:common}, a violating cycle implies a cycle of basic relations (called a basic path) that contains all events of the violating cycle.
%We can check consistency of an execution by examining cycles. 
We can argue about consistency of an execution by examining the paths of basic relations.
\begin{lemma}\label{lem:paths}
Given a relation $r$ of a common MCM and events $e_1,e_2$ of an execution, whether $r(e_1,e_2)$ holds is determined by the basic paths from $e_1$ to $e_2$.
\end{lemma}
\begin{proof}
We use a structural induction over the Kleene iteration.

Induction Basis: Any named relation is initially the empty relation and trivially satisfies the condition. 
Any basic relation trivially satisfies the condition. 

Induction Step: Assume all current assignments of relations are determined by the basic paths between related events. 
Two events $e_1$ and $e_2$ are related by some relation $r$ iff the basic paths from $e_1$ to $e_2$ satisfy some property.
A named relation is updated according to its defining equation using the current assignments of the named relations and basic relations. 
Let $r_1$, $r_2$ be either named or basic relations.
%Let events be related by  if the basic paths between the events have some property.
Events $e_1,e_2$ are related by $r_1 \cup r_2$ ($r_1 \cap r_2$) if the basic paths from $e_1$ to $e_2$ satisfy the property of $r_1$ or $r_2$ (resp. $r_1$ and $r_2$).
Similar, events are related by $ r_1 \setminus r_2$ if the basic paths satisfy the property of $r_1$ but not $r_2$. 
For $ r_1 \cap sloc, r_1 \cap sthd, r_1 \cap {\bnev} \times {\bnev}$ two events $e_1,e_2$ are related if the property of $r_1$ is satisfied and $e_1$ and $e_2$ satisfy an additional condition. 
The condition for the events $e_1,e_2$ can be expressed as a additional condition for paths from $e_1$ to $e_2$.
 
The relation $r_ 1;r_2$ relates $e_1$ to $e_2$ if there is an $e_3$ such that the basic paths from $e_1$ to $e_3$ ensure $r_1(e_1,e_3)$ and the  basic paths from $e_3$ to $e_2$ ensure $r_1(e_3,e_2)$.
It follows that $r_ 1;r_2(e_1,e_3)$ depends on the basic paths from $e_1$ to $e_2$ over some $e_3$. 
The relations $r_1^*$ and $ r_1^+$ are derived using previously examined operators over relations and thus they satisfy the condition.

%For , we add the identity relation which also satisfies the path condition (there is always a path of length 0 from some $e_1$ to $e_1$).
A named relation still satisfies the path condition after a Kleene iteration step.
\end{proof}

\begin{proof}[Proof of Theorem 1.1]%\autoref{thm:TraceComplete} (1)]
We show that $\Pallpsi$ is not portable iff $\psi$ has an unsatisfying assignment. 

($\Rightarrow$): We assume an execution $\exec$ exists that is $\mtarget$-consistent but not $\msource$-consistent.  
According to Lemma~\ref{lem:cycles}, the execution has a violating cycle for $\msource$. 
We assume towards contradiction that no event of $\Pnp$ is executed. 
%Since $P_\psi$ is $\mtarget$ consistent and always portable (Lemma~\ref{lem:psiportable}), the violating cycle does not only contain events from $P_\psi$.
The read $r \leftarrow y$ in $t''_1$ has to read from the write in $t_1$ (in $\Ppsi$) according to uniproc (the execution is $\mtarget$-consistent).
It occurs after $t_1$ in the program order. The read has incoming basic relations from events in $\Ppsi$ but no outgoing relations to some event in $\Ppsi$.
Any read $r \leftarrow y$ in another thread can either read from the write in $\Ppsi$ (it has an incoming $rf$ relation from $\Ppsi$) or it reads the initial value which results in an outgoing from-read relation to $\Ppsi$. There are no other basic relations between a read $r \leftarrow y$ and some event in $\Ppsi$.
So any read $r \leftarrow y$ has either basic incoming relations from $\Ppsi$ ($\rfrel, \porel$) or outgoing to $\Ppsi$ ($fr$), not both. 

It follows from Lemma~\ref{lem:common}, that no read $r \leftarrow y$ has both incoming and outgoing derived relations and so no read $r \leftarrow y$ is in a violating cycle. 
Any violating cycle is in $\Ppsi$.
Further, there is no basic path from some event $e_1$ in $\Ppsi$ to one of the reads and back to some $e_2$ in $\Ppsi$. 
The reads do not affect the basic paths between events in $\Ppsi$. 
So the violating cycle for $\msource$ is still present if we remove the reads and thus restrict the execution to events of $\Ppsi$. 
Since removing the reads does not affect the basic paths in $\Ppsi$, no violating cycle for $\mtarget$ has been added. 
It follows that the execution of $\Ppsi$ is consistent with $\mtarget$ but not $\msource$. 
This is a contradiction to $\Ppsi$ being always portable for common MCMs (Lemma~\ref{lem:psiportable}).

It follows that any violating cycle requires events from $\Pnp$ to be executed.
Since an event of $\Pnp$ can only be executed if $y=1$ was read in the thread (and thus written by $t_1$), the if-condition in Line~\ref{line:readpsi} was not satisfied.
It follows that there is an assignment that does not satisfy $\psi$.

($\Leftarrow$): We now assume that there is an assignment that doesn't satisfy $\psi$.
There is an SC-consistent execution $Z$ of $\Ppsi$ that executes the write $y\define r_{c,1}$. 
We extend this execution of $\Ppsi$ to an execution $\exec$ of $\Pallpsi$: 
We ensure that all reads $r\leftarrow y$ read from $y\define r_{c,1}$ and thus $\Pnp$ is executed.
Let $Y$ be an execution of $\Pnp$ that is $\mtarget$-consistent but not $\msource$-consistent. 
This results in an execution $\exec\supset Y\cup Z$ of $\Pallpsi$ that contains the executions of $\Ppsi$ and $\Pnp$, the read from relation for the reads of $y$ and the required program order additions. 
Since $\Pnp$ has no registers or locations in common with the rest of the program and occurs last in $\porel$, no basic relation of $\exec$ leaves $Y$. 
It follows that $\exec$ contains the same basic paths between events of $Y$ as $Y$.
Thus the violating cycle for $\msource$ in $Y$ is also in $\exec$. 
%There is a violating cycle in $\Pnp$ and 
It follows that $\exec$ is not $\msource$-consistent. 
We show that $\exec$ is still $\mtarget$-consistent. 
We assume towards contradiction that there is a violating cycle for $\mtarget$: 
As before, it holds that a violating cycle must contain an event from $Y$. 
Since $Y$ is never left, any basic cycle of $\exec$ with events in $Y$ must be contained entirely in $Y$. 
It follows from (Lemma~\ref{lem:common}) that a violating cycle for $\mtarget$ must be entirely in $Y$.
This is a contradiction to the execution of $\Pnp$ being $\mtarget$-consistent.
%Since $\Pnp$ is never left and $\exec$ is $\msource$ consistent in $\Ppsi$, it follows that every cycle in $\exec$ violating an assertion of $\msource$ must be contained entirely in $\Pnp$. 

%To be more precise, the cycle is due to a sequence of base relations (Lemma~\ref{lem:common}). 
%These base relations never leave $\Pnp$ (Lemma~\ref{Lemma:Leaving}).
%Hence, the cycle already exists in $\Pnp$.
%The argumentation for irreflexivity is similar.

%It follows that $\Pallpsi$ has an execution that is $\mtarget$ consistent but not $\msource$ consistent iff it has an execution that contains some threads of $\Pnp$. 
%This is the case if there is an execution where $y$ is set to 1 in $\Ppsi$, which means there is an assignment that does not satisfy $\psi$.
%	This is equivalent to the validity of$\forall x_1, \dots, x_m : \psi$.
Since $\Pallpsi$ is a polynomial-time reduction, portability is $\Pi^P_1$-hard. 
\end{proof}

%TODO: add refs for portability definitions
\subsection{$\Pi^P_2$-Completeness of State Portability}

We introduce Lemma~\ref{lem:StateComplete} and \autoref{thm:StateHard} in order to show that state portability is both in $\Pi^P_2$ and is $\Pi^P_2$-hard.
It follows, that state portability is $\Pi^P_2$-complete for common MCMs and thus Theorem~\ref{thm:TraceComplete}.2 is correct.

\begin{lemma}\label{lem:StateComplete}
State-based portability is in $\Pi_2$ for all MCMs.
\end{lemma}

\begin{proof}
We encode the state portability property in a closed formula (i.e. all variables are quantified) of the form $\forall \exists \psi$. 
%A program is state portable, if for any $\mtarget$ consistent execution $\exec_2$, there is an $ \msource$ consistent execution $\exec_1$ and a state \mem\ such that both executions compute \mem. 
We have already shown how to encode consistency of an execution $\exec$ with an MCM $\mm$ as a formula ($X \in \consistent{{\mm}} P$) in~\autoref{sec:encoding}. 
%The execution $\exec$ is given by all variables in the formula. 
Again we encode numbers as sequences of Boolean variables.
Let $val(e)$ be the value that is read/written by a read or write event $e$ and $loc(e)$ the location it accesses.
We can encode the property $\reach X=\sigma$ in a Boolean formula as follows.
If a write has no outgoing $co$ relation, then it must have the same value as the location in the state:
$$\bigwedge_{w\in \stores} (\bigwedge_{w'\in \stores} \neg (w\co w')) \rightarrow val(w)=\state (loc(w)).$$
In a similar way we can ensure that the last operation in the program order on a register $r$ has the value $\state (r)$.
This means we can construct Boolean formulas for properties of the form $X \in \consistent{{\mm}} P$ and $\reach X = \sigma$.
With this, we can construct the following formula:% $\forall X \exists Y : \; \psi$: %\hp{$\sigma$ is uniquely defined by X, remove from the quantification; remove this first quantification and just mentiion the on below; in the article we use $\imp$, be consistent}

 $$\forall X\; \exists Y:\; (X \in \consistent{{\mtarget}} P ) \imp$$
 $$( Y \in \consistent{{\msource}} P \land \reach X  =\reach Y  ).$$

This is equivalent to state portability (see Definition~\ref{def:StatePort}).
%Note that the formula is closed. 
%If $\exec_2$ is $\mtarget$ consistent then $\exec_1$ is $\msource$ consistent and $\state (\exec_1)=\state (\exec_2)=\state $.
The state portability problem from $\msource$ to $\mtarget$ can be expressed as a closed quantified formula of the form $\forall \exists \psi$ and thus state based portability is in $\Pi^P_2$. 
%Any free variable of $\psi$ can be added to the all quantifier without changing its semantics.
\end{proof}

We now introduce the program $P_\psi$ and examine its behavior. We will then use $P_\psi$ in order to prove $\Pi^P_2$-hardness. 

Let $\psi (x_1 \dots x_m)$ be a be a Boolean formula over variables $x_1 \dots x_m$.
The concurrent program $P_\psi:= t_1\parallel t_2$ with two threads $t_1$ and $t_2$ is defined below. 
The program is similar to the program in the previous section. 
It contains additional synchronization in order to ensure that the formula assignment and the computed state match.
%it checks is also computed by the program. 
The program either computes a satisfying assignment of $\psi$ ($y=1$), an unsatisfying assignment ($y=0$) or it ends with an error ($y=2$). 

We use the value 0 to encode the Boolean value false. 
To avoid confusion, we assume that the variables are initialized with some other unused value, e.g. 3. 
This does not interfere with the validity of the proofs since our program only assigns constants and thus 0 and 3 are interchangeable.

We will see that it is sufficient to examine the program under SC (the strongest common MCM) 
where an execution is an interleaving of the two thread executions.
The threads first create some variable assignment $A$; thread $t_1$ assigns 0 to the variables and $t_2$ assigns 1. 
The assignment $A$ is determined by the interleaving of those writes. If the write $x_i\define 0$ of $t_1$ is followed by $x_i\define 1$ of $t_2$
($w_{i,0} \co w_{i,1}$), then $x_i$ is set to 1.
Then $t_1$ ensures that $t_2$ has executed all its writes (so that the assignment doesn't change anymore) by
%all writes have occurred and the assignment read by $t_1$ is the same as the one read by $t_2$\fix{why?}. 
using the synchronization variables $x'_1 \dots x'_m$ to check that $t_1$ and $t_2$ reads the same assignment.
% checking if $t_2$ reads the same assignment as $t_1$ (using synchronization variables $x'_1 \dots x'_m$).
If that is not the case, then some writes of $t_2$ have not occurred yet and $y$ is set to 2 in Line~\ref{line:y=2}. %\fix{"the values of registers in $t_2$ are forwarded to prime registers in $t_1$ using $x'_i$ variables". Use variable z to avoid confusions and p,q,r for different registers}
If all the writes from $t_2$ have occurred, $t_1$ checks whether $A$ satisfies $\psi$ ($y$ is set to 1) or not ($y$ is set to 0).

		\begin{algorithm}\label{thread:t1}
			\NoCaptionOfAlgo
			$r_{c,0} \leftarrow 0;\ r_{c,1}\leftarrow 1;\ r_{c,2} \leftarrow 2$\;  \label{line:rci}
			 $x_1 \define r_{c,0};\ \ldots x_m\define r_{c,0}$ \tcp*{$w_{1,0} \dots w_{m,0}$}
			 $r'_1 \leftarrow x'_1;\ldots r'_m \leftarrow x'_m$\label{line:rx1done} \tcp*{$r'_{1} \dots r'_{m}$}
			 $r_1 \leftarrow x_1;\ldots r_m \leftarrow x_m$ \label{line:rx}\tcp*{$r_{1} \dots r_{m}$}
			\eIf(\tcp*[f]{If $A$ is incomplete}) {$\neg (r_1=r'_1\wedge \dots \wedge r_m=r'_m)$}{\label{line:checkerror}
				$y\define r_{c,2}$ \tcp*{exit with error.} \label{line:y=2}
			}{
			    \eIf(\tcp*[f]{If $A$ satisfies $\psi$,}) {$\psi(r_1 \dots r_m)$}{\label{line:readpsi}
			    		$y\define r_{c,1}$ \tcp*{return 1.}
			    }{
			    		$y\define r_{c,0}$ \tcp*{If it does not, return 0.}
			    }
			} 
			\caption{$t_1$}  
		\end{algorithm}

			\NoCaptionOfAlgo

		\begin{algorithm}\label{thread:t2}
			$r_{c,1} \leftarrow 1;\ r_{c,3} \leftarrow 3$\;
			 $x_1 \define r_{c,1};\ \ldots x_m\define r_{c,1}$ \tcp*{$	w_{1,1} \dots w_{m,1}$}
			 $r_1 \leftarrow x_1;\ldots r_m \leftarrow x_m$ \tcp*{$\bar{r}_1 \dots \bar{r}_m$}\label{line:t2readx}
			 $x'_1 \define r_1;\ \ldots x'_m\define r_m$ \tcp*{$\bar{w}_1 \dots \bar{w}_m$} \label{line:t2writex}
			 $x'_1 \define r_{c,3};\ \ldots x'_m\define r_{c,3}$ 	\label{line:t2overwritex}
			\caption{$t_2$}  
			\RestoreCaptionOfAlgo
		\end{algorithm}

To simplify our study we will define a state only over its locations, not the registers. This does not change the complexity of the state portability problem:
We could simply add instructions to our input programs that write all registers to locations in the end. We can use our simpler notion of states without registers on the input program with the added instructions to solve the original state portability problem with registers.

An assignment $A$ of a set of Boolean variables $V$ is a function $A\subset V\times \{ 0,1 \}$ that assigns either 0 or 1 to a variable. 
%So $A\subset \state $ holds if \mem\ contains assignments for a set of locations $V$ and they are the same as $A$.
\begin{definition}
Given an Assignment $A$ of $x_{1} \dots x_n$, locations $y_1\dots y_m$ and a number $a\in \mathbb{N}$, let 
$\stdef{A}{\ass{y_1 \dots y_n}{a}}$ denote the state with 
\begin{itemize}
	\item $\stdef{A}{\ass{y_1 \dots y_n}{a}}(x_i)=A(x_i)$ for $i\leq n$ and 
	\item $\stdef{A}{\ass{y_1 \dots y_n}{a}}(y_j)=a$ for $j\leq m$,
	\item $\stdef{A}{\ass{y_1 \dots y_n}{a}}(z)=v$ for $z$ any other location and $v$ the initial value (we use 3).
\end{itemize}
\end{definition}

We lift the definition accordingly to $$\stdef{A}{\ass{y_1 \dots y_n}{a};\ass{z_1 \dots z_l}{b}}.$$
The program $P_\psi$ can compute some assignment $A$ with $y=1$ or $y=0$ depending on whether $A$ satisfies $\psi$.

The following lemmas show that $P_\psi$ behaves similar for all common MCMs.

\begin{lemma}\label{lem:progpsiuniproc}
Let $\mm$ contain uniproc and $A$ be an assignment of $\psi$. 
%There is an $\mm$ consistent execution of $P_\psi$ that computes an \mem\ with $A\subset \state $ and $\state (y)=1$ (resp. $\state (y)=0$) 
It holds $A\models \psi$ (resp. $A \not\models \psi$) only if
$$ \exists X \in \consistent{{\mm}}{P_\psi}: \reach X=\stdef{A}{\ass{y}{1}}$$
$$(\text{resp. }\reach X =\stdef{A}{\ass{y}{0}}).$$
%$${} \land {\reach X (y)=1}\;\; \text{(resp. }\reach X (y)=0\text{).}$$
\end{lemma}
\begin{proof}

We show that any $\mm$-consistent execution with $y=0$ or $y=1$ computes a desired state.
Let $\exec$ be an execution that satisfies \uniproc\ and computes some \mem\ with $\state (y)=1$ ($\state (y)=0$ is analogue). 
Since $y\define 1$ is executed in $t_1$, 
it follows that $\psi$ is satisfied by the values of $x_1\dots x_m$ read by $r_1 \dots r_m$ in Line~\ref{line:rx} and also $r'_1, \dots ,r'_m$ in Line~\ref{line:rx1done} read the same values. We call this assignment $A$.
Since the reads of Line~\ref{line:t2overwritex} of $t_2$ occur after the writes $\bar{w}_1 \dots \bar{w}_m$ they are ordered last in $\corel$ according to uniproc. 
The execution computes 3 for $x'_1\dots x'_m$. 

It remains to show that the writes accessed by $r_1\dots r_m$ are indeed computed by $\exec$, meaning they are ordered last in $\corel$.
Towards contradiction, we assume this is not the case.
Then, there is a write $w_{i,1}$ or $w_{i,0}$ that is accessed by a read but its value is not computed (it is not last in $co$).

\paragraph*{Case 1:} Assume that this write is $w_{i,1}$. 
The write is accessed by a read (${w_{i,1}} \rf {r_i}$) and it is not last in the coherence order (${w_{i,1}} \co {w_{i,0}}$). 
According to ${\frrel} \define {\rfrel^{-1};\corel}$, it holds ${r_i} \fr {w_{i,0}}$.
Since $w_{i,0}$ occurs before $r_i$ in $t_1$ and they access the same location, they are related w.r.t $\porel$ and $\slocrel$ and thus $w_{i,0}\poloc r_i\fr w_{i,0}$. 
This cycle is a contradiction to $\exec$ satisfying \uniproc\ which is property (iii) of common MCMs.

\paragraph*{Case 2:} Assume that there is a write $w_{i,0}$ that is read (${w_{i,0}} \rf {r_i}$) and its value is not computed (${w_{i,0}} \co {w_{i,1}}$). It follows that $r_i$ reads the value 0. 
Since $y=1$ is computed, the condition in \autoref{line:readpsi} is not satisfied and since ${val(r_i)} = {val(r'_i)}$, we know that $r'_i$ also reads 0.
This means $r'_i$ reads not the initial value 3 so it must read from $\bar{w}_i$.
For the write $\bar{w}_i$ that $r'_i$ reads from (${\bar{w}_i} \rf {r'_i}$) follows that $val(\bar{w}_i)=val(r'_i)=0$.
According to the data-flow, $\bar{w}_i$ writes the value that was obtained by the previous read $\bar{r}_i$ which must be 0 ($val(\bar{r}_i)=val(\bar{w}_i)=0$). 
Since $\bar{r}$ reads 0 it must read from the only write of 0 ($w_{i,0}\rf \bar{r}_i$). 
From $w_{i,0}\rf \bar{r}_i$ and $w_{i,0}\co w_{i,1}$ follows $\bar{r}_i\fr w_{i,1}$. 
This leads to the cycle $w_{i,1}\poloc \bar{r}_i \fr w_{i,1}$, which is a contradiction to $\exec$ satisfying \uniproc. 

So the writes accessed by $r_1\dots r_m$ are ordered last by $\corel$ and thus the assignment that $\psi$ is checked against is computed by the execution.
It follows that $\stdef{A}{\ass{y}{1}}$ is computed. 
\end{proof}

\begin{lemma}\label{lem:progpsi}
Let $\mm$ be a common MCM and $A$ be an assignment of $\psi$. 
%There is an $\mm$ consistent execution of $P_\psi$ that computes an \mem\ with $A\subset \state $ and $\state (y)=1$ (resp. $\state (y)=0$) 
It holds $A\models \psi$ (resp. $A \not\models \psi$) iff
$$ \exists X \in \consistent{{\mm}}{P_\psi}: \reach X=\stdef{A}{\ass{y}{1}}$$
$$(\text{resp. }\reach X =\stdef{A}{\ass{y}{0}}).$$
%$${} \land {\reach X (y)=1}\;\; \text{(resp. }\reach X (y)=0\text{).}$$
\end{lemma}

\begin{proof}
Given some assignment $A$, we can easily construct an SC-consistent execution where the writes to $x_1 \dots x_n$ are interleaved according to the assignment ($w_{i,0}\co w_{i,1}$ if $x_i$ is satisfied). 
Then $t_2$ reads those values and writes them to $x'_1\dots x'_n$. 
Now $t_1$ reads $x'_1\dots x'_n$ and the if-condition in \autoref{line:checkerror} is not satisfied, we go in the else-branch.
%\hp{confusing, the condition has a negation and I guess you want to say that you go to line 7}
So $t_1$ sets $y$ to $0$ or $1$ depending on whether $A\models \psi$ and $t_2$ sets $x'_1\dots x'_n$ back to the initial value 3. 
It follows that for every assignment $A$, there is a SC consistent execution $\exec$ that computes $\stdef{A}{\ass{y}{0}}$ or $\stdef{A}{\ass{y}{1}}$ depending on whether $A$ satisfies $\psi$.
Since any SC-consistent execution is consistent with all common MCMs, we can compute the desired state for any assignment of $\psi$.
The other direction follows diretly from Lemma~\ref{lem:progpsiuniproc}.
\end{proof}

In order to show $\Pi^P_2$-hardness, we reduce validity of a closed formula $\forall x_1 \dots x_n\exists y_1 \dots y_m: \psi$ to state portability.
The idea is to construct a program that uses $P_\psi$ in order to check if some assignment satisfies $\psi$ and then overwrite $y_1\dots y_m$ with $1$ so that the assignment of $y_1 \dots y_m$ is not given by the computed state.
If $\psi$ was not satisfied, the non-portable component $\Pnp$ is executed. 
If the execution of $\Pnp$ is $\mtarget$-consistent but not $\msource$-consistent, then it pretends that the formula was satisfied by setting $y$ to 1. This means, that under $\mtarget$, any assignment of $x_1 \dots x_n$ with $y=1$ can be computed.
It follows that the program is portable if any assignment of $x_1 \dots x_n$ with $y=1$ can be computed under $\msource$.
Under $\msource$ however, pretending is not possible. 
Here $\Ppsi$ can only be set to $1$ by $\Ppsi$. 
So under $\msource$, an assignment of $x_1 \dots x_n$ and $y=1$ can only be computed if there is some assignment of $y_1 \dots y_m$ so that $\psi$ is satisfied.
The program is portable if $\forall x_1 \dots x_n\exists y_1 \dots y_m: \psi$ holds.
 
We want a simple non portable program that always computes the initial state except under $\mtarget$, where it can set a location $z$ to 1.
We use a program $P_{np}= t'_1 \parallel \dots \parallel t'_k$ with the following properties: 
Any execution consistent with $\msource$ computes $3$ for all its locations and contains no write that sets $z$ to 1.
The $\mtarget$-consistent executions compute either $z=1$ or $z=3$ and $3$ for all other locations. 
%There is a $\mtarget$ consistent execution that computes a state with $z=1$ and $3$ for all other locations.
The program contains only one write to $z$ which is in $t'_1$. 

We assume the state portability problem from $\msource $ to $\mtarget $ is not trivial and a program exists that has an $\mtarget$-consistent execution such that no $\msource$-consistent execution computes the same state \mem. 
We can assume that the program only assigns constant values and has no write on $z$. 

Similarly to the synchronization of $P_\psi$, we can add a mechanism at the end of all threads that does the following:
all threads check if they read \mem; if so, they communicate that to $t_1$ which sets $z$ accordingly and then all threads set all other locations back to $3$. 
It follows that a program $P_{np}$ with the required properties exists.
%We can also ensure that any computation that does not execute all threads of $P_nr$ is $\msource$ consistent.

Given a formula $\forall x_1 \dots x_n\exists y_1 \dots y_m: \psi$, we use $P_\psi= t_1\parallel t_2$ and $P_{np}= t'_1 \parallel \dots \parallel t'_k$ to construct a program $P_s:= t^s_1 \parallel \dots \parallel t^s_k$. We define the threads below. 

\begin{algorithm}
	%\NoCaptionOfAlgo
	 $t_1$ \tcp*{Try some assignment.}
	 $r_y\leftarrow y$\;
	\If(\tcp*[f]{If $\psi$ was not satisfied,})
	{$r_y=0$}{\label{line:readphi}
	     $t'_1$ \tcp*{execute $P_{np}$.}
	     $r_z\leftarrow z$\;\label{line:readPnr}
	    \If(\tcp*[f]{If not $\msource$-consistent,}) {$r_z=1$}{
	         $z\define r_{c,3}$ \tcp*{pretend it is $\msource$-consistent }
	         %$r_y\leftarrow 1$\;
	         $y\define r_{c,1}$ \label{line:pretendsat} \tcp*{and pretend $\psi$ was satisfied.}
	    }
	} 
	    %\If {$r_y=1$}{
	         $y_1\define r_{c,1};\cdots\ y_m\define r_{c,1}$\label{line:overwrite} \tcp*{Overwrite $y_1..y_m$ assignment.}
	    %}
	\caption{$t^s_1$}
	\label{thread:ts1}  
\end{algorithm}
%\RestoreCaptionOfAlgo

Let $t_i:=skip$ for $i\leq 3\leq k$ ($t_1$ and $t_2$ are from $P_\psi$). 
The threads are defined for $2\leq i\leq k$ as
$$t^s_i:=t_i;\ r_y \leftarrow y; \textbf{ if }(r_y=0) \textbf{ then }t'_i.$$

In general terms, $P_s$ does the following: First, it executes $P_\psi$. If the state computed by $P_\psi$ did not satisfy $\psi$ ($y=0$), then it executes $P_{np}$, which is not portable. If the execution of $P_{np}$ is $\mtarget$, but not $\msource$-consistent ($z=1$), then $P_s$ pretends, that the formula was satisfied by setting $y$ to 1. 
Afterward, $y_1\dots y_m$ are set to 1, so that their former assignment checked in $P_\psi$ is no longer given by the computed state.

\begin{lemma}\label{lem:Psmm2}
For every assignment $A$ of $x_1\dots x_n$ holds
  $$\exists X \in \consistent{{\mtarget}} {P_s}:\; \reach X =\stdef{A}{\ass{y,y_1\dots y_m}{1}}.$$
\end{lemma}

\begin{proof}
According to Lemma~\ref{lem:progpsi}, the following holds: 
For every assignment $A$ of $x_1\dots x_n$ and $A'$ of $y_1\dots y_m$, there is an $\mtarget$-consistent execution $\exec$ of $P_\psi$, such that either 
$ \reach {\exec} = \stdef{A\cup A'}{\ass{y}{1}}$ or 
$ \reach {\exec} = \stdef{A\cup A'}{\ass{y}{0}}$. 
We examine both cases: 

\paragraph{Case 1:} If $ \reach {\exec} = \stdef{A\cup A'}{\ass{y}{1}}$, then $A\cup A'\models \psi$ and according to Lemma~\ref{lem:progpsi} there is an SC-consistent execution of $\Ppsi$ that computes $\stdef{A\cup A'}{\ass{y}{1}}$.
We can easily extend this to an SC-consistent execution $\exec'$ (represented by an interleaving) of $P_s$.
After $P_\psi$ is executed, we read 1 with reads $r_y \leftarrow y$ of all threads. So the subsequent if conditions are not satisfied. Then we execute the writes in Line~\ref{line:overwrite}.
%the write $y\define r_{c,0}$ in $t_1$ is not executed and the if condition in Line~\ref{line:readphi} of $t_1^s$ and the if conditions in the other threads $t_2^s \dots t_k^s$ are not satisfied.  
So $y_1\define r_{c,1} \dots y_m\define r_{c,1};$ are the only writes executed outside of $P_\psi$. 
These writes are ordered after the assignments of 0 to $y_1 \dots y_m$ in $\corel$ and thus $\reach{\exec'} (y_i)=1$ for $i\leq m$. 
Since there are no further executed writes outside of $P_\psi$, the computed state otherwise coincides with $ \stdef{A\cup A'}{\ass{y}{1}}$ computed by $\exec$.
The extension of $\exec$ computes $\stdef{A}{\ass{y,y_1\dots y_m}{1}}$. 

\paragraph{Case 2:} If $ \reach {\exec} = \stdef{A\cup A'}{\ass{y}{0}}$, then write $y\define r_{c,0}$ is executed in $t_1$. %and the condition in Line~\ref{line:readphi} is satisfied and thus $t'_1$ is executed. 
We construct an execution $\exec'\supset \exec$ of $P_s$ in the following way:
We ensure all reads $r_y\leftarrow y$ of threads $t^s_i$ with $i\leq k$ read from $y\define r_{c,0}$ and thus all threads $t'_i$ are executed. 
This means $P_{np}$ is executed. 
According to the definition of $P_{np}$, there is an $\mtarget$-consistent execution $Y$ of $\Pnp$ such that $\reach Y (z)=1$ is written by some write $w_z$ in $t'_1$ and $Y$ computes the initial value for all other locations of $\Pnp$.  %and the execution $ {\exec [P_{np}]}$ of $P_{np}$ is $.
We enforce $\exec' \supseteq Y$ and ensure the read $r_z\leftarrow z$ reads from the write in $P_{np}$ ($z=1$) and thus the following if condition is satisfied and the writes $z\define 3$ and $y\define 1$ are executed. We order them last in $co$.
%It is easy to see that $\exec'$ computes the assignment $A$ and $y=1=y_i$ for $i\leq m$.
It follows that $\exec'$ computes $\stdef{A}{\ass{y,y_1\dots y_m}{1}}$

We show that $\exec'$ is still $\mtarget$-consistent. We partition the events of $\exec'$ into sets $\A$, $\B$ and $\C$ and show that there is no violating cycle of $\mtarget$ inside or between these sets. 

Set $\A$ consists of events in $\exec$ and the subsequent reads $r_y\leftarrow y$ of all threads. 
Set $\B$ consists of the events of $Y$.
Set $\C$ contains the events $r_z\leftarrow z$ in $t_1$ and the following writes $z\define r_{c,3}$, $y\define r_{c,1}$ and $y_1\define r_{c,1};\cdots\ y_m\define r_{c,1}$.

%Execution \C\ obviously contains no basic cycles since all its events access different locations.

%According to Lemma~\ref{lem:common}, new cycles are added only if new basic (and $fr$) cycles are added.

The following two conditions hold:
(i) $\B$ is only left by basic relations leading to $\C$. This is the case since $\A$ precedes $\B$ in the program order and has no common registers.
(ii) There are no basic relations leading from $\C$ to some other set: The events in $\C$ are ordered last in the program order and the writes are ordered last in the coherence order. The read $r_z\leftarrow z$ has no outgoing $fr$ relation, since there is only one write to $z$.

From (i) and (ii) follows that $\exec'$ contains no basic cycle that contains events from more than one of the sets. 
According to the path property (Lemma~\ref{lem:common}) exists no violating cycle for $\mtarget$ in $\exec'$ that contains events from more than one of the sets.
So any violating cycle must be contained entirely in one of the sets.

From (i) and (ii) follows (iii): If two events $e_1, e_2$ are in the same set, then there is no basic path from $e_1$ to $e_2$ that leaves the set.

We consider $\A$: 
There are read from relations from $y\define r_{c,0}$ in $t_1$ to all reads $r_y\leftarrow y$ in $t^s_i$ with $i\leq k$. 
However, there are no outgoing basic relations from any of those reads to other events in $\A$. 
It follows from (iii) that there is no basic path from a read $y\define r_{c,0}$ to another event in $\A$ and according to the path property, there is no relation from a read $y\define r_{c,0}$ to some other element in $\A$. It follows that a read $r_y\leftarrow y$ cannot occur in a violating cycle.
From (iii) follows that the basic paths in $\exec'$ between events of $\A$ (and thus the violating cycles for $\mtarget$) are the same as in $\exec$. 
%The same arguments can be applied to show that the basic paths in $\exec$ between events of $\A$ are the same in $\exec [ \Ppsi ]$. 
Since $\exec$ is $\mtarget$-consistent, it has no violating cycle for $\mtarget$ and thus $\exec'$ has no violating cycle for $\mtarget$ in $\A$.

We consider $\B$: 
From (iii) follows that the basic paths in $\exec'$ between events in $\B$ are the same as in $Y$. 
According to Lemma~\ref{lem:paths}, $\exec'$ has a violating cycle for $\mtarget$  in $\B$ iff $Y$ has a violating cycle for $\mtarget$. 
The execution $Y$ is defined to be $\mtarget$-consistent so there is no violating cycle in $\B$.

We consider $\C$: 
The set contains no basic cycle and no basic relation leaves $\C$ according to (ii). It follows that there is no basic cycle that contains events of $\C$.
According to the path property, there is no violating cycle in $\C$.
It follows that $\exec'$ is $\mtarget$-consistent. 

%Let $\state=\reach {\exec'}$.
%The writes in $\C$ are last in the coherence order. This ensures, that $\state (y)=\state (y_i)=1$ for $i\leq m$ and $\state (z)=3$ holds. 
%Since $x_1 \dots x_n$ are not used outside of $P_\psi$, $A\subset \state $ still holds.
%All other locations are set to $3$ in any computed state according to the definition of $P_{np}$.
%It follows, that the modified execution has the desired property.
\end{proof}

\begin{lemma}\label{lem:mssat}
There is an $\msource$-consistent execution of $P_s$ that computes $\stdef{A}{\ass{y,y_1\dots y_m}{1}}$ with some assignment $A$ of $x_1 \dots x_n$  
%a memory state \mem\ with $\state (z)=0;\; \state (y)=\state (y_i)=1$ for $i\leq m$ and 
iff there is an assignment $A'$ of $y_1 \dots y_m$ such that $A\cup A' \models \psi$. 
\end{lemma}
\begin{proof}
If a program $P'$ is contained in a program $P$, we can restrict an execution \exec\ of $P$ to an execution of $P'$. We simply remove all events and relation on events that are not in $P'$. We denote $X$ restricted to $P'$ as $\exec [ P' ]$.

($\Rightarrow$): Let $\exec$ be an $\msource$-consistent execution of $P_s$ that computes $\reach \exec (y)=1$.
We assume towards contradiction that $\exec$ contains a write that sets $z$ to 1.
According to the definition of $\Pnp$ the execution $\exec [ \Pnp ]$ contains a violating cycle for $\msource$.
As with the proof of Lemma~\ref{lem:Psmm2}, we can show that no basic path in $\exec$ between events of $\exec [ \Pnp ]$ leaves $\exec [ \Pnp ]$.
It follows that the basic paths between events in $\exec [ \Pnp ]$ are the same in $\exec$ and $\exec [ \Pnp ]$ and thus 
$\exec$ also contains the violating cycle for $\msource$. This is a contradiction to $\exec$ being $\msource$-consistent.

It follows that $\exec$ cannot read $z=1$ in Line~\ref{line:readPnr} and the write in Line~\ref{line:pretendsat} is not executed.
So $y\define r_{c,1}$ must have been executed in $P_\psi$ for $\exec$ to compute $y=1$.
 
The execution $\exec$ satisfies uniproc since $\msource$ is common. It has no basic cycle that violates uniproc.
It follows that $\exec[{P_\psi}]$ has no such basic cycle either and thus satisfies uniproc. 
%In the proof of Lemma~\ref{lem:progpsi}, we argue that an execution of $\Ppsi$ that satisfies uniproc and computes $y=1$ (resp. y=0) also computes a satisfying (resp. unsatisfying assignment).
It follows from Lemma~\ref{lem:progpsiuniproc}, that $\exec[{P_\psi}]$ computes a satisfying assignment of $x_1\dots y_m$. 
%We can apply Lemma~\ref{lem:progpsi} here to infer that $\exec[{P_\psi}]$ computes a satisfying assignment.
%We know, that $\Ppsi$ only writes 1 to $y$ if the assignment it checks in Line~\ref{line:readpsi} satisfies $\psi$. %(which is later partially overwritten by $t^s_1$)
Since nothing is written to $x_1 \dots x_n$ outside of $P_\psi$, $\exec$ computes the same values for $x_1 \dots x_n$ as $\exec[{P_\psi}]$.
A write $y_i\leftarrow r_{c,0}$ in $t_1$ is related to write $y_i\leftarrow r_{c,1}$ in~\autoref{line:overwrite} of $t^s_1$ in $\porel \cap \slocrel$ and according to uniproc also in $\corel$. This implies that $\exec$ computes 1 for $y_1\dots y_m$.
For the synchronization variables of $\Ppsi$, uniproc ensures that $\exec$ computes the initial values.

It follows, that if there is an $\msource$-consistent execution that computes $\stdef{A}{\ass{y,y_1\dots y_m}{1}}$ for some assignment $A$ of $x_1\dots x_n$, then there is an assignment $A'$ of $y_1,..y_m$ such that $A\cup A' \models \psi$. 

($\Leftarrow$): According to Lemma~\ref{lem:progpsi}, for each assignment $A$ of $x_1 \dots x_n$ and $A'$ of $y_1 \dots y_m$ where $A\cup A'$ satisfies $\psi$, there is an SC-consistent execution $\exec$ of $\Ppsi$ such that $\reach {\exec}=\stdef{A\cup A'}{\ass{y}{1}}$. 
%We construct an execution $Y$ from of $\Pallpsi$ from $\exec$ by adding all reads $r_y \leftarrow y$ with a $\rfrel$ from write $y\define r_{c,1}$ in $\Ppsi$.
%We also add the writes in~\autoref{line:overwrite} and order them last in $\corel$.
This SC execution of $\Ppsi$ can be represented as an interleaving of the local executions. 
To this interleaving, we append the subsequent reads $r_y\leftarrow y$. 
They read the value 1 and the next if conditions are not satisfied. 
Then we append the remaining writes in~\autoref{line:overwrite}. 
The resulting interleaving represents an SC-consistent execution $\exec'$ of $P_s$ .
%Since $y\define 1$ occurred in $t_1$, it follows that $t_1^s$ reads from that write and thus the condition in Line~\ref{line:readphi} is not satisfied. 

The only writes not in $P_\psi$ that are executed by $\exec'$ are in Line~\ref{line:overwrite} 
(they are last in the program order and thus also last in $po\cap sloc$) and according to \uniproc, they must be ordered after any write $y_i\define 0$ in $\corel$. 
It follows $\reach{\exec'} =\stdef{A}{\ass{y,y_1\dots y_m}{1}}$.
%$\reach \exec' (y_i)=1$ for $i\leq m$ as well as $A\subset \reach \exec' $ and $\reach \exec' (y)=1$.
\end{proof}

\begin{theorem}\label{thm:StateHard}
State-based portability is $\Pi_2$-hard for common MCMs.
\end{theorem}
\begin{proof}
Given common MCMs $\msource$ and $\mtarget$, we argue that $P_s$ is a reduction of validity of a formula $\forall x_1 \dots x_n\exists y_1 \dots y_m: \psi(x_1 \dots y_m)$ to state portability from $\msource$ to $\mtarget$

%Any $\mtarget$-consistent execution that executes no events of $P_{np}$ is also $\msource$-consistent. 
%We will omit the detailled proof. The idea is to examine the executions that are not SC consistent. One can see that any such execution violates uniproc.
%As long as no part of $P_{np}$ is executed, $P_s$ is state portable. 

Note that any execution of $P_s$ satisfying uniproc computes 1 for $y_1\dots y_m$ and the initial values for the synchronization variables. 
We show that the following properties (i)-(iii) hold:
\begin{itemize}
	\item[(i)] Any state with $y=2$ that is computable with an $\mtarget$-consistent execution of $P_s$ can be computed by an $\msource$-consistent computation of $P_s$. 
It is easy to see that any state with $y=2$ that can be computed by an execution satisfying uniproc, can be computed by an SC-consistent execution. 
	\item[(ii)]	 Any state with $y=0$ that is computed by an $\mtarget$-consistent execution $\exec$ of $P_s$ can be computed by an $\msource$-consistent computation. 
No event of $\Pnp$ is executed, the inital value is computed for locations of $\Pnp$. 
According to uniproc, any such state has the form $\stdef{A}{\ass{y}{0};\ass{y_1\dots y_m}{1}}$ for some assignment A of $x_1\dots x_n$.
It is easy to see that $A$ and $y=0$ is computed by $\exec[\Ppsi]$ and $\exec[\Ppsi]$ satisfies $\uniproc$. 
According to Lemma~\ref{lem:progpsiuniproc}, there is an assignment $A'$ of $y_1\dots y_m$ such that $A\cup A'\not\models \psi$. 
It follows from Lemma~\ref{lem:progpsi} that there is an SC-consistent execution of $\Ppsi$ that computes $\stdef{A\cup A'}{\ass{y}{0}}$. This can be easily extended to an SC-consistent execution of $P_s$ that computes $\stdef{A}{\ass{y}{0};\ass{y_1\dots y_m}{1}}$.
Since $\msource$ is common, $\stdef{A}{\ass{y}{0};\ass{y_1\dots y_m}{1}}$ is also computable under $\msource$.
	\item[(iii)]  Any state with $y=1$ computed by an $\mtarget$-consistent execution $\exec$ has the form $\stdef{A}{\ass{y,y_1\dots y_m}{1}}$ for some assignment $A$ of $x_1\dots x_n$. This holds since $\exec[\Pnp ]$ is also $\mtarget$-consistent (the basic paths between its events are the same as in $\exec$) and according to the definition of $\Pnp$, it computes the initial value for all its locations except maybe 1 for $z$. If that is the case, then the write $z\leftarrow r_{c,3}$ in $t^s_1$ is executed and according to uniproc ordered later in $\corel$.
\end{itemize}

%These states can be computed under $\mtarget$ according to Lemma~\ref{lem:Psmm2}.

From (i)-(iii) follows that $P_s$ is state portable from $\msource$ to $\mtarget$ iff any $\mtarget$-computable state  $\stdef{A}{\ass{y,y_1\dots y_m}{1}}$ can be computed under $\msource$.
Lemma~\ref{lem:Psmm2} implies that $P_s$ is state portable from $\msource$ to $\mtarget$ iff any state $\stdef{A}{\ass{y,y_1\dots y_m}{1}}$ can be computed under $\msource$.
According to Lemma~\ref{lem:mssat}, this is the case iff for every assignment $A$ of $x_1\dots x_n$ there is an assignment $A'$ of $y_1\dots y_m$ with $A\cup A' \models \psi$.
This is the case iff $\forall x_1 \dots x_n \exists y_1 \dots y_m: \psi (x_1 \dots y_m)$ is valid.
%\fix{rework text here}
%Any state \mem\ with $y=0$ or $y=2$ that can be computed under $\msource$, can also be computed under $\mtarget$. 
%The proof is straightforward (any such state is SC computable) so we omit it here.
%We know from Lemma~\ref{lem:Psmm2}, that for every assignment $A$ of $x_1 \dots x_n$, there is an $\mtarget$-consistent execution that computes $\stdef{A}{\ass{y,y_1\dots y_m}{1}}$
%It follows, that the program $P_s$ is portable if for every assignment $A$ of $x_1 \dots x_n$, there is an $\msource$-consistent execution that computes \mem\ with $A\subset \state $ and $ \state (z)=0;\; \state (y)=\state (y_i)=1$ for $i\leq m$.
%According to Lemma~\ref{lem:mssat}, $P_s$ is portable if 
%It follows, that state portability is $\Pi_2$ hard for all common MCMs that satisfy \uniproc.
\end{proof}

% !TEX spellcheck = en-EN
%!TEX root = sas2017.tex

\section{Complete Experiments}
\label{sec:Aexperiments}

%!TEX root = sas2017.tex
% !TEX spellcheck = en-EN

\begin{table}[h]
\setlength{\tabcolsep}{4.5pt}
\def\sep{\hspace{10pt}}
\centering
\footnotesize
\tt
\scalebox{.9}{
\begin{tabular}{l@{\sep} cccccccccccccccccccc}
 \rm\small Benchmark
& $\rot{70} {SC\text{-}TSO}$
& $\rot{70} {SC\text{-}PSO}$
& $\rot{70} {SC\text{-}Power}$
& $\rot{70} {TSO\text{-}PSO}$
& $\rot{70} {TSO\text{-}Power}$
& $\rot{70} {PSO\text{-}Power}$
& $\rot{70} {RMO\text{-}Alpha}$
& $\rot{70} {RMO\text{-}Power}$
& $\rot{70} {Alpha\text{-}RMO}$
& $\rot{70} {Alpha\text{-}Power}$
\\
\midrule

\rm\bench{Bakery} &
\redcross & \redcross & \redcross & \redcross & \redcross & \redcross & \redcross & \tick & \redcross & \tick \newrow
\rm\bench{Bakery x86} &
\tick & \redcross & \redcross & \redcross & \redcross & \redcross & \redcross & \redcross & \redcross & \redcross \newrow
\rm\bench{Bakery Power} &
\tick & \redcross& \tick & \redcross& \tick & \tick & \redcross& \tick & \redcross& \tick \newrow
\rm\bench{Burns} &
\redcross & \redcross& \redcross& \tick & \redcross& \redcross& \redcross& \tick & \redcross& \tick \newrow
\rm\bench{Burns x86} &
\tick & \tick & \redcross& \tick & \redcross& \redcross& \redcross& \redcross& \redcross& \redcross\newrow
\rm\bench{Burns Power} &
\tick & \tick & \tick & \tick & \tick & \tick & \redcross& \tick & \redcross& \tick \newrow
\rm\bench{Dekker} &
\redcross & \redcross& \redcross& \redcross& \redcross& \redcross& \tick & \tick & \redcross & \tick \newrow
\rm\bench{Dekker x86} &
\tick & \redcross& \redcross& \redcross& \redcross& \redcross& \redcross& \redcross& \redcross& \redcross\newrow
\rm\bench{Dekker Power} &
\tick & \redcross& \tick & \redcross& \tick & \tick & \redcross& \tick & \redcross& \tick \newrow
\rm\bench{Lamport} &
\redcross & \redcross& \redcross& \redcross& \redcross& \redcross& \redcross& \redcross& \redcross& \tick \newrow
\rm\bench{Lamport x86} &
\tick & \tick & \redcross& \tick & \redcross& \redcross& \redcross& \redcross& \tick & \redcross\newrow
\rm\bench{Lamport Power} &
\tick & \tick & \tick & \tick & \tick & \tick & \redcross& \tick & \tick & \tick \newrow
\rm\bench{Parker} &
\redcross & \redcross& \redcross& \redcross& \redcross& \tick & \tick & \tick & \tick & \tick \newrow
\rm\bench{Parker x86} &
\tick & \redcross& \redcross& \redcross& \redcross& \tick & \tick & \tick & \tick & \tick \newrow
\rm\bench{Parker Power} &
\tick & \redcross& \tick & \redcross& \tick & \tick & \tick & \tick & \tick & \tick \newrow
\rm\bench{Peterson} &
\redcross & \redcross& \redcross& \redcross& \redcross& \redcross& \tick & \tick & \redcross& \tick \newrow
\rm\bench{Peterson x86} &
\tick & \redcross& \redcross& \redcross& \redcross& \redcross& \tick & \redcross& \redcross& \redcross\newrow
\rm\bench{Peterson Power} &
\tick & \redcross& \tick & \redcross& \tick & \tick & \tick & \tick & \redcross& \tick \newrow
\rm\bench{Szymanski} &
\redcross & \redcross& \redcross& \tick & \redcross& \redcross& \redcross& \tick & \redcross& \tick \newrow
\rm\bench{Szymanski x86} &
\tick & \tick & \redcross& \tick & \redcross& \redcross& \redcross& \redcross& \redcross& \redcross\newrow
\rm\bench{Szymanski Power} &
\tick & \tick & \tick & \tick & \tick & \tick & \redcross& \tick & \redcross& \tick \newrow
\\
\end{tabular}
}
\rm
\caption{Bounded portability analysis of mutual exclusion algorithms: portable ({\tick}), non-portable ({\redcross})}
\label{tab:mutual}
\end{table}

This section presents the complete set of experiments for portability of mutual exclusion algorithms. Besides the portability analysis between SC, TSO and Power shown on~\autoref{tab:portability}, we report on several MCMs combinations. The results of the portability analysis are shown in~\autoref{tab:mutual}.
The complete set of encoding and solving times are shown respectively in~\autoref{fig:completeencoding} and~\autoref{fig:completesolving}.

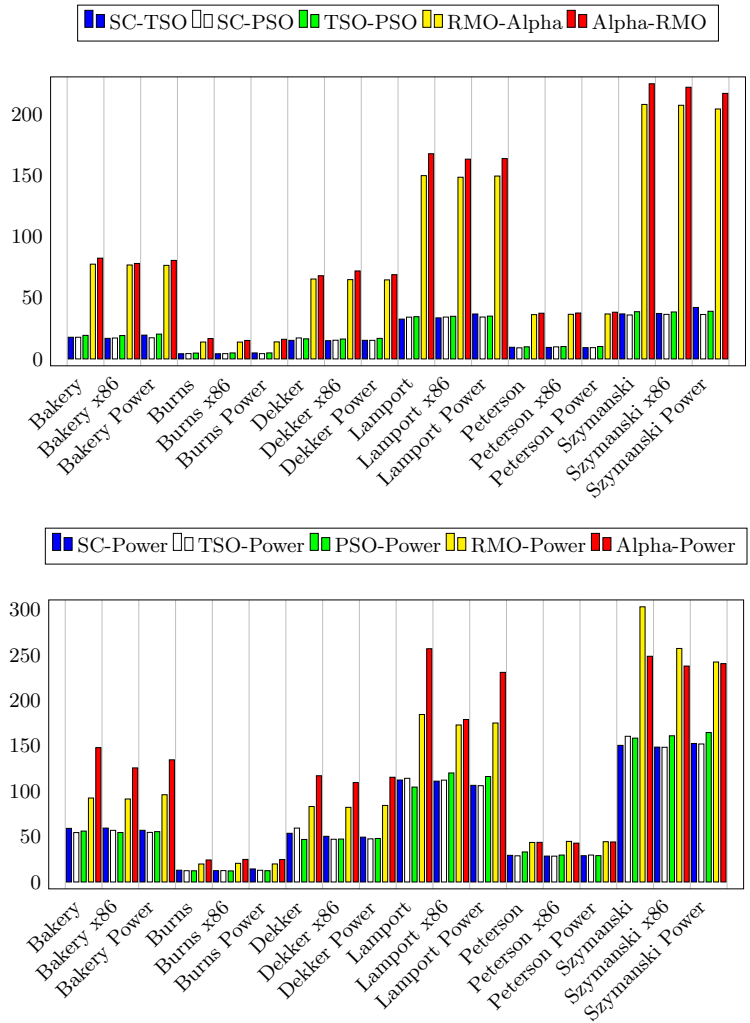
\begin{figure}[h]
\centering
\subfigure{\scalebox{.87}{%!TEX root = ../sas2017.tex

\begin{tikzpicture}
\begin{axis}[
	symbolic x coords={
	Bakery, Bakery x86, Bakery Power, 
	Burns, Burns x86, Burns Power,
	Dekker, Dekker x86, Dekker Power,
	Lamport, Lamport x86,Lamport Power,
%	Parker, Parker x86,Parker Power,
	Peterson, Peterson x86, Peterson Power,
	Szymanski, Szymanski x86, Szymanski Power,
	a
	},
	xtick=data,
	xtick style={draw=none},
	ytick style={draw=none},
	ytick={0,50,100,...,150,200, 250},
	x tick label style={rotate=45, anchor=east, font=\small},
	enlargelimits=0.025,
	legend style={at={(0.5,1.25)},
		anchor=north,legend columns=5,
	legend entries={SC-TSO,SC-PSO, TSO-PSO, RMO-Alpha, Alpha-RMO}, font=\footnotesize},
	ybar interval=0.7,
	width=\textwidth,
	height=6cm
]

\addplot [fill=blue]
	coordinates {
(Bakery, 17.711787939071655)
(Bakery x86, 16.815124988555908)
(Bakery Power, 19.452289819717407)
(Burns, 4.2296130657196045)
(Burns x86, 4.222445011138916)
(Burns Power, 4.9057769775390625)
(Dekker, 15.182631015777588)
(Dekker x86, 14.89327883720398)
(Dekker Power, 15.218568086624146)
(Lamport, 32.53489804267883)
(Lamport x86, 33.54331588745117)
(Lamport Power, 36.726147174835205)
%(parker, 1.7272579669952393)
%(parker x86, 1.8788740634918213)
%(parker Power, 1.84458589553833)
(Peterson, 9.537302017211914)
(Peterson x86, 9.417800903320312)
(Peterson Power, 9.224287033081055)
(Szymanski, 36.77258491516113)
(Szymanski x86, 37.155718088150024)
(Szymanski Power, 42.06637406349182)
(a,0)};
\addplot [fill=white]
	coordinates {
(Bakery, 17.71239709854126)
(Bakery x86, 17.08566689491272)
(Bakery Power, 17.281736850738525)
(Burns, 4.36043381690979)
(Burns x86, 4.3241260051727295)
(Burns Power, 4.316703796386719)
(Dekker, 17.16452407836914)
(Dekker x86, 15.287949085235596)
(Dekker Power, 15.195696830749512)
(Lamport, 34.145390033721924)
(Lamport x86, 34.20997500419617)
(Lamport Power, 34.24053192138672)
%(parker, 1.7635059356689453)
%(parker x86, 1.7368559837341309)
%(parker Power, 1.8228280544281006)
(Peterson, 9.079921007156372)
(Peterson x86, 9.876070022583008)
(Peterson Power, 9.2201669216156)
(Szymanski, 35.92042112350464)
(Szymanski x86, 36.4064040184021)
(Szymanski Power, 36.337804079055786)
(a,0)};
\addplot [fill=green]
	coordinates {
(Bakery, 19.266618013381958)
(Bakery x86, 19.10199284553528)
(Bakery Power, 20.279268980026245)
(Burns, 4.783245086669922)
(Burns x86, 4.902663946151733)
(Burns Power, 4.917539119720459)
(Dekker, 16.422616958618164)
(Dekker x86, 16.261545181274414)
(Dekker Power, 16.7957820892334)
(Lamport, 34.526187896728516)
(Lamport x86, 34.84059500694275)
(Lamport Power, 34.97525691986084)
%(parker, 1.9176299571990967)
%(parker x86, 1.910064935684204)
%(parker Power, 1.9337139129638672)
(Peterson, 9.911800146102905)
(Peterson x86, 10.10535192489624)
(Peterson Power, 10.15679383277893)
(Szymanski, 38.586021900177)
(Szymanski x86, 38.36858797073364)
(Szymanski Power, 38.94428896903992)
(a,0)};
\addplot [fill=yellow]
	coordinates {
(Bakery, 77.43256998062134)
(Bakery x86, 76.76262283325195)
(Bakery Power, 76.4599220752716)
(Burns, 13.814864158630371)
(Burns x86, 13.747427940368652)
(Burns Power, 13.928956985473633)
(Dekker, 65.30407905578613)
(Dekker x86, 64.79374814033508)
(Dekker Power, 64.58525776863098)
(Lamport, 149.76748991012573)
(Lamport x86, 148.46767115592957)
(Lamport Power, 149.43270111083984)
%(parker, 4.638587951660156)
%(parker x86, 5.333816051483154)
%(parker Power, 4.659434080123901)
(Peterson, 36.180014848709106)
(Peterson x86, 36.47496795654297)
(Peterson Power, 36.74676299095154)
(Szymanski, 207.9479649066925)
(Szymanski x86, 207.26525592803955)
(Szymanski Power, 204.21560096740723)
(a,0)};
\addplot [fill=red]
	coordinates {
(Bakery, 82.34192895889282)
(Bakery x86, 78.01531314849854)
(Bakery Power, 80.51177191734314)
(Burns, 16.663074016571045)
(Burns x86, 15.101838111877441)
(Burns Power, 16.046077013015747)
(Dekker, 67.98671388626099)
(Dekker x86, 71.88794994354248)
(Dekker Power, 68.83069491386414)
(Lamport, 167.64864993095398)
(Lamport x86, 163.22713589668274)
(Lamport Power, 163.69334292411804)
%(parker, 5.075859069824219)
%(parker x86, 4.853262186050415)
%(parker Power, 4.959830045700073)
(Peterson, 37.35984206199646)
(Peterson x86, 37.52140998840332)
(Peterson Power, 38.16824007034302)
(Szymanski, 224.83411693572998)
(Szymanski x86, 221.9922640323639)
(Szymanski Power, 216.96360301971436)
(a,0)
};
\end{axis}
\end{tikzpicture}}}
\subfigure{\scalebox{.87}{%!TEX root = ../sas2017.tex

\begin{tikzpicture}
\begin{axis}[
	symbolic x coords={
	Bakery, Bakery x86, Bakery Power, 
	Burns, Burns x86, Burns Power,
	Dekker, Dekker x86, Dekker Power,
	Lamport, Lamport x86,Lamport Power,
%	Parker, Parker x86,Parker Power,
	Peterson, Peterson x86, Peterson Power,
	Szymanski, Szymanski x86, Szymanski Power,
	a
	},
	xtick=data,
	xtick style={draw=none},
	ytick style={draw=none},
	ytick={0,50,100,...,200,250,300},
	x tick label style={rotate=45, anchor=east, font=\small},
	enlargelimits=0.025,
	legend style={at={(0.5,1.25)},
	anchor=north,legend columns=5,
	legend entries={SC-Power,TSO-Power, PSO-Power, RMO-Power, Alpha-Power}, font=\footnotesize},
	ybar interval=0.7,
	width=\textwidth,
	height=6cm
]
\addplot [fill=blue]
	coordinates {
(Bakery, 59.05555701255798)
(Bakery x86, 59.36018395423889)
(Bakery Power, 56.916510820388794)
(Burns, 12.968171119689941)
(Burns x86, 12.452924013137817)
(Burns Power, 14.357208013534546)
(Dekker, 53.6379599571228)
(Dekker x86, 50.312155961990356)
(Dekker Power, 49.373197078704834)
(Lamport, 112.31699204444885)
(Lamport x86, 111.06283402442932)
(Lamport Power, 106.44781613349915)
%(parker, 4.2793591022491455)
%(parker x86, 4.332067012786865)
%(parker Power, 4.3231470584869385)
(Peterson, 29.347244024276733)
(Peterson x86, 28.539752960205078)
(Peterson Power, 28.929024934768677)
(Szymanski, 150.4757010936737)
(Szymanski x86, 148.56943702697754)
(Szymanski Power, 152.56696391105652)
(a,0)};
\addplot [fill=white]
	coordinates {
(Bakery, 54.40743398666382)
(Bakery x86, 56.86174511909485)
(Bakery Power, 54.626919984817505)
(Burns, 12.412890195846558)
(Burns x86, 12.52488398551941)
(Burns Power, 12.91234803199768)
(Dekker, 59.376790046691895)
(Dekker x86, 47.08935308456421)
(Dekker Power, 47.49696111679077)
(Lamport, 114.19770884513855)
(Lamport x86, 112.17199516296387)
(Lamport Power, 106.07640600204468)
%(parker, 4.84516716003418)
%(parker x86, 4.509006023406982)
%(parker Power, 4.536892890930176)
(Peterson, 28.8381130695343)
(Peterson x86, 28.47548198699951)
(Peterson Power, 29.65829110145569)
(Szymanski, 160.4891541004181)
(Szymanski x86, 148.37688899040222)
(Szymanski Power, 152.01106595993042)
(a,0)};
\addplot [fill=green]
	coordinates {
(Bakery, 56.07311487197876)
(Bakery x86, 54.47545790672302)
(Bakery Power, 55.44184398651123)
(Burns, 12.326245069503784)
(Burns x86, 12.19865608215332)
(Burns Power, 12.438342809677124)
(Dekker, 46.836266040802)
(Dekker x86, 47.2587411403656)
(Dekker Power, 47.898717164993286)
(Lamport, 104.5262770652771)
(Lamport x86, 120.09160900115967)
(Lamport Power, 116.18441104888916)
%(parker, 4.524386167526245)
%(parker x86, 4.84607720375061)
%(parker Power, 4.568542003631592)
(Peterson, 33.07738900184631)
(Peterson x86, 29.65832495689392)
(Peterson Power, 29.01974582672119)
(Szymanski, 158.44853115081787)
(Szymanski x86, 161.03548502922058)
(Szymanski Power, 164.60874891281128)
(a,0)};
\addplot [fill=yellow]
	coordinates {
(Bakery, 92.47750806808472)
(Bakery x86, 91.41643404960632)
(Bakery Power, 96.11202907562256)
(Burns, 19.80647087097168)
(Burns x86, 20.556342840194702)
(Burns Power, 19.872626066207886)
(Dekker, 83.11427998542786)
(Dekker x86, 82.2297899723053)
(Dekker Power, 84.3297049999237)
(Lamport, 184.48494505882263)
(Lamport x86, 172.94712710380554)
(Lamport Power, 175.1193699836731)
%(parker, 6.245426893234253)
%(parker x86, 5.978340148925781)
%(parker Power, 6.043774843215942)
(Peterson, 43.555743932724)
(Peterson x86, 44.71226382255554)
(Peterson Power, 44.35728716850281)
(Szymanski, 303.1020109653473)
(Szymanski x86, 257.22806692123413)
(Szymanski Power, 242.29167008399963)
(a,0)};
\addplot [fill=red]
	coordinates {
(Bakery, 147.9792459011078)
(Bakery x86, 125.65532207489014)
(Bakery Power, 134.53474402427673)
(Burns, 24.237879037857056)
(Burns x86, 24.8580539226532)
(Burns Power, 24.7087562084198)
(Dekker, 117.00957894325256)
(Dekker x86, 109.46362400054932)
(Dekker Power, 115.35802793502808)
(Lamport, 256.9212529659271)
(Lamport x86, 178.96000790596008)
(Lamport Power, 230.89001393318176)
%(parker, 6.072597026824951)
%(parker x86, 6.472984075546265)
%(parker Power, 6.23146915435791)
(Peterson, 43.66850209236145)
(Peterson x86, 42.7906060218811)
(Peterson Power, 44.07647895812988)
(Szymanski, 248.63054990768433)
(Szymanski x86, 237.85160207748413)
(Szymanski Power, 240.5077929496765)
(a,0)};
\end{axis}
\end{tikzpicture}}}
\caption{Encoding times (in secs.) for portability of mutual exclusion algorithms.}
\label{fig:completeencoding}
\end{figure}

\begin{figure}[h]
\centering
\subfigure{\scalebox{.87}{%!TEX root = ../sas2017.tex

\begin{tikzpicture}
\begin{axis}[
	symbolic x coords={
	Bakery, Bakery x86, Bakery Power, 
	Burns, Burns x86, Burns Power,
	Dekker, Dekker x86, Dekker Power,
	Lamport, Lamport x86,Lamport Power,
%	Parker, Parker x86,Parker Power,
	Peterson, Peterson x86, Peterson Power,
	Szymanski, Szymanski x86, Szymanski Power,
	a
	},
	xtick=data,
	xtick style={draw=none},
	ytick style={draw=none},
	ytick={0,0.5,1,...,1.5,2,2.5},
	x tick label style={rotate=45, anchor=east, font=\small},
	enlargelimits=0.025,
	legend style={at={(0.5,1.25)},
		anchor=north,legend columns=5,
	legend entries={SC-TSO,SC-PSO, TSO-PSO, RMO-Alpha, Alpha-RMO}, font=\footnotesize},
	ybar interval=0.7,
	width=\textwidth,
	height=6cm
]

\addplot [fill=blue]
	coordinates {
(Bakery, 0.21112990379333496)
(Bakery x86, 0.2102651596069336)
(Bakery Power, 0.21029901504516602)
(Burns, 0.05124306678771973)
(Burns x86, 0.04942202568054199)
(Burns Power, 0.050054073333740234)
(Dekker, 0.16898703575134277)
(Dekker x86, 0.16709613800048828)
(Dekker Power, 0.1702289581298828)
(Lamport, 0.4890270233154297)
(Lamport x86, 0.759335994720459)
(Lamport Power, 0.7899990081787109)
%(Parker, 0.01889801025390625)
%(Parker Power, 0.019166946411132812)
%(Parker x86, 0.018632173538208008)
(Peterson, 0.10684895515441895)
(Peterson x86, 0.10451698303222656)
(Peterson Power, 0.10692906379699707)
(Szymanski, 0.5494050979614258)
(Szymanski x86, 0.5911979675292969)
(Szymanski Power, 0.6410701274871826)
(a, 0)
};
\addplot [fill=white]
	coordinates {
(Bakery, 0.21400117874145508)
(Bakery x86, 0.21073293685913086)
(Bakery Power, 0.21464896202087402)
(Burns, 0.04929804801940918)
(Burns x86, 0.050644874572753906)
(Burns Power, 0.05086708068847656)
(Dekker, 0.16834092140197754)
(Dekker x86, 0.1729729175567627)
(Dekker Power, 0.16947293281555176)
(Lamport, 0.4916388988494873)
(Lamport x86, 0.7281951904296875)
(Lamport Power, 0.8419220447540283)
%(Parker, 0.01861405372619629)
%(Parker Power, 0.01949000358581543)
%(Parker x86, 0.018869876861572266)
(Peterson, 0.10475778579711914)
(Peterson x86, 0.10595202445983887)
(Peterson Power, 0.10976791381835938)
(Szymanski, 0.5311329364776611)
(Szymanski x86, 0.5857129096984863)
(Szymanski Power, 0.6311180591583252)
(a, 0)
};
\addplot [fill=green]
	coordinates {
(Bakery, 0.2190721035003662)
(Bakery x86, 0.2321910858154297)
(Bakery Power, 0.2348921298980713)
(Burns, 0.053073883056640625)
(Burns x86, 0.055757999420166016)
(Burns Power, 0.05456399917602539)
(Dekker, 0.20906805992126465)
(Dekker x86, 0.2102038860321045)
(Dekker Power, 0.21227598190307617)
(Lamport, 0.5205211639404297)
(Lamport x86, 0.9647560119628906)
(Lamport Power, 1.1647891998291016)
%(Parker, 0.02191305160522461)
%(Parker Power, 0.02324700355529785)
%(Parker x86, 0.022266864776611328)
(Peterson, 0.11618685722351074)
(Peterson x86, 0.11365413665771484)
(Peterson Power, 0.1182708740234375)
(Szymanski, 0.6846680641174316)
(Szymanski x86, 0.7036402225494385)
(Szymanski Power, 0.7138059139251709)
(a, 0)
};
\addplot [fill=yellow]
	coordinates {
(Bakery, 0.5961809158325195)
(Bakery x86, 0.603308916091919)
(Bakery Power, 0.6050448417663574)
(Burns, 0.10297894477844238)
(Burns x86, 0.11391115188598633)
(Burns Power, 0.10470199584960938)
(Dekker, 0.5138731002807617)
(Dekker x86, 0.5196599960327148)
(Dekker Power, 0.56563401222229)
(Lamport, 1.322260856628418)
(Lamport x86, 1.3322899341583252)
(Lamport Power, 1.3607261180877686)
%(Parker, 0.030733108520507812)
%(Parker Power, 0.031781911849975586)
%(Parker x86, 0.0314328670501709)
(Peterson, 0.26342082023620605)
(Peterson x86, 0.26907801628112793)
(Peterson Power, 0.2674829959869385)
(Szymanski, 1.563870906829834)
(Szymanski x86, 1.6607310771942139)
(Szymanski Power, 1.5786058902740479)
(a, 0)
};
\addplot [fill=red]
	coordinates {
(Bakery, 0.600074052810669)
(Bakery x86, 0.6026499271392822)
(Bakery Power, 0.606834888458252)
(Burns, 0.10565304756164551)
(Burns x86, 0.10723114013671875)
(Burns Power, 0.10496401786804199)
(Dekker, 0.5253000259399414)
(Dekker x86, 0.5312240123748779)
(Dekker Power, 0.5235648155212402)
(Lamport, 1.341567039489746)
(Lamport x86, 1.7987451553344727)
(Lamport Power, 1.8820040225982666)
%(Parker, 0.030813217163085938)
%(Parker Power, 0.0325009822845459)
%(Parker x86, 0.031785011291503906)
(Peterson, 0.2712550163269043)
(Peterson x86, 0.2696409225463867)
(Peterson Power, 0.2713041305541992)
(Szymanski, 1.5394089221954346)
(Szymanski x86, 1.544600009918213)
(Szymanski Power, 1.6145939826965332)
(a, 0)
};
\end{axis}
\end{tikzpicture}}}
\subfigure{\scalebox{.87}{%!TEX root = ../sas2017.tex

\begin{tikzpicture}
\begin{axis}[
	symbolic x coords={
	Bakery, Bakery x86, Bakery Power, 
	Burns, Burns x86, Burns Power,
	Dekker, Dekker x86, Dekker Power,
	Lamport, Lamport x86,Lamport Power,
%	Parker, Parker x86,Parker Power,
	Peterson, Peterson x86, Peterson Power,
	Szymanski, Szymanski x86, Szymanski Power,
	a
	},
	xtick=data,
	xtick style={draw=none},
	ytick style={draw=none},
	ytick={0,5,10,...,35,40},
	x tick label style={rotate=45, anchor=east, font=\small},
	enlargelimits=0.025,
	legend style={at={(0.5,1.25)},
		anchor=north,legend columns=5,
	legend entries={SC-Power,TSO-Power, PSO-Power, RMO-Power, Alpha-Power}, font=\footnotesize},
	ybar interval=0.7,
	width=\textwidth,
	height=6cm
]

\addplot [fill=blue]
	coordinates {
(Bakery, 1.718883991241455)
(Bakery x86, 2.0373098850250244)
(Bakery Power, 2.0865390300750732)
(Burns, 0.2777390480041504)
(Burns x86, 0.3118760585784912)
(Burns Power, 0.3193979263305664)
(Dekker, 1.4639639854431152)
(Dekker x86, 1.7945969104766846)
(Dekker Power, 1.7987449169158936)
(Lamport, 4.230733871459961)
(Lamport x86, 5.859885931015015)
(Lamport Power, 22.798483848571777)
%(Parker, 0.07511496543884277)
%(Parker Power, 0.08404111862182617)
%(Parker x86, 0.08121991157531738)
(Peterson, 0.5806450843811035)
(Peterson x86, 0.74149489402771)
(Peterson Power, 0.7837228775024414)
(Szymanski, 11.447784185409546)
(Szymanski x86, 10.17952585220337)
(Szymanski Power, 14.220797061920166)
(a, 0)
};
\addplot [fill=white]
	coordinates {
(Bakery, 1.669646978378296)
(Bakery x86, 2.095571994781494)
(Bakery Power, 2.1380279064178467)
(Burns, 0.2794959545135498)
(Burns x86, 0.3151819705963135)
(Burns Power, 0.328110933303833)
(Dekker, 1.4670288562774658)
(Dekker x86, 1.8313229084014893)
(Dekker Power, 1.8584771156311035)
(Lamport, 4.643279790878296)
(Lamport x86, 6.377609968185425)
(Lamport Power, 33.291738986968994)
%(Parker, 0.076995849609375)
%(Parker Power, 0.08423495292663574)
%(Parker x86, 0.0838158130645752)
(Peterson, 0.6031379699707031)
(Peterson x86, 0.7748010158538818)
(Peterson Power, 0.8080220222473145)
(Szymanski, 11.43646502494812)
(Szymanski x86, 9.99736499786377)
(Szymanski Power, 16.323559999465942)
(a, 0)
};
\addplot [fill=green]
	coordinates {
(Bakery, 1.684959888458252)
(Bakery x86, 2.0256781578063965)
(Bakery Power, 2.142043113708496)
(Burns, 0.28162097930908203)
(Burns x86, 0.3223989009857178)
(Burns Power, 0.3321700096130371)
(Dekker, 1.4776170253753662)
(Dekker x86, 1.8423659801483154)
(Dekker Power, 1.8567891120910645)
(Lamport, 4.935337781906128)
(Lamport x86, 6.459616184234619)
(Lamport Power, 33.15235900878906)
%(Parker, 0.07509803771972656)
%(Parker Power, 0.08389806747436523)
%(Parker x86, 0.08329415321350098)
(Peterson, 0.5931689739227295)
(Peterson x86, 0.7798309326171875)
(Peterson Power, 0.8026549816131592)
(Szymanski, 11.339897155761719)
(Szymanski x86, 9.59848403930664)
(Szymanski Power, 15.555482864379883)
(a, 0)
};
\addplot [fill=yellow]
	coordinates {
(Bakery, 1.9555270671844482)
(Bakery x86, 2.542707920074463)
(Bakery Power, 2.4902281761169434)
(Burns, 0.3144710063934326)
(Burns x86, 0.3866710662841797)
(Burns Power, 0.3856320381164551)
(Dekker, 1.771177053451538)
(Dekker x86, 2.242156982421875)
(Dekker Power, 2.1123149394989014)
(Lamport, 5.6347129344940186)
(Lamport x86, 6.703173875808716)
(Lamport Power, 31.81052017211914)
%(Parker, 0.08194994926452637)
%(Parker Power, 0.09115219116210938)
%(Parker x86, 0.09536504745483398)
(Peterson, 0.7770531177520752)
(Peterson x86, 0.9919090270996094)
(Peterson Power, 0.9977619647979736)
(Szymanski, 10.13765001296997)
(Szymanski x86, 10.490145921707153)
(Szymanski Power, 14.136280059814453)
(a, 0)
};
\addplot [fill=red]
	coordinates {
(Bakery, 2.0100860595703125)
(Bakery x86, 2.4069149494171143)
(Bakery Power, 2.4795260429382324)
(Burns, 0.3532569408416748)
(Burns x86, 0.39566493034362793)
(Burns Power, 0.38312482833862305)
(Dekker, 1.7838408946990967)
(Dekker x86, 2.158390998840332)
(Dekker Power, 2.09065580368042)
(Lamport, 7.027853012084961)
(Lamport x86, 6.6686530113220215)
(Lamport Power, 35.319161891937256)
%(Parker, 0.0815129280090332)
%(Parker Power, 0.08915019035339355)
%(Parker x86, 0.08945894241333008)
(Peterson, 0.7625889778137207)
(Peterson x86, 0.9955828189849854)
(Peterson Power, 1.0100979804992676)
(Szymanski, 10.348663091659546)
(Szymanski x86, 10.590330839157104)
(Szymanski Power, 18.54880404472351)
(a, 0)
};
\end{axis}
\end{tikzpicture}}}
\caption{Solving times (in secs.) for portability of mutual exclusion algorithms.}
\label{fig:completesolving}
\end{figure}
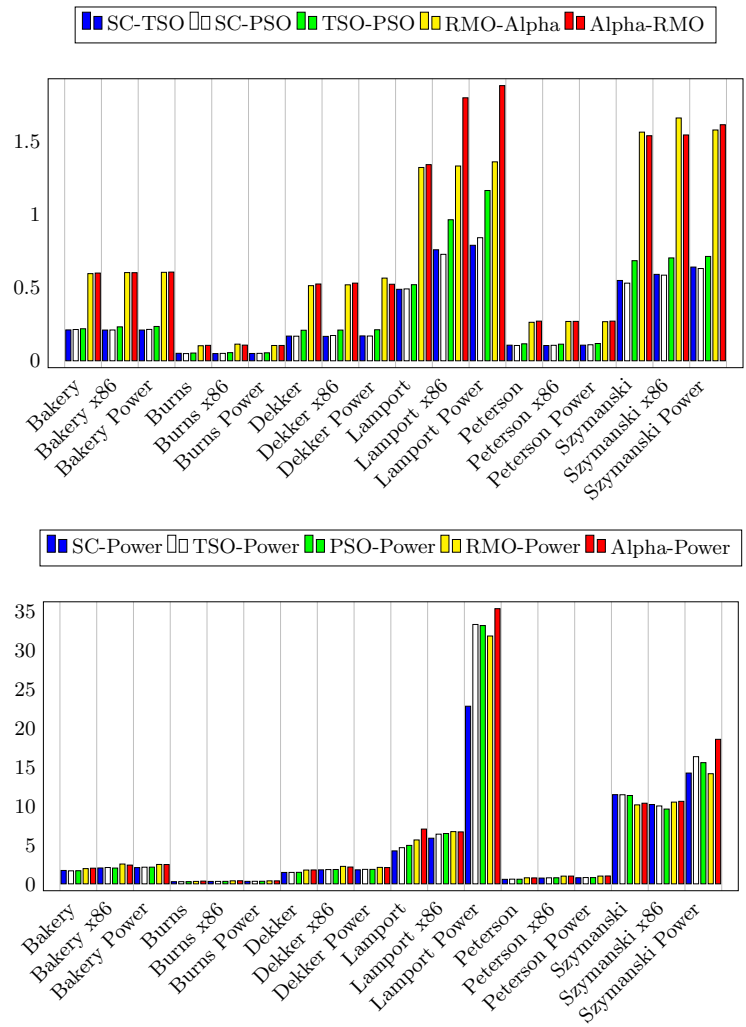

%what is portability take models + programs as input
%what is portable 
%
%new paragraph: bounded model checking+ reduce to smt:
%
%two solution (shorten) alternation, lfp in smt
%remove comparison state portability "`interestingly...
%
%move independent interest or fixed point in introduction (+conclussion)
%
%kleene sentence+reference
%
%false positives-> second reachability query
%
%add sentence experiments
%rfe e?
%ref appendix?
%!TEX root = sas2017.tex

\section{Common Executions}
\label{sec:common}
%\flo{question: why is $po$ part of the execution definition?}
For a given program $P$, let $\execsof{P}$ be the set of its (consistent and inconsistent) executions. We show how to check portability of a high-level program $P_H$ based on the executions of two different low-level programs $P_S$ and $P_T$ that were compiled  from $P_H$ to different architectures.
%\roland{This is the explanation of why the definition below has to be so complicated.}
The following concept of high-level portability is more involved than~\autoref{def:traceportability} since 
programs compiled towards different architectures may differ greatly and it is often difficult to directly compare their executions. 
The definition relies on a formula $\execproj{\exec_L,\exec_H}$ which holds iff the high-level execution $\exec_H$ is the projection (see below) of the low-level one $\exec_L$ according to the compilation mapping. 
We consider executions of $P_S$ and $P_T$ to be similar if they have the same projection over $P_H$.
We adapt our method. Instead of looking for a single low-level execution that is consistent with $\mtarget$ but not $\msource$, we now look for a high-level execution that is the projection of two low-level executions; an execution of $P_S$ not consistent with $\msource$ and an execution of $P_T$ consistent with $\mtarget$.
In order to still be able to encode portability as an existential SMT query, formula $\execproj{\exec_L,\exec_H}$ has to be existential as well.
% -----------------------------------------------------------------------------------
% -----------------------------------------------------------------------------------
% -----------------------------------------------------------------------------------
% -----------------------------------------------------------------------------------
\begin{definition}[High-level Portability]\label{def:2progport}
Let $\msource$, $\mtarget$ be two MCMs, $P_H$ a high-level program, and $P_S$, $P_T$ the two low-level programs after compiling to the corresponding architectures. Program $P_H$ is \emph{not portable from $\msource$ to $\mtarget$} if there are executions $\exec_H\in \execsof{P_H}, \exec_S\in \execsof{P_S}$ and $\exec_T\in \execsof{P_T}$ such that
$$\execproj{\exec_S,\exec_H} \land \execproj{\exec_T,\exec_H} \land \exec_S \not\in \consistent{{\msource}} {P_S} \land \exec_T  \in \consistent{{\mtarget}} {P_T}.$$
\end{definition}
% -----------------------------------------------------------------------------------
% -----------------------------------------------------------------------------------
% -----------------------------------------------------------------------------------
% -----------------------------------------------------------------------------------
For compilers doing complex optimisations such as common subexpression elimination, speculative execution, etc., creating the formula $\execproj{\exec_L,\exec_H}$ is not trivial, but can still be done. Wickerson et al already study the construction of projected high-level executions for litmus tests \cite{memalloy}.
We leave the concrete details of the implementation of such a mapping for, e.g., LLVM-compiler-generated x86/Power assembly code, for future research. 
However, we claim this can be done with a reasonable amount of implementation effort. 
For simpler compilers where we can obtain a function from low-level instructions to high-level ones %\hp{what do we have to ask the mapping? one-to-one?}
 (such function exists even in the presence of instructions reordering), we give below a concrete formula $\execproj{\exec_L,\exec_H}$. 
 This relation between executions uses a direct mapping between events; it does not depend on the program order which allows us to handle any instruction reorderings without difficulties.

Given a program $P$, we denote its set of instructions as $I_P$. A low-level program $P_L$ is obtained by compiling an unrolled high-level program $P_H$. %we associate an execution $\exec_P$ of $P$ with an execution $\exec_H$ of the high-level program.
Each memory event of an execution $\exec_L$ corresponds to a low-level instruction, which in turn was compiled from a high-level instruction. We define a function $\hlcom : \events_L \rightarrow I_H$ that assigns to each memory event in the low-level execution the corresponding high-level instruction. %It partitions the events of $P$ according to the high-level instructions they were compiled from. 
%Here, each execution $\exec_P$ of $P$ induces an execution $\hlexec_{P,H}{\exec }$ of the original high-level program $H$ in the following way:
We use $\hlcom$ to relate executions of the low-level and high-level programs.
An execution $\exec_H$ of a high-level program $P_H$ is a set of executed instructions $\executedi_H \subseteq I_H $ and relations $\rfrel$ and $\corel$ between them.
\begin{definition}[Execution Projections]\label{def:hlexec}
Let $\exec_L, \exec_H$ be respectively executions of a low-level program $P_L$ and a high-level program $P_H$ and let $\hlcom: \events_L \rightarrow I_H$ be a function. We define
\begin{align}
\execproj{\exec_L,\exec_H} \define & \bigwedge_{e \in \executed_L} \hlcom (e)\in \executedi_H \land \label{eq:proj1}\\
& \bigwedge_{ e_1, e_2\in \executed_L \atop {\rel \in \{ \rf , \co \} }} (e_1 \rel e_2 \Rightarrow \hlcom(e_1) \rel \hlcom(e_2)) \land \label{eq:proj2}\\
& \bigwedge_{{ {i_1, i_2} \in {\executedi_H}}} (i_1 \rel i_2 \Rightarrow \label{eq:proj3}\\
& \;\;\;\;\;\; \bigvee_{ {e_1,e_2} \in {\executed_L}} (e_1 \rel e_2 \land i_1=\hlcom (e_1)\land i_2=\hlcom(e_2)))\label{eq:proj4} \nonumber
\end{align}
%$\rel \in \{ \rf , \co \}:$
\end{definition}
We ensure that each executed event is mapped to a unique executed high-level instruction (\ref{eq:proj1}), that events related with $\rfrel$ or $\corel$ are mapped to instructions that are related by the same relation (\ref{eq:proj2}), and that if two instructions are related, then there are two corresponding events in the low-level program with the same relation (\ref{eq:proj3}).
Note that one high-level instruction may correspond to multiple memory events.%: i++ results in events $reg\leftarrow i; reg\leftarrow reg+1; i\define reg$. Thus, a relation $c_1 \rel c_2$ does not implicate $e_1 \rel e_2$ for all events $e_1,e_2$ that originate from $c_1,c_2$.

%This gives us the following notion of equivalence: $$\equivexec(\exec_1, \exec_2) \Leftrightarrow \hlexec(\exec_1 )=\hlexec(\exec_2 )$$

The function $\hlcom$ can be easily constructed by keeping track of the original instructions during the compilation. 
\autoref{def:hlexec} can be encoded in SAT and thus portability of a high-level program can be encoded in an existential SMT formula according to Definition~\ref{def:2progport} and \ref{def:hlexec}.

%equiv compiler mapping reorders -> not relate low level instructions: instead relate to high-level then say equiv iff common high-level
%app D contains proposal for some/simple comp opt. to do the full is still research /future work
%
%C.3 encoding times add text figures-> hernan
%
%the third problem ammend to sec D

%\fi
\end{document}